\documentclass[11pt]{amsart}

\usepackage{amsmath}
\usepackage{amsfonts}
\usepackage{amssymb}
\usepackage{amsthm}
\usepackage{mathptmx}
\usepackage{mathrsfs}
\usepackage{latexsym}
\usepackage{times}
\usepackage[mathscr]{euscript}
\usepackage[isolatin]{inputenc}

\usepackage[pagebackref]{hyperref}
\usepackage[backrefs,alphabetic,initials]{amsrefs}

\DeclareMathAlphabet{\mathpzc}{OT1}{pzc}{m}{it}

\newtheorem{thm}{Theorem}[section]
\newtheorem{lem}[thm]{Lemma}
\newtheorem{prop}[thm]{Proposition}
\newtheorem{cor}[thm]{Corollary}

\theoremstyle{definition}
\newtheorem{defn}[thm]{Definition}
\newtheorem{ex}[thm]{Example}
\newtheorem*{aknow}{Acknowledgments}

\theoremstyle{remark}
\newtheorem{rem}[thm]{Remark}

\usepackage{color}

\newcommand{\spk}{\mathfrak{sp}}
\newcommand{\gl}{\mathfrak{gl}}
\newcommand{\slk}{\mathfrak{sl}}

\newcommand{\ok}{\mathfrak{o}}
\newcommand{\ospk}{\mathfrak{osp}}
\newcommand{\zk}{\mathfrak{z}}

\newcommand{\g}{\mathfrak{g}}

\newcommand{\hk}{\mathfrak{h}}
\newcommand{\lk}{\mathfrak{l}}

\newcommand{\qk}{\mathfrak{q}}
\newcommand{\tk}{\mathfrak{t}}

\newcommand{\jk}{\mathfrak{j}}
\newcommand{\uk}{\mathfrak{u}}

\newcommand\Bc{\mathcal{B}}
\newcommand\Lc{\mathcal L}
\newcommand\Dc{\mathcal D}

\newcommand\Ec{\mathcal E}
\newcommand\Ic{\mathcal I}
\newcommand\Jc{\mathcal J}
\newcommand\Pc{\mathcal P}

\newcommand\Oc{\mathcal O}

\newcommand{\Sb}{\mathsf{S}}

\newcommand{\Zb}{\mathsf{Z}}

\newcommand\CC{\mathbb C}
\newcommand\NN{\mathbb N}
\newcommand\ZZ{\mathbb Z}

\newcommand\PP{\mathbb P}

\newcommand{\Cs}{\mathscr C}

\newcommand{\Ds}{\mathscr D}
\newcommand{\Es}{\mathscr E}

\newcommand{\Rs}{\mathscr R}
\newcommand{\Ss}{\mathscr S}

\newcommand{\Ns}{\mathscr N}

\newcommand{\Vs}{\mathscr V}
\newcommand{\Ws}{\mathscr W}
\newcommand{\Zs}{\mathscr Z}

\newcommand{\io}{\operatorname{\iota}}

\newcommand{\ad}{\operatorname{ad}}
\newcommand{\Ad}{\operatorname{Ad}}

\newcommand{\im}{\operatorname{Im}}

\renewcommand\dfrac{\displaystyle \frac}

\newcommand{\End}{\operatorname{End}}

\newcommand{\sym}{\operatorname{Sym}}
\newcommand{\Id}{\operatorname{Id}}
\newcommand{\dup}{\operatorname{dup}}
\renewcommand\hat\widehat
\renewcommand\tilde\widetilde 
\newcommand{\spa}{\operatorname{span}}
\newcommand{\GL}{\operatorname{GL}}
\newcommand{\OO}{\operatorname{O}}
\newcommand{\Sp}{\operatorname{Sp}}
\newcommand{\SO}{\operatorname{SO}}

\newcommand\adpo{\ad_{\tt P}}

\newcommand{\degz}{\operatorname{deg}_{\ZZ}}

\newcommand{\zero}{\overline{0}}
\newcommand{\um}{\overline{1}}

\newcommand{\ze}{{\scriptscriptstyle{\overline{0}}}}
\newcommand{\un}{{\scriptscriptstyle{\overline{1}}}}

\newcommand{\gO}{\g_{\ze}}
\newcommand{\gI}{\g_{\un}}
\newcommand{\gd}{\g^*}

\newcommand{\oplusp}{{ \ \overset{\perp}{\mathop{\oplus}} \ }}

\newcommand{\iiso}{\overset{\mathrm{i}}{\simeq}}
\newcommand{\cb}{{\overline{C}}}
\newcommand{\cpb}{{\overline{C'}}}

\newcommand{\diag}{\operatorname{diag}}

\newcommand{\Nsh}{\widehat{{\Ns}}}

\newcommand{\Dsh}{\widehat{{\Ds}}}

\newcommand{\Si}{\Ss_{\mathrm{inv}}}
\newcommand{\Sih}{\widehat{\Si}}

\newcommand{\Ich}{\widetilde{{\Ic}}}
\newcommand{\ps}{\PP^1}

\newcommand{\ta}{{ \ \mathop{\times}\limits_{\mathrm{a}} \ }}

\newcommand{\Alt}{\operatorname{Alt}}

\begin{document}

\title{Singular quadratic Lie superalgebras}

\author{Minh Thanh Duong, Rosane Ushirobira}

\address{Minh Thanh Duong, Department of Physics, University of
  Education of Ho Chi Minh city, 280 An Duong Vuong, Ho Chi Minh city,
  Vietnam.\[\]\hspace{0.3cm} Rosane Ushirobira, Institut de
  Math\'ematiques de Bourgogne, Universit\'e de Bourgogne, B.P. 47870,
  F-21078 Dijon Cedex \& Non-A Inria Lille - Nord Europe, France}

\email{thanhdmi@hcmup.edu.vn}
\email{Rosane.Ushirobira@u-bourgogne.fr}



\date{\today}

\begin{abstract}
  In this paper, we give a generalization of results in \cite{PU07}
  and \cite{DPU10} by applying the tools of graded Lie algebras to
  quadratic Lie superalgebras. In this way, we obtain a numerical
  invariant of quadratic Lie superalgebras and a classification of
  singular quadratic Lie superalgebras, i.e. those with a nonzero
  invariant. Finally, we study a class of quadratic Lie superalgebras
  obtained by the method of generalized double extensions.
\end{abstract}

\maketitle

\section{Introduction}

Throughout the paper, the base field is $\CC$ and all vector spaces
are complex and finite-dimensional.  We denote the ring $\ZZ/2\ZZ$ by
$\ZZ_2$ as in superalgebra theory.

Let us begin with a $\ZZ_2$-graded vector space
$\g=\g_\ze\oplus\g_\un$. We denote by $\Alt(\g_\ze)$ the Grassmann
algebra of $\g_\ze$, that is, the algebra of alternating multilinear
forms on $\g_\ze$ equipped with the wedge product and by
$\sym(\g_\un)$ the algebra of symmetric multilinear forms on $\g_\un$.

We say that $\g$ is a \textit{quadratic} $\ZZ_2$-graded vector space
if it is endowed with a non-degenerate even supersymmetric bilinear
form $B$ (that is, $B$ is symmetric on $\g_\ze$, skew-symmetric on
$\g_\un$ and $B(\g_\ze,\g_\un)=0$). In addition, if there is a Lie
superalgebra structure $[\cdot, \cdot]$ on $\g$ such that $B$ is
invariant, i.e. $B([X,Y],Z) = B(X,[Y,Z])$ for all $X,Y,Z\in\g$, then
$\g$ is called a \textit{quadratic} (or \textit{orthogonal} or
\textit{metrised}) Lie superalgebra.

Algebras endowed with an invariant bilinear form appear in many areas
of Mathematics and Physics and they are a remarkable algebraic
object. A structural theory of quadratic Lie algebras, based on the
notion of a double extension (a combination of a central extension and
a semi-direct product), was introduced by V. Kac \cite{Kac85} in the
solvable case and by A. Medina and P. Revoy \cite{MR85} in the general
case. Another interesting construction, the $T^*$-extension, based on
the notion of a generalized semi-direct product of a Lie algebra and
its dual space was given by M. Bordemann \cite{Bor97} for solvable
quadratic Lie algebras. Both notions have been generalized for
quadratic Lie superalgebras in papers by H. Benamor and S. Benayadi
\cite{BB99} and by I. Bajo, S. Benayadi and M. Bordemann \cite{BBB}.

A third approach, based on the concept of super Poisson bracket, was
introduced in \cite{PU07}, providing several interesting properties of
quadratic Lie algebras: the authors consider $(\g,[ \cdot, \cdot], B)$
a non-Abelian quadratic Lie algebra and define a 3-form $I$ on $\g$ by
\[I(X,Y,Z) = B([X,Y],Z), \ \forall\ X,Y,Z\in\g.
\]
Then $I$ is nonzero and $\{I,I\} = 0$, where $\{ \cdot, \cdot \}$ is
the super Poisson bracket defined on the ($\ZZ$-graded) Grassmann
algebra $\Alt(\g)$ of $\g$, by
\[\{\Omega, \Omega '\} = (-1)^{\degz(\Omega)+1} \sum_{j=1}^n
\io_{X_j}(\Omega) \wedge \io_{X_j}(\Omega'), \ \forall \ \Omega, \
\Omega' \in \Alt(\g) \] with $\{X_1, \dots, X_n\}$ a fixed orthonormal
basis of $\g$.

Conversely, given a quadratic vector space $(\g,B)$ and a nonzero
3-form $I\in\Alt^3(\g)$ satisfying $\{I,I\} = 0$, then there is a
non-Abelian Lie algebra structure on $\g$ such that $B$ is invariant.

The element $I$ carries some useful information about corresponding
quadratic Lie algebras. For instance, when $I$ is decomposable and
nonzero, corresponding quadratic Lie algebras are called {\em
  elementary} quadratic Lie algebras and they are exhaustively
classified in \cite{PU07}. This classification is based on basic
properties of quadratic forms and the super Poisson bracket. In this
case, $\dim([\g,\g])=3$ and coadjoint orbits have dimension at most
2. In \cite{DPU10}, the authors consider further a notion that is
called the $\dup$-\textit{number} of a non-Abelian quadratic Lie
algebra $\g$. It is defined by $\dup(\g)=\dim\left(\{\alpha \in \g^*
  \mid\ \alpha \wedge I = 0\}\right)$ where $\g^*$ is the dual space
of $\g$. The dup-number receives values 0, 1 or 3 and it measures the
decomposability of $I$. For instance, $I$ is decomposable if and only
if $\dup(\g)=3$. Moreover, it is also a numerical invariant of
quadratic Lie algebras under Lie algebra isomorphisms, meaning that if
$\g$ and $\g'$ are isomorphic quadratic Lie algebras then
$\dup(\g)=\dup(\g')$. Its proof is rather non-trivial. It is obtained
through a description of the space generated by invariant symmetric
bilinear forms on a quadratic Lie algebra with nonzero
dup-number. Such quadratic Lie algebra is called a {\em singular}. An
unexpected property is that there are many non-degenerate invariant
symmetric bilinear forms on a singular quadratic Lie algebra. Though
they can be linearly independent all of them are equivalent in the
solvable case, i.e. two solvable singular quadratic Lie algebras with
same Lie algebra structure are isometrically isomorphic (or {\em
  i-isomorphic}, for short). In other words, isomorphic and
i-isomorphic notions are equivalent on solvable singular quadratic Lie
algebras. Another remarkable result is that all singular quadratic Lie
algebras are classified up to isomorphism by $O(n)$-adjoint orbits of
the Lie algebra $\ok(n)$.

The purpose of this paper is to give a interpretation of this last
approach for quadratic Lie superalgebras. We combine it with the
notion of double extension as it was done for quadratic Lie algebras
\cite{DPU10}. Further, we use the notion of generalized double
extension. In result, we obtain a rather colorful picture of quadratic
Lie superalgebras.

Let us give some details of our main results. First, let $\g$ be a
quadratic $\ZZ_2$-graded vector space. Recall that $\Alt(\g_\ze)$ and
$\sym(\g_\un)$ are $\ZZ$-graded algebras and this gradation can be
used to define a $\ZZ \times \ZZ_2$-gradation on each
algebra. Consider then the $\ZZ\times\ZZ_2$-graded
\textit{super-exterior algebra} of $\g^*$ defined by
 \[\Es(\g)= \Alt(\g_\ze)\
 \underset{\ZZ\times\ZZ_2}{\otimes} \ \sym(\g_\un)\] with the natural
 \textit{super-exterior product}
 \[ (\Omega \otimes F) \wedge (\Omega' \otimes F') =
 (-1)^{\degz(F)\degz(\Omega')} (\Omega \wedge \Omega') \otimes F
 F', \] for all $\Omega,\Omega'$ in $\Alt(\g_\ze)$ and $F,F'$ in
 $\sym(\g_\un)$. It is clear that $\Es$ is commutative and
 associative. For more details of the algebra $\Es(\g)$ the reader
 should refer to \cite{Sch79}, \cite{BP89} or \cite{MPU09}.

 In \cite{MPU09}, the authors use the \textit{super $\ZZ \times
   \ZZ_2$-Poisson bracket} $\{\cdot,\cdot\}$ on the super-exterior
 algebra $\Es(\g)$ defined as follows:
 \[ \{ \Omega \otimes F, \Omega' \otimes F' \} =
 (-1)^{\degz(F)\degz(\Omega')} \left( \{ \Omega, \Omega'\} \otimes F
   F' + (\Omega \wedge \Omega') \otimes \{F,F' \} \right),\] for all
 $\Omega$, $\Omega' \in \Alt(\g_\ze)$, $F$, $F' \in \sym(\g_\un)$. 

 In Section 1, we will recall some simple properties of the super $\ZZ
 \times \ZZ_2$-Poisson bracket that are necessary for our purpose.

 It is easy to check that for a quadratic Lie superalgebra $(\g,B)$,
 if we define a trilinear form $I$ on $\g$ by
 \[I(X,Y,Z) = B([X,Y], Z),\ \forall\ X,Y,Z\in \g\] then
 $I\in\Es^{(3,\zero)}(\g)$. Therefore, it seems to be natural to ask:
 when does $\{I,I\} = 0$? We shall give an affirmative answer to this
 in Proposition \ref{3.1.17}. Moreover, similarly to the Lie algebra
 case, we show that non-Abelian quadratic Lie superalgebra structures
 on a quadratic $\ZZ_2$-graded vector space $(\g,B)$ are in one-to-one
 correspondence with nonzero elements $I$ in $\Es^{(3,\zero)}(\g)$
 satisfying $\{I,I\} = 0$.

 In Section 2, we introduce the notion of $\dup$-\textit{number} for a
 non-Abelian quadratic Lie superalgebra $\g$:
 \[\dup(\g)=\dim\left(\{\alpha \in \g^* \mid\ \alpha \wedge I = 0
   \}\right)\] and consider the set of quadratic Lie superalgebras
 with nonzero $\dup$-number: the set of singular quadratic Lie
 superalgebras. Similarly to quadratic Lie algebras, if $\g$ is
 non-Abelian then $\dup(\g)\in\{0,1,3\}$. Thanks to Lemma \ref{3.2.1},
 if $\dup(\g)=3$ then $\g_\un$ is a central ideal of $\g$ and $\g_\ze$
 is an elementary quadratic Lie algebra. Therefore, we focus on
 singular quadratic Lie superalgebras $\g$ with $\dup(\g) =1$. We call
 them \textit{singular quadratic Lie superalgebras of type
   $\Sb_1$}. However, differently than the Lie algebra case, the
 element $I$ may be decomposable.

 We list in Section 3 all non-Abelian reduced quadratic Lie
 superalgebras with $I$ decomposable (see Definition \ref{3.2.5} for
 the definition of a reduced quadratic Lie superalgebra). We call them
 {\em elementary} quadratic Lie superalgebras. In this case,
 $\dup(\g)$ is nonzero. In particular, if $\dup(\g)=3$ then $\g$ is an
 elementary quadratic Lie algebra. If $\dup(\g)=1$ then we obtain
 three quadratic Lie superalgebras with 2-dimensional even
 part. Actually, we prove in Proposition \ref{3.4.1} that if $\g$ is a
 non-Abelian quadratic Lie superalgebra with 2-dimensional even part
 then $\dup(\g)=1$.

 Section 4 details a study of quadratic Lie superalgebras with
 2-dimensional even part. We apply the concept of double extension as
 in \cite{DPU10} with a little change by replacing a quadratic vector
 space by a symplectic vector space and keeping the other conditions
 (see Definition \ref{3.4.6}). Then we obtain a structure that we
 still call a {\em double extension} and one has (Proposition
 \ref{3.4.8}):

\bigskip
{\sc{Theorem 1:}}

{\em A quadratic Lie superalgebra has a 2-dimensional even part if and
  only if it is a double extension}.  \bigskip

By a very similar process as in \cite{DPU10} for solvable singular
quadratic Lie algebras, a classification of quadratic Lie
superalgebras with 2-dimensional even part up to isomorphism is given
as follows. Let $\Ss(2+2n)$ be the set of such structures on
$\CC^{2+2n}$. We call an algebra $\g\in\Ss(2+2n)$ {\em diagonalizable}
(resp. {\em invertible}) if it is the double extension by a
diagonalizable (resp. invertible) map. Denote the subsets of nilpotent
elements, diagonalizable elements and invertible elements in
$\Ss(2+2n)$, respectively by $\Ns(2+2n)$, $\Ds(2+2n)$ and by
$\Si(2+2n)$. Denote by $\Nsh(2+2n)$, $\Dsh(2+2n)$, $\Sih(2+2n)$ the
sets of isomorphic classes in $\Ns(2+2n)$, $\Ds(2+2n)$, $\Si(2+2n)$
respectively and by $\Dsh_\mathrm{red}(2+2n)$ the subset of
$\Dsh(2+2n)$ including reduced ones. Also, we denote by
$\ps(\spk(2n))$ the projective space of $\spk(2n)$ with the action
induced by the $\Sp(2n)$-adjoint action on $\spk(2n)$. Then we have
the classification result (Propositions \ref{3.4.13} and \ref{3.4.14}
and Appendix):

\bigskip
{\sc{Theorem 2:}} 
{\em
\begin{itemize}
\item[(i)] Let $\g$ and $\g'$ be elements in $\Ss(2+2n)$. Then $\g$
  and $\g'$ are isometrically isomorphic and if and only if they are
  isomorphic.
\item[(ii)] There is a bijection between $\Nsh(2+2n)$ and the set of
  nilpotent $\Sp(2n)$-adjoint orbits of $\spk(2n)$. It induces a
  bijection between $\Nsh(2+2n)$ and the set of partitions
  $\Pc_{-1}(2n)$ of $2n$ in which odd parts occur with even
  multiplicity.
\item[(iii)] There is a bijection between $\Dsh(2+2n)$ and the set of
  semisimple $\Sp(2n)$-orbits of $\ \ps(\spk(2n))$. 
\item[(iv)] There is a bijection between $\Sih(2+2n)$ and the set of
  invertible $\Sp(2n)$-orbits of $\ \ps(\spk(2n))$. 
\item[(v)] There is a bijection between $\widehat{\Ss}(2+2n)$ and the
  set of $\Sp(2n)$-orbits of
  $\ps(\spk(2n))$. 
\end{itemize}}

\bigskip

As for quadratic Lie algebras, we have the notion of quadratic
dimension of a quadratic Lie superalgebra. In the case $\g$ is a
quadratic Lie superalgebra having a 2-dimensional even part, we can
compute its quadratic dimension as follows:
\[ d_q(\g) = 2 + \dfrac{(\dim(\Zs(\g) - 1)) (\dim(\Zs(\g) - 2)}{2}. \]
where $\Zs(\g)$ is the center of $\g$. It indicates that there are
many non-degenerate invariant even supersymmetric bilinear forms on a
quadratic Lie superalgebra with 2-dimensional even part but by Theorem
2 (i), all of them are equivalent.

Section 5 contains more results on a singular quadratic Lie
superalgebra $(\g,B)$ of type $\Sb_1$, that is, those with 1-valued
$\dup$-number. The first result is that $\g_\ze$
is solvable and so $\g$ is solvable. Moreover, by Definition \ref{3.5.3} and Lemma
\ref{3.5.5}, the Lie superalgebra $\g$ can be realized as the double
extension of a quadratic $\ZZ_2$-graded vector space
$\qk=\qk_\ze\oplus\qk_\un$ by a map $\cb=\cb_0+\cb_1
\in\ok(\qk_\ze)\oplus\spk(\qk_\un)$. Denote by $\Lc(\qk_\ze)$
(resp. $\Lc(\qk_\un)$) the set of endomorphisms of $\qk_\ze$
(resp. $\qk_\un$). We give isomorphic and i-isomorphic
characterizations of two singular quadratic Lie superalgebras of type
$\Sb_1$ as follows (Proposition \ref{3.5.7}).

\bigskip
{\sc{Theorem 3:}}

{\em Let $\g$ and $\g'$ be two double extensions of $\qk$ by $\cb
  =\cb_0+\cb_1$ and $\overline{C'} =\overline{C_0'}+\overline{C_1'}$,
  respectively. Assume that $\cb_1$ is nonzero. Then
\begin{enumerate}
\item there exists a Lie superalgebra isomorphism between $\g$ and
  $\g'$ if and only if there exist invertible maps $P\in\Lc(\qk_\ze)$,
  $Q\in\Lc(\qk_\un)$ and a nonzero $\lambda\in\CC$ such that
	
\begin{itemize}
	\item[(i)] $\overline{C_0'}=\lambda P\cb_0 P^{-1}$ and $P^*P\cb_0 = \cb_0$.
	\item[(ii)] $\overline{C_1'}=\lambda Q\cb_1 Q^{-1}$ and $Q^*Q\cb_1 = \cb_1$.
\end{itemize}
where $P^*$ and $Q^*$ are the adjoint maps of $P$ and $Q$ with respect to $B|_{\qk_\ze\times\qk_\ze}$ and $B|_{\qk_\un\times\qk_\un}$.
\item there exists an i-isomorphism between $\g$ and $\g'$ if and only
  if there is a nonzero $\lambda\in\CC$ such that $\overline{C_0'}$ is
  in the $\OO(\qk_\ze)$-adjoint orbit through $\lambda \cb_0$ and
  $\overline{C_1'}$ is in the $\Sp(\qk_\un)$-adjoint orbit through
  $\lambda \cb_1$.
\end{enumerate}}

We recall a remarkable result in \cite{DPU10} that two solvable
singular quadratic Lie algebras are i-isomorphic if and only if they
are isomorphic. A similar situation occurs for two quadratic Lie
superalgebras with 2-dimensional even part as in Theorem 2. Therefore,
there is a very natural question: is this result also true for two
singular quadratic Lie superalgebras? We have an affirmative answer as
follows (Proposition \ref{3.5.15} for $\dup(\g)=1$ and \cite{DPU10}
for $\dup(\g)=3$):

\bigskip
{\sc{Theorem 4:}}

{\em Let $\g$ and $\g'$ be two solvable singular quadratic Lie
  superalgebras. Then $\g$ and $\g'$ are i-isomorphic if and only if
  they are isomorphic}.  \bigskip

We close the problem on singular quadratic Lie superalgebras with an
assertion that (Proposition \ref{3.5.16}):

\bigskip
{\sc{Theorem 5:}}

{\em The $\dup$-number is invariant under Lie superalgebra
  isomorphism}.

\bigskip

As a consequence of its proof, we obtain a formula for the quadratic
dimension of reduced singular quadratic Lie superalgebras of type
$\Sb_1$ having $[\g_\un,\g_\un]\neq\{0\}$.

In the last Section, we study the structure of a quadratic Lie
superalgebra $\g$ such that its element $I$ has the form:
\[I = J\wedge p\] where $p\in \g_{\un}^*$ is nonzero and $J\in
\Alt{}^1(\g_{\ze})\otimes \sym{}^1(\g_{\un})$ is indecomposable. We
call $\g$ a {\em quasi-singular quadratic Lie superalgebra}. With the
notion of {\em generalized double extension} given by I. Bajo,
S. Benayadi and M. Bordemann in \cite{BBB}, we prove that (Corollary
\ref{3.6.5} and Proposition \ref{3.6.8}):

\bigskip

{\sc{Theorem 6:}}

{\em A quasi-singular quadratic Lie superalgebra is a generalized
  double extension of a quadratic $\ZZ_2$-graded vector space. This
  superalgebra is 2-nilpotent.}  \bigskip

In the Appendix, we recall fundamental results in the classification
of $\OO(m)$-adjoint orbits of $\ok(m)$ and $\Sp(2n)$-adjoint orbits of
$\spk(2n)$. The classification of nilpotent and semisimple orbits is
well-known. We further give here the classification of {\em
  invertible} orbits, i.e. orbits of isomorphisms in $\ok(m)$ and
$\spk(2n)$. By the Fitting decomposition, we obtain a complete
classification in the general case.

Many concepts used in this paper are generalizations of the quadratic
Lie algebra case. We do not recall their original definitions
here. For more details the reader can refer to \cite{PU07} and
\cite{DPU10}.

\begin{aknow}
  We would like to thank D. Arnal and R. Yu for many valuable remarks
  and stimulating discussions concerning the first version of this
  article. Moreover, we would like to thank S. Benayadi for very
  interesting suggestions for the improvement of Section 5.

  This article is dedicated to our admirable mentor Georges Pinczon
  (1948 -- 2010). He suggested the main idea in Section 1 and
  discussed results in Sections 2 and 3.
\end{aknow}

\section{Applications of graded Lie algebras to quadratic Lie
  superalgebras}

Let $\g=\g_\ze\oplus \g_\un$ be a $\ZZ_2$-graded vector space. We call
$\g_\ze$ and $\g_\un$ respectively the {\em even} and the {\em odd}
part of $\g$.  We begin by reviewing the construction of the
super-exterior algebra of the dual space $\g^*$ of $\g$. Then we
define the super $\ZZ \times \ZZ_2$-Poisson bracket on $\g^*$ (for
more details, see \cite{MPU09} and \cite{Sch79}). 

\subsection{The super-exterior algebra of $\gd$} Denote by
$\Alt(\g_\ze)$ the algebra of alternating multilinear forms on
$\g_\ze$ and by $\sym(\g_\un)$ the algebra of symmetric multilinear
forms on $\g_\un$. Recall that $\Alt(\g_\ze)$ is the exterior algebra
of $\g_\ze^*$ and $\sym(\g_\un)$ is the symmetric algebra of
$\g_\un^*$. These algebras are $\ZZ$-graded algebras. We define a $\ZZ
\times \ZZ_2$-gradation on $\Alt(\gO)$ and on $\sym(\gI)$ by
\[\Alt^{(i,\zero)}(\gO) = \Alt^i(\gO), \quad
\Alt^{(i,\um)}(\gO) = \{0\}\]
\[\text{and}
\quad \sym^{(i,\overline{i})}(\gI) = \sym^i(\gI),\quad
\sym^{(i,\overline{j})}(\gI) = \{0\} \quad \text{if} \quad
\overline{i} \neq \overline{j},\] where $i,j\in\ZZ$ and $\overline{i},
\overline{j}$ are respectively the residue classes modulo 2 of $i$ and
$j$.

The {\it super-exterior algebra} of $\gd$ is the $\ZZ \times \ZZ_2$-graded
algebra defined by: \[\Es(\g)= \Alt(\gO) \underset{\ZZ   \times
\ZZ_2}{\otimes} \sym(\gI)\] endowed with the {\em   super-exterior product} on
$\Es(\g)$: \[ (\Omega \otimes F) \wedge (\Omega' \otimes F') = (-1)^{f\omega'}
(\Omega \wedge \Omega') \otimes F F', \] for all $\Omega \in \Alt(\gO)$,
$\Omega' \in \Alt^{\omega'}(\gO)$, $F \in \sym^f(\gI)$, $F' \in \sym(\gI)$.
Remark that the $\ZZ \times \ZZ_2$-gradation on $\Es(\g)$ is given by: \[
\text{if } A = \Omega \otimes F\in \Alt^{\omega}(\gO) \otimes \sym^{f}(\gI) \
\text{ with } \omega,  \ f \in \ZZ, \ \text{ then } A \in \Es^{(\omega + f,
\overline{f})}(\g).\] So, in terms of the $\ZZ$-gradations of $\Alt(\gO)$ and
$\sym(\gI)$, we have: \[\Es^n(\g) = \bigoplus^{n}_{m=0}\left(\Alt^m(\gO)
\otimes \sym^{n-m}(\gI)\right)\] and in terms of the $\ZZ_2$-gradations,
\[\Es_{\ze}(\g) = \Alt(\g_\ze)
\otimes\left(\underset{j\geq0}{\oplus}\sym^{2j}(\g_\un) \right)\ \text{ and }\
\Es_{\un}(\g) = \Alt(\g_\ze)
\otimes\left(\underset{j\geq0}{\oplus}\sym^{2j+1}(\g_\un) \right).\] Notice
that the graded vector space $\Es(\g)$ endowed with this product is a
commutative and associative graded algebra.

Another equivalent construction is given in \cite{BP89}: $\Es(\g)$ is
the graded algebra of super-antisymmetric multilinear forms on
$\g$. The algebras $\Alt(\gO)$ and $\sym(\gI)$ are regarded as
subalgebras of $\Es(\g)$ by identifying $\Omega: = \Omega\otimes 1$,
$F: = 1\otimes F$, and the tensor product $\Omega \otimes F =
(\Omega\otimes 1)\wedge (1\otimes F)$ for all $\Omega\in \Alt(\gO)$,
$F\in \sym(\gI)$.

\subsection{The super $\ZZ\times \ZZ_2$-Poisson bracket on $\Es(\gd)$}
Let us assume that the vector space $\g$ is equipped with a
non-degenerate even supersymmetric bilinear form $B$. That
means \[B(X,Y) = (-1)^{xy}B(Y,X)\] for all homogeneous $X\in\g_x$,
$Y\in\g_y$ and $B(\g_\ze,\g_\un) = 0$. In this case, $\dim(\g_\un)$
must be even and $\g$ is also called a {\em quadratic $\ZZ_2$-graded
  vector space.}

The Poisson bracket on $\sym(\gI)$ and the super Poisson bracket on
$\Alt(\gO)$ are defined as follows. Let $\Bc = \{X_1, ..., X_n, Y_1,
..., Y_n\}$ be a Darboux basis of $\g_\un$, meaning that $B(X_i,X_j) =
B(Y_i,Y_j) = 0$ and $B(X_i,Y_j) = \delta_{ij}$, for all $1 \leq i,j
\leq n$. Let $\{p_1, ..., p_n, q_1, ..., q_n\}$ be its dual
basis. Then the algebra $\sym(\gI)$ regarded as the polynomial algebra
$\CC[p_1, ..., p_n, q_1, ..., q_n]$ is equipped with the {\em Poisson
bracket}:
\[ \{F,G\} = \sum_{i=1}^n \left( \dfrac{\partial F}{\partial p_i}
  \dfrac{\partial G}{\partial q_i} - \dfrac{\partial F}{\partial q_i}
  \dfrac{\partial G}{\partial p_i} \right), \ \text{for all} \ F, G
\in \sym(\gI).
\]
It is well-known that the algebra $\left( \sym(\gI), \{ \cdot, \cdot \}
\right)$ is a Lie algebra.  Now, let $X \in \g_\ze$ and denote by
$\iota_X$ the derivation of $\Alt(\gO)$ defined by:
\[
\iota_X (\Omega) (Z_1, \dots, Z_k) = \Omega (X, Z_1, \dots, Z_k),\
\forall \ \Omega \in \Alt^{k+1}(\gO),\ X, Z_1, \dots, Z_k \in \g_\ze \ (k
\geq 0),
\]
and $\iota_X(1) = 0$. Let $\{Z_1, \dots, Z_m\}$ be a fixed orthonormal basis of
$\g_\ze$. The {\em super Poisson bracket} on the algebra
$\Alt(\gO)$ is defined by (see \cite{PU07} for details):
\[\{\Omega, \Omega '\} = (-1)^{k+1} \sum_{j=1}^m \io_{Z_j}(\Omega)
\wedge \io_{Z_j}(\Omega'), \ \forall \ \Omega \in \Alt^k(\gO),\ \Omega'
\in \Alt(\gO).\] Remark that the definitions above do not depend on the choice of the basis.

Next, for any $\Omega \in \Alt^k(\gO)$, we define the map $\adpo(\Omega)$ by \[
\adpo(\Omega)\left(\Omega'\right)  = \{ \Omega, \Omega'\}, \ \forall
\ \Omega' \in \Alt(\gO).\] It is easy to check that $\adpo(\Omega)$ is a super-derivation
of degree $k-2$ of the algebra $\Alt(\gO)$. One
has: \[\adpo(\Omega)\left(\{\Omega', \Omega''\}\right) = \{
\adpo(\Omega)(\Omega'), \Omega''\} + (-1)^{kk'} \{ \Omega',
\adpo(\Omega)(\Omega'')\}, \] for all $\Omega' \in \Alt^{k'}(\gO)$,
$\Omega '' \in \Alt(\gO)$. Therefore $\Alt(\gO)$ is a graded
Lie algebra for the super-Poisson bracket.

\begin{defn}\cite{MPU09}\label{3.1.4}

The {\em super $\ZZ \times \ZZ_2$-Poisson bracket} on $\Es(\g)$ is given by:
\[ \{ \Omega \otimes F, \Omega' \otimes F' \} = (-1)^{f \omega'}
\left( \{ \Omega, \Omega'\} \otimes F F' + (\Omega \wedge \Omega')
  \otimes \{F,F' \} \right),\] for all $\Omega \in \Alt(\gO)$, $\Omega'
\in \Alt^{\omega'}(\gO)$, $F \in \sym(\gI){}^f$, $F' \in \sym(\gI)$.
\end{defn}

By a straightforward computation, it is easy to obtain the following result:
\begin{prop}\label{3.1.5}
The algebra $\Es(\g)$ is a graded Lie algebra with the super $\ZZ \times \ZZ_2$-Poisson bracket. More precisely, for all $A  \in \Es^{(a,b)}(\g)$, $A' \in \Es^{(a',b')}(\g)$ and $A'' \in \Es^{(a'',b'')}(\g)$:

\begin{enumerate}
	\item $\{A',A\} = -(-1)^{aa'+ bb'}\{A,A'\}.$
	\item $\{\{A,A'\},A''\} = \{A,\{A', A''\}\} - (-1)^{aa'+ bb'}\{A',\{A, A''\}\}.$
\end{enumerate}
Moreover, one has $\{A,A'\wedge A''\} = \{A,A'\}\wedge A'' + (-1)^{aa'+ bb'}A'\wedge\{A,A''\}.$
\end{prop}

\subsection{Super-derivations} Denote by $\Lc(\Es(\g))$ the vector space of
endomorphisms of $\Es(\g)$. Let $\adpo(A) := \{A,.\}$, for all $A \in
\Es(\g)$. Then $\adpo(A)\in \Lc(\Es(\g))$ and: \[\adpo(\{A,A'\}) =
\adpo(A)\circ\adpo(A') - (-1)^{aa'+ bb'}\adpo(A')\circ\adpo(A)\] for all $A$,
$A' \in \Es^{(a',b')}(\g)$. The space $\Lc(\Es(\g))$ is naturally $\ZZ \times
\ZZ_2$-graded as follows: \[ \deg(F) = (n,d),\ n\in\ZZ,\ d\in\ZZ_2 \text{ if }
\deg(F(A)) = (n + a, d + b),\ \text{where}\ A\in \Es^{(a,b)}(\g).\]

We denote by $\End^n_f(\Es(\g))$ the subspace of endomorphisms of degree
$(n,f)$ of $\Lc(\Es(\g))$. It is clear that if $A\in \Es^{(a,b)}(\g)$ then
$\adpo(A)$ has degree $(a-2,b)$. Moreover, it is known that $\Lc(\Es(\g))$ is
also a graded Lie algebra, frequently denoted by $\gl(\Es(\g))$ and equipped
with the Lie super-bracket: \[ [F,G] = F\circ G - (-1)^{np+fg}G\circ F,\ \
\forall\ F\in \End^n_f(\Es(\g)),\ G\in \End^p_g(\Es(\g)).\] Therefore, by
Proposition \ref{3.1.5}, we obtain that $\adpo$ is a graded Lie algebra
homomorphism from $\Es(\g)$ onto $\gl(\Es(\g))$. In other words, one has:
\[\adpo(\{A,A'\}) = [\adpo(A), \adpo(A')], \ \ \forall \ A, A'\in \Es(\g).\]

\begin{defn}
A homogeneous endomorphism $D\in \gl(\Es(\g))$ of degree $(n,d)$ is called a {\em super-derivation} of degree $(n,d)$ of $\Es(\g)$ (for the super-exterior product) if it satisfies the following condition:
\[D(A\wedge A') = D(A)\wedge A' + (-1)^{na+db}A\wedge D(A'),\ \forall\ A\in \Es^{(a,b)}(\g),\ A'\in \Es(\g).
\]
\end{defn}

Denote by $\Ds^n_d(\Es(\g))$ the space of super-derivations of degree $(n,d)$
of $\Es(\g)$ then we obtain a $\ZZ \times \ZZ_2$-gradation of the space of
super-derivations $\Ds(\Es(\g))$ of $\Es(\g)$: \[\Ds(\Es(\g)) =
\bigoplus_{(n,d) \in\ZZ \times \ZZ_2}\Ds^n_d(\Es(\g)) \] and $\Ds(\Es(\g))$
becomes a graded subalgebra of $\gl(\Es(\g))$ \cite{NR66}. Moreover, the last
formula in Proposition \ref{3.1.5} affirms that $\adpo(A)\in\Ds(\Es(\g))$, for
all $A\in \Es(\g)$.

Another example of a super-derivation in $\Ds(\Es(\g))$ is given in
\cite{BP89} as follows. Let $X\in \g_x$ be a homogeneous element in $\g$ of
degree $x$ and define the endomorphism $\io_X$ of $\Es(\g)$ by \[\io_X(A)(X_1,
..., X_{a-1}) = (-1)^{xb}A(X,X_1, ..., X_{a-1}) \] for all
$A\in\Es^{(a,b)}(\g),\ X_1, \dots, X_{a-1}\in V$. Then one has \[\io_X(A\wedge
A') = \io_X(A)\wedge A' + (-1)^{-a+xb}A\wedge\io_X(A') \] holds for all
$A\in\Es^{(a,b)}(\g),\ A'\in \Es(\g)$. It means that $\io_X$ is a super-
derivation of $\Es(\g)$ of degree $(-1,x)$. The proof of the following Lemma
is straightforward:

\begin{lem}\label{3.1.10} Let $X_\ze\in \g_\ze$ and $X_\un\in \g_\un$. Then,
for all $\Omega\otimes F\in \Alt^{\omega}(\gO)\otimes \sym^f(\gI)$:

\begin{enumerate}
	\item $\io_{X_\ze} (\Omega\otimes F)=\io_{X_\ze} (\Omega)\otimes F $,
	\item $\io_{X_\un} (\Omega\otimes F)=(-1)^{\omega}\Omega\otimes\io_{X_\un} (F)$.
\end{enumerate}
\end{lem}


\begin{rem}\label{3.1.11}\hfill
\begin{enumerate}
	
  \item If $\Omega\in \Alt^{\omega}(\gO)$ then $\io_{X}(\Omega)(X_1, \dots,
X_{\omega-1}) = \Omega(X,X_1, \dots, X_{\omega-1})$ for all $ X, X_1, \dots,
X_{\omega-1}\in \g_\ze$. That coincides with the previous definition of
$\io_X$ on $\Alt(\gO)$.

  \item Let $X$ be an element in a fixed Darboux basis of $\g_\un$ and $p\in
\g_\un^*$ be its dual form. By the Corollary II.1.52 in \cite{Gie04} one has:
\[\io_X(p^n)(X^{n-1}) = (-1)^n p^n(X^n) = (-1)^n(-1)^{n(n-1)/2}n!.\] Moreover,
$\dfrac{\partial p^n}{\partial p}(X^{n-1}) =
n(p^{n-1})(X^{n-1})=(-1)^{(n-1)(n-2)/2}n!$. It implies that
\[\io_X(p^n)(X^{n-1}) = -\dfrac{\partial p^n}{\partial p}(X^{n-1}).\] Since
each $F\in \sym^f(\gI)$ can be regarded as a polynomial in the variable $p$,
one has the following property: \[\io_X(F) = -\dfrac{\partial F}{\partial p},\
\forall\ F\in \sym(\gI).   \]

\end{enumerate}
\end{rem}

The super-derivations $\io_X$ play an important role in the description of
the space $\Ds(\Es(\g))$ (for details, see \cite{Gie04}). For instance, they
can used to express the super-derivation $\adpo(A)$ defined above:

\begin{prop}\label{3.1.12}
Fix an orthonormal basis $\{X_\ze^1, \dots, X_\ze^m\}$ of $\g_\ze$ and a Darboux basis $\Bc=\{X_\un^1, \dots, X_\un^n, Y_\un^1, \dots, Y_\un^n\}$ of $\g_\un$. Then the super $\ZZ \times \ZZ_2$-Poisson bracket on $\Es(\g)$ is given by:
\begin{eqnarray*}
\{A,A'\}& = & (-1)^{\omega+f+1}\sum^{m}_{j=1}\io_{X_\ze^j}(A)\wedge\io_{X_\ze^j}(A')
\\&&+(-1)^{\omega}\sum^{n}_{k=1}\left(\io_{X_\un^k}(A)\wedge\io_{Y_\un^k}(A') -\io_{Y_\un^k}(A)\wedge\io_{X_\un^k}(A')\right)
\end{eqnarray*}
for all $A\in \Alt^{\omega}(\gO) \otimes \sym^f(\gI)$ and $A' \in \Es(\g)$. 

\end{prop}

\begin{proof} Let $A = \Omega \otimes F\in \Alt^{\omega}(\gO) \otimes
\sym^f(\gI)$ and $A' = \Omega' \otimes F'\in \Alt^{\omega'}(\gO) \otimes
\sym^{f'}(\gI)$. The super $\ZZ \times \ZZ_2$-Poisson bracket of $A$ and $A'$
is defined by: \[ \{ A, A' \} = (-1)^{f \omega'} \left( \{ \Omega, \Omega'\}
\otimes F F' + (\Omega \wedge \Omega')   \otimes \{F,F' \} \right).\]

By the definition of the super Poisson bracket on $\Alt(\gO)$ combined with Lemma \ref{3.1.10} (1), one has
\begin{eqnarray*}
&\{ \Omega, \Omega'\}\otimes FF'& = (-1)^{\omega+1}\sum^{m}_{j=1}\left(\io_{X_\ze^j}(\Omega)\wedge\io_{X_\ze^j}(\Omega')\right)\otimes FF'
\\&& = (-1)^{f(\omega'-1)+\omega+1}\sum^{m}_{j=1}\left(\io_{X_\ze^j}(\Omega)\otimes F\right)\wedge\left(\io_{X_\ze^j}(\Omega')\otimes F'\right)
\\ &&= (-1)^{f\omega'+\omega+f+1}\sum^{m}_{j=1}\io_{X_\ze^j}(A)\wedge\io_{X_\ze^j}(A').
\end{eqnarray*}

Let $\{p_1, \dots, p_n, q_1, \dots, q_n\}$ be the dual basis of $\Bc$. By Remark \ref{3.1.11} (2), the Poisson bracket on $\sym(\gI)$ can be expressed by:
\[ \{F,F'\} = \sum_{k=1}^n \left( \dfrac{\partial F}{\partial p_k}
\dfrac{\partial F'}{\partial q_k} - \dfrac{\partial F}{\partial q_k}
\dfrac{\partial F'}{\partial p_k} \right)=\sum^{n}_{k=1}\left(\io_{X_\un^k}(F)\io_{Y_\un^k}(F')-\io_{Y_\un^k}(F)\io_{X_\un^k}(F')\right).
\]
Combined with Lemma \ref{3.1.10} (2), we obtain
\[ (\Omega \wedge \Omega')   \otimes \{F,F' \} = (\Omega \wedge
\Omega')\otimes \sum^{n}_{k=1}\left(\io_{X_\un^k}(F)\io_{Y_\un^k}(F')-\io_{Y_\un^k}(F)\io_{X_\un^k}(F')\right)\] \[= (-1)^{(f-1)\omega'}\sum^{n}_{k=1}\left(
\left(\Omega\otimes\io_{X_\un^k}(F)\right)\wedge\left(\Omega'\otimes\io_{Y_\un
^k}(F')\right) - \left(\Omega\otimes\io_{Y_\un^k}(F)\right)\wedge\left(\Omega'
\otimes\io_{X_\un^k}(F')\right)\right)\]   \[=(-1)^{f\omega'+\omega}\sum^{n}_{
k=1}\left(\io_{X_\un^k}(A)\wedge\io_{Y_\un^k}(A')-\io_{Y_\un^k}(A)\wedge\io_{X
_\un^k}(A')\right). \]  The result follows.

\end{proof}

Since the bilinear form $B$ is non-
degenerate and even, then there is an (even) isomorphism $\phi$ from $\g$ onto
$\g^*$ defined by $\phi(X)(Y) = B(X,Y)$, for all $X$, $Y\in\g$.

\begin{cor}\label{3.1.13}

The following expressions:

\begin{enumerate}

  \item $ \{ \alpha, A \} = \io_{\phi^{-1}(\alpha)}(A)$,

	\item $ \{ \alpha, \alpha' \} = B(\phi^{-1}(\alpha), \phi^{-1}(\alpha'))$, 

\end{enumerate}

hold for all $\alpha, \alpha'\in \g^*,\ A\in \Es(\g)$.
\end{cor}

\begin{proof}\hfill
\begin{enumerate}
	\item We apply Proposition \ref{3.1.12}, respectively for $\alpha = (X_\ze^i)^* = \phi(X_\ze^i)$, $i=1,\dots, m$, $\alpha = (Y_\un^l)^* = \phi(X_\un^l)$ and $\alpha = (-X_\un^l)^* = \phi(Y_\un^l)$, $l=1,\dots, n$ to obtain the result.
	\item Let $\alpha\in\g_x^*,\ \alpha'\in\g_{x'}^*$ be homogeneous forms in $\g^*$, one has
	\begin{eqnarray*}
	\hspace{1.0cm} \{\alpha,\alpha'\} = \io_{\phi^{-1}(\alpha)}(\alpha')=(-1)^{xx'}\alpha'(\phi^{-1}(\alpha)) 
	= (-1)^{xx'}B(\phi^{-1}(\alpha'),\phi^{-1}(\alpha)) \\ = B(\phi^{-1}(\alpha), \phi^{-1}(\alpha')).
	\end{eqnarray*}
\end{enumerate}
\end{proof}

In this section, Proposition \ref{3.1.12} and Corollary \ref{3.1.13} are enough for our purpose. But as a consequence of Lemma 6.9 in \cite{PU07}, one has a more general result of Proposition \ref{3.1.12} as follows:

\begin{prop}\label{3.1.14} 
Let $\{X_\ze^1, \dots, X_\ze^m\}$ be a basis of
$\g_\ze$ and $\{\alpha_1, \dots, \alpha_m\}$ its dual basis. Let $\{Y_\ze^1,
\dots, Y_\ze^m\}$ be the basis of $\g_\ze$ defined by $Y_\ze^i =
\phi^{-1}(\alpha_i)$. Set $\Bc=\{X_\un^1, \dots, X_\un^n, Y_\un^1, \dots,
Y_\un^n\}$ be a Darboux basis of $\g_\un$. Then the super $\ZZ \times
\ZZ_2$-Poisson bracket on $\Es(\g)$ is given by
\begin{eqnarray*}
\{A,A'\} &=& (-1)^{\omega+f+1}\sum^{m}_{i,j=1}B(Y_\ze^i,Y_\ze^j)\io_{X_\ze^i}(A)\wedge\io_{X_\ze^j}(A') \\ &&+(-1)^{\omega}\sum^{n}_{k=1}\left(\io_{X_\un^k}(A)\wedge\io_{Y_\un^k}(A')-\io_{Y_\un^k}(A)\wedge\io_{X_\un^k}(A')\right)
\end{eqnarray*}
for all $A \in \Alt^{\omega}(\gO) \otimes \sym^f(\gI)$ and $A' \in \Es(\g)$. 
\end{prop}

\subsection{Super-antisymmetric linear maps}
Consider the vector space
\[\Ec= \bigoplus_{n\in\ZZ}\limits\Ec^n,\]
where $\Ec^n = \{0\}$ if $n\leq -2$, $\Ec^{-1} = \g$ and $\Ec^n$ is the space of super-antisymmetric $n+1$-linear mappings from $\g^{n+1}$ onto $\g$.  Each of the subspaces $\Ec^n$ is $\ZZ_2$-graded then the space $\Ec$ is $\ZZ\times\ZZ_2$-graded by 
\[\Ec = \bigoplus^{n\in\ZZ}_{f\in\ZZ_2}\limits\Ec^n_f.\]

There is a natural isomorphism between the spaces $\Ec$ and
$\Es(\g)\otimes\g$. Moreover, $\Ec$ is a graded Lie algebra, called
the {\em graded Lie algebra} of $\g$. It is isomorphic to
$\Ds(\Es(\g))$ by the graded Lie algebra isomorphism $D$ such that if
$F=\Omega\otimes X\in \Ec^n_{\omega +x}$ then $D_F =
-(-1)^{x\omega}\Omega\wedge \iota_X\in \Ds^n_{\omega
  +x}(\Es(\g))$. For more details on the Lemma below, see for
instance, \cite{BP89} and \cite{Gie04}.

\begin{lem}\label{3.1.15}\hfill

Fix $F\in \Ec^1_\ze$, denote by $d=D_F$ and define the product $[X,Y]
= F(X,Y)$, for all $X,Y\in \g$. Then one has
\begin{enumerate}
	\item $d(\phi)(X,Y) = -\phi([X,Y])$, for all $X,Y\in \g,\ \phi\in \g^*.$

	\item The product $[~,~]$ becomes a Lie super-bracket if and only
       if $d^2 = 0$. In this case, $d$ is called a {\em super-exterior
         differential} of $\Es(\g)$.
\end{enumerate}
\end{lem}

\subsection{Quadratic Lie superalgebras} The construction of graded Lie
algebras and the super $\ZZ \times \ZZ_2$-Poisson bracket above can be applied
to the theory of quadratic Lie superalgebras. This later is regarded as a
graded version of the quadratic Lie algebra case and we obtain then similar
results.

\begin{defn}
A {\em quadratic} Lie superalgebra $(\g,B)$ is a $\ZZ_2$-graded vector
space $\g$ equipped with a non-degenerate even supersymmetric bilinear
form $B$ and a Lie superalgebra structure such that $B$ is invariant,
i.e. $B([X,Y],Z) = B(X,[Y,Z])$, for all $X$, $Y$, $Z \in \g$.
\end{defn}

\begin{prop}\label{3.1.17}
Let $(\g,B)$ be a quadratic Lie superalgebra and define a trilinear form $I$ on $\g$ by
\[I(X,Y,Z) = B([X,Y], Z),\ \forall\ X,Y,Z\in \g.\]
Then one has
\begin{enumerate}
	\item $I\in \Es^{(3,\zero)}(\g) = \Alt^3(\g_\ze)\oplus\left(\Alt^1(\gO) \otimes 
  \sym^2(\gI) \right).$
	\item $d = -\adpo(I).$
	\item $\{I,I\} = 0.$
\end{enumerate}
\end{prop} 

\begin{proof} 
The assertion (1) follows clearly from the properties of $B$. Note that 
$B([\g_\ze, \g_\ze],\g_\un) = B([\g_\un, \g_\un],\g_\un)=0$.

For (2), fix an orthonormal basis $\{X_\ze^1, \dots, X_\ze^m\}$ of $\g_\ze$ and a Darboux basis $\{X_\un^1, \dots, X_\un^n,\linebreak Y_\un^1, \dots, Y_\un^n\}$ of $\g_\un$. Let $\{\alpha_1,\dots,\alpha_m\}$ and $\{\beta_1,\dots,\beta_n,\gamma_1,\dots,\gamma_n\}$ be their dual basis, respectively. Then for all $X,Y\in \g, \ i=1,...,m, \ l=1,...,n$ we have:
\begin{eqnarray*}
& &\adpo(I)(\alpha_i)(X,Y)\\ &&=\left( \sum^{m}_{j=1}\limits\io_{X_\ze^j}(I)\wedge\io_{X_\ze^j}(\alpha_i)-\sum^{n}_{k=1}\limits\left(\io_{X_\un^k}(I)\wedge\io_{Y_\un^k}(\alpha_i)-\io_{Y_\un^k}(I)\wedge\io_{X_\un^k}(\alpha_i)\right)\right)(X,Y)
	\\ &&=\left( \sum^{m}_{j=1}\io_{X_\ze^j}(I)\wedge\io_{X_\ze^j}(\alpha_i)\right)(X,Y)=\left(\io_{X_\ze^i}(I)\wedge\io_{X_\ze^i}(\alpha_i)\right)(X,Y)
	\\ &&= B(X_\ze^i,[X,Y]) = \alpha_i([X,Y]) = -d(\alpha_i)(X,Y),
\end{eqnarray*}

\begin{eqnarray*}
& &\adpo(I)(\beta_l)(X,Y)\\ &&=\left( \sum^{m}_{j=1}\limits\io_{X_\ze^j}(I)\wedge\io_{X_\ze^j}(\beta_l)-\sum^{n}_{k=1}\limits\left(\io_{X_\un^k}(I)\wedge\io_{Y_\un^k}(\beta_l)-\io_{Y_\un^k}(I)\wedge\io_{X_\un^k}(\beta_l)\right)\right)(X,Y)
	\\ &&= \left(\sum^{n}_{k=1}\io_{Y_\un^k}(I)\wedge\io_{X_\un^k}(\beta_l)\right)(X,Y)= \left(\io_{Y_\un^l}(I)\wedge\io_{X_\un^l}(\beta_l)\right)(X,Y)
	\\ &&= -\io_{Y_\un^l}(I)(X,Y) = -B(Y_\un^l,[X,Y]) = \beta_l([X,Y]) = -d(\beta_l)(X,Y).
\end{eqnarray*}


Similarly, $\adpo(I)(\gamma_l) = -d (\gamma_l)$ for $1 \leq l \leq n$.
Therefore, $d = -\adpo(I)$.
	
  Moreover, $\adpo(\{I,I\}) = [\adpo(I),\adpo(I)]=[d,d] = 2d^2 = 0$.
  Therefore, for all $1\leq i\leq m$, $1\leq j,k\leq n$ one has
  $\{\alpha_i,\{I,I\}\}$ = $\{\beta_j,\{I,I\}\}$ = $\{\gamma_k,\{I,I\}\} = 0$.
  Those imply $\iota_X\left(\{I,I\}\right) = 0$ for all $X\in\g$ and hence, we
  obtain $\{I,I\} = 0$.

  \end{proof}

Conversely, let $\g$ be a quadratic $\ZZ_2$-graded vector space equipped with
a bilinear form $B$ and $I$ be an element in $ \Es^{(3,\zero)}(\g)$. Define $d
= -\adpo(I)$ then $d\in \Ds^1_\ze(\Es(\g))$. Therefore, $d^2 = 0$ if and only
if $\{I,I\}=0$. Let $F$ be the structure in $\g$ corresponding to $d$ by the
isomorphism $D$ in Lemma \ref{3.1.15}, one has

\begin{prop}\label{3.1.18}
$F$ becomes a Lie superalgebra structure if and only if $\{I,I\}=0$. In this case, with the notation $[X,Y]: = F(X,Y)$ one has:
\[I(X,Y,Z) = B([X,Y],Z),\ \forall\ X,Y,Z\in\g.\]
Moreover, the bilinear form $B$ is invariant.
\end{prop}
\begin{proof}
We need to prove that if $F$ is a Lie superalgebra structure then $I(X,Y,Z) = B([X,Y],Z)$, for all $ X,Y,Z\in\g$. Indeed, let $\{X_\ze^1, \dots, X_\ze^m\}$ be an orthonormal basis of $\g_\ze$ and $\{X_\un^1, \dots, X_\un^n, Y_\un^1, \dots, Y_\un^n\}$ be a Darboux basis of $\g_\un$ then one has
\[d=-\adpo(I) = -\sum^{m}_{j=1}\io_{X_\ze^j}(I)\wedge\io_{X_\ze^j} + \sum^{n}_{k=1}\io_{X_\un^k}(I)\wedge\io_{Y_\un^k} - \sum^{n}_{k=1}\io_{Y_\un^k}(I)\wedge\io_{X_\un^k}.\]
It implies that
\[ F= \sum^{m}_{j=1}\io_{X_\ze^j}(I)\otimes{X_\ze^j} + \sum^{n}_{k=1}\io_{X_\un^k}(I)\otimes{Y_\un^k} - \sum^{n}_{k=1}\io_{Y_\un^k}(I)\otimes{X_\un^k}.\]
Therefore, for all $i$ we obtain
\[B([X,Y], {X_\ze^i}) = \io_{X_\ze^i}(I)(X,Y) = I(X_\ze^i,X,Y) = I(X,Y,X_\ze^i),\]
\[B([X,Y], {X_\un^i}) = -\io_{X_\un^i}(I)(X,Y) = -I(X_\un^i,X,Y) = I(X,Y,X_\un^i),\]
\[\text{ and }B([X,Y], {Y_\un^i}) = -\io_{Y_\un^i}(I)(X,Y) = -I(Y_\un^i,X,Y) = I(X,Y,Y_\un^i).\]
These show that $I(X,Y,Z) = B([X,Y],Z)$, for all $ X,Y,Z\in\g$. Since $I$ is super-antisymmetric and $B$ is supersymmetric, then one can show that $B$ is invariant.
\end{proof}

The two previous propositions show that on a quadratic $\ZZ_2$-graded vector
space $(\g,B)$, quadratic Lie superalgebra structures with the same $B$ are in
one to one correspondence with elements $I\in \Es^{(3,\ze)}(\g)$ satisfying
$\{I,I\}=0$ and such that the  super-exterior  differential of $\Es(\g)$ is
$d=-\adpo(I)$. This correspondence provides an approach to the theory of
quadratic Lie superalgebras through $I$.

\begin{defn}
Given a quadratic Lie superalgebra $(\g,B)$. The element $I$ defined as above is also an invariant of $\g$ since $\Lc_X(I) = 0$, for all $ X\in\g$ where $\Lc_X=D(\ad_\g(X))$ is the Lie super-derivation of $\g$. Therefore, $I$ is called the {\em associated invariant} of $\g$.
\end{defn}

The following Lemma is a simple, yet interesting result.

\begin{lem}\label{3.1.20}
Let $(\g,B)$ be a quadratic Lie superalgebra and $I$ be its associated invariant. Then $\io_X(I) = 0$ if and only if $X\in \Zs(\g)$.
\end{lem}

\begin{proof}
Since $\io_X(I)(\g,\g) = B(X,[\g,\g])$ and $\Zs(\g) = [\g,\g]^\bot$ where
 $[\g,\g]^\bot$ denotes the orthogonal subspace of $[\g,\g]$. We have then 
 $\io_X(I) = 0$ if and only if $X\in \Zs(\g)$.
\end{proof}
\begin{defn}
  Let $(\g,B)$ and $(\g',B')$ be two quadratic Lie superalgebras.  We say
  that $(\g,B)$ and $(\g',B')$ are {\em isometrically isomorphic} (or
  {\em i-isomorphic}) if there exists a Lie superalgebra isomorphism $A$
  from $\g$ onto $\g'$ satisfying \[ B'(A(X), A(Y)) = B(X,Y), \
  \forall \ X, Y \in \g.\] In other words, $A$ is an i-isomorphism if
  it is a (necessarily even) Lie superalgebra isomorphism and an isometry. We write $\g \iiso
  \g'$.
\end{defn}

Note that two isomorphic quadratic Lie superalgebras $(\g,B)$ and
$(\g',B')$ are not necessarily i-isomorphic by the example below:
\begin{ex}
  Let $\g = \ospk(1,2)$ and $B$ its Killing form. Recall that $\g_\ze
  = \ok(3)$. Consider another bilinear form $B' = \lambda B$, $\lambda
  \in \CC$, $\lambda \neq 0$. In this case, $(\g,B)$ and $(\g,\lambda
  B)$ can not be i-isomorphic if $\lambda \neq 1$ since
  $(\g_{\zero},B)$ and $(\g_{\zero},\lambda B)$ are not i-isomorphic.
\end{ex}

\section{The dup-number of quadratic Lie superalgebras}

Let $(\g,B)$ be a quadratic Lie superalgebra and $I$ be its associated
invariant. Then by Proposition \ref{3.1.17} we have a decomposition
\[I=I_0+I_1 \] where $I_0\in \Alt^3 (\g_\ze)$ and $I_1\in \Alt^1
(\g_\ze)\otimes \sym^2(\g_\un)$. Since $\{I,I\}=0$, then $\{I_0,I_0\}=0$. It
means that $\g_\ze$ is a quadratic Lie algebra with the associated 3-form
$I_0$, a rather obvious result. It is easy to see that $\g_\ze$ is Abelian
(resp. $[\g_\un,\g_\un] = \{0\}$) if and only if $I_0 = 0$ (resp. $I_1 =0$).
These cases will be fully studied in the sequel. Define the following
subspaces of $\g^*$:

\begin{eqnarray*}
  \Vs_I &=& \{\alpha \in \g^* \mid\ \alpha \wedge I = 0 \},\\
  \Vs_{I_0} &=& \{\alpha \in \g_\ze^* \mid\ \alpha \wedge I_0 = 0 \},\\
  \Vs_{I_1} &=& \{\alpha \in \g_\ze^* \mid\ \alpha \wedge I_1 = 0 \}.
\end{eqnarray*}

\begin{lem}\label{3.2.1}
Let $\g$ be a non-Abelian quadratic Lie superalgebra then one has
\begin{enumerate}
\item $\dim(\Vs_I)\in \{0,
  1, 3 \}$,
\item $\dim(\Vs_I)=3$ if and only if $I_1=0$, $\g_\ze$ is non-Abelian and $I_0$ is decomposable in $\Alt^3 (\gO)$.
\end{enumerate}
\end{lem}
\begin{proof}
Let $\alpha = \alpha_0 + \alpha_1\in\g_\ze^*\oplus\g_\un^*$ then one has
\[\alpha\wedge I = \alpha_0\wedge I_0 + \alpha_0\wedge I_1 + \alpha_1\wedge I_0 +\alpha_1\wedge I_1,
\]

where $\alpha_0\wedge I_0\in \Alt^4 (\g_ {\zero})$, $\alpha_0\wedge I_1\in
\Alt^2 (\g_{\zero}) \otimes \sym^2(\gI)$, $\alpha_1\wedge I_0\in \Alt^3
(\gO)\otimes \sym^1(\gI)$ and $\alpha_1\wedge I_1\in \Alt^1 (\gO)\otimes
\sym^3(\gI)$.

Hence, $\alpha\wedge I = 0$ if and only if $ \alpha_1 = 0$ and $\alpha_0\wedge
I_0 = \alpha_0\wedge I_1=0$. It means that $\Vs_I = \Vs_{I_0}\cap\Vs_{I_1}$.
If $I_0\neq 0$ then $\dim(\Vs_{I_0})\in \{0,1,3\}$ and if $I_1\neq 0$ then
$\dim(\Vs_{I_1})\in \{0,1\}$. Therefore, $\dim(\Vs_I)\in \{0, 1, 3 \}$ and
$\dim(\Vs_I)=3$ if and only if $I_1=0$ and $\dim(\Vs_{I_0})=3$. \end{proof}

The previous Lemma allows us to introduce the notion of dup-{\em number} for
quadratic Lie superalgebras as we did for quadratic Lie algebras.

\begin{defn} Let $(\g,B)$ be a non-Abelian quadratic Lie superalgebra
  and $I$ be its associated invariant. The {\em $\dup$-number}
  $\dup(\g)$ is defined by
    \[ \dup(\g) = \dim(\Vs_I).\]
\end{defn}

Given a subspace $W$ of $\g$, if $W$ is {\em non-degenerate} (with
respect to the bilinear form $B$), that is, if the restriction of $B$
on $W\times W$ is non-degenerate, then the orthogonal subspace
$W^\bot$ of $W$ is also non-degenerate. In this case, we use the
notation
\[\g = W \oplusp W^\bot.
\]

The decomposition result below is a generalization of the quadratic
Lie algebra case. Its proof can be found in \cite{PU07} and
\cite{DPU10}.

\begin{prop} \label{3.2.3}  
  Let $(\g,B)$ be a non-Abelian quadratic Lie superalgebra. Then there
  are a central ideal $\zk$ and an ideal $\lk \neq \{0\}$ such
  that:

\begin{enumerate}

\item $\g = \zk \oplusp \lk$ where $\left( \zk, B|_{\zk \times \zk} \right)$ and $\left(\lk,
    B|_{\lk \times \lk} \right)$ are quadratic Lie superalgebras. Moreover,
  $\lk$ is non-Abelian.

\item The center $\Zs(\lk)$ is totally isotropic, i.e. $\Zs(\lk)
  \subset [\lk, \lk]$.

\item Let $\g'$ be a quadratic Lie superalgebra and $A : \g \to \g'$ be a
  Lie superalgebra isomorphism. Then \[ \g' = \zk' \oplusp \lk'\] where
  $\zk' = A(\zk)$ is central, $\lk' = A(\zk)^\perp$, $\Zs(\lk')$ is
  totally isotropic and $\lk$ and $\lk'$ are isomorphic. Moreover if
  $A$ is an i-isomorphism, then $\lk$ and $\lk'$ are i-isomorphic.

\end{enumerate}

\end{prop}

The Lemma below shows that the previous decomposition has a good behavior with
respect to the $\dup$-number.

\begin{lem}\label{3.2.4}
  Let $\g$ be a non-Abelian quadratic Lie superalgebra. Write $\g =
  \zk \oplusp \lk$ as in Proposition \ref{3.2.3} then $\dup(\g) =
  \dup(\lk)$.
\end{lem}

\begin{proof}
  Since $[\zk,\g] = \{0\}$ then $I\in\Es^{(3,\zero)}(\lk)$. Let
  $\alpha\in \g^*$ such that $\alpha\wedge I = 0$, we show that
  $\alpha\in\lk^*$. Assume that $\alpha = \alpha_1+\alpha_2$, where
  $\alpha_1\in\zk^*$ and $\alpha_2\in\lk^*$. Since $\alpha\wedge I =
  0$, $\alpha_1\wedge I \in\Es(\zk)\otimes\Es(\lk)$ and
  $\alpha_2\wedge I \in\Es(\lk)$ then one has $\alpha_1\wedge I =
  0$. Therefore, $\alpha_1 = 0$ since $I$ is nonzero in
  $\Es^{(3,\zero)}(\lk)$. That means $\alpha \in \lk^*$ and then
  $\dup(\g) = \dup(\lk)$.
\end{proof}

Clearly, $\zk=\{0\}$ if and only if $\Zs(\g)$ is totally isotropic. By the
above Lemma, it is enough to restrict our study on the $\dup$-number of non-
Abelian quadratic Lie superalgebras with totally isotropic center.

\begin{defn}\label{3.2.5}
A quadratic Lie superalgebra $\g$ is {\em reduced} if it satisfies:

\begin{enumerate}

\item $\g \neq \{0\}$

\item $\Zs(\g)$ is totally isotropic.

\end{enumerate}

\end{defn}

Notice that a reduced quadratic Lie superalgebra is necessarily non-Abelian.

\begin{defn} 

Let $\g$ be a non-Abelian quadratic Lie superalgebra. We say that:

\begin{enumerate}

\item $\g$ is an {\em ordinary} quadratic Lie superalgebra if $\dup(\g) =
  0$.

\item $\g$ is a {\em singular} quadratic Lie superalgebra if $\dup(\g) \geq
  1$. 

\begin{itemize}

\item[(i)] $\g$ is a {\em singular} quadratic Lie superalgebra of {\em type
    $\Sb_1$} if $\dup(\g) = 1$.

\item[(ii)] $\g$ is a {\em singular} quadratic Lie superalgebra of {\em type
    $\Sb_3$} if $\dup(\g) = 3$.

\end{itemize}

\end{enumerate}

\end{defn}

By Lemma \ref{3.2.1}, if $\g$ is a singular quadratic Lie superalgebra
of type $\Sb_3$ then $I= I_0$ is decomposable in $\Alt^3 (\gO)$. One
has $I(\g_\ze,\g_\un,\g_\un) = B([\g_\ze,\g_\un],\g_\un) = 0$. It
implies $[\g_\un,\g] = \{0\}$ since $B$ is non-degenerate. Hence in
this case, $\g_{\un}$ is a central ideal, $\g_{\zero}$ is a singular
quadratic Lie algebra of type $\Sb_3$ and then the classification is
known in \cite{PU07}. Therefore, we are mainly interested in singular
quadratic Lie superalgebras of type $\Sb_1$.

Before proceeding, we give other simple properties of singular quadratic Lie superalgebras:

\begin{prop}\label{3.2.7}
Let $(\g,B)$ be a singular quadratic Lie superalgebra. If $\g_{\zero}$ is non-Abelian then $\g_{\zero}$ is a singular quadratic Lie algebra.
\end{prop}

\begin{proof}
By the proof of Lemma \ref{3.2.1}, one has $\Vs_I = \Vs_{I_0}\cap\Vs_{I_1}$. Therefore, $\dim( \Vs_{I_0})\geq 1$. It means that $\g_{\zero}$ is a singular quadratic Lie algebra.
\end{proof}

Given $(\g,B)$ a singular quadratic Lie superalgebra of type $\Sb_1$. Fix
$\alpha\in \Vs_I$ and choose $\Omega_0\in \Alt^2 (\gO)$, $\Omega_1\in
\sym^2(\gI)$ such that $I = \alpha\wedge\Omega_0+\alpha\otimes\Omega_1$. Then
one has  \[ \{I,I\} = \{\alpha\wedge\Omega_0,\alpha\wedge\Omega_0\}+
2\{\alpha\wedge\Omega_0,\alpha\}\otimes\Omega_1+ \{\alpha,\alpha\}\otimes
\Omega_1\Omega_1.\]   By the equality $\{I,I\} = 0$, one has
$\{\alpha\wedge\Omega_0,\alpha\wedge\Omega_0\} = 0$, $\{\alpha,\alpha\} = 0$
and $\{\alpha,\alpha\wedge\Omega_0\}=0$. These imply that $\{\alpha,I\} = 0$.
Hence, if we set $X_0 = \phi^{-1}(\alpha)$ then $X_0\in \Zs(\g)$ and
$B(X_0,X_0) = 0$ (Corollary \ref{3.1.13} and Lemma \ref{3.1.20}).

\begin{prop} \label{3.2.8}
Let $(\g,B)$ be a singular quadratic Lie superalgebra. If $\g$ is reduced then $\g_\ze$ is reduced.
\end{prop}
\begin{proof}

As above, if $\g$ is a singular quadratic Lie superalgebra of type $\Sb_3$ then $\g_\un$ is central. By $\g$ reduced and $\g_\un$ non-degenerate, $\g_\un$ must be zero and then the result follows. 

If $\g$ is a singular quadratic Lie superalgebra of type $\Sb_1$. Assume that $\g_{\ze}$ is not reduced, i.e. $\g_{\ze}= \zk \oplusp \lk$ where $\zk$ is a non-trivial central ideal of $\g_{\ze}$, there is $X\in \zk$ such that $B(X,X)=1$. Since $\g$ is singular of type $\Sb_1$ then $\g_{\ze}$ is also singular. Hence, the element $X_0$ defined as above must be in $\lk$ and $I_0 = \alpha\wedge \Omega_0\in \Alt^3 (\lk)$ (see \cite{DPU10} for details). We also have $B(X,X_0) = 0$.

Let $\beta = \phi(X)$ so $\io_X(I) = \{\beta, I\} = \{\beta, \alpha\wedge\Omega_0+\alpha\otimes\Omega_1\} = 0$. That means $X\in \Zs(\g)$. This is a contradiction since $\g$ is reduced. Hence $\g_{\ze}$ must be reduced.
\end{proof}

\begin{lem} \label{3.2.9} Let $\g_1$ and $\g_2$ be non-Abelian quadratic
  Lie superalgebras. Then $\g_1 \oplusp \g_2$ is an ordinary quadratic Lie
  algebra.
\end{lem}

\begin{proof}
Set $\g = \g_1 \oplusp \g_2$. Denote by $I$, $I_1$ and $I_2$ their non-trivial associated invariants, respectively. One has $\Es(\g) = \Es(\g_1)\otimes \Es(\g_2)$, $ \Es^k(\g) = \bigoplus_{r+s=k}\limits\Es^r(\g_1)\otimes \Es^s(\g_2)$ and $I = I_1 + I_2$ where $I_1\in \Es^3(\g_1)$, $I_2\in \Es^3(\g_2)$. Therefore, if $\alpha = \alpha_1+\alpha_2\in \g_1^*\oplus\g_2^*$ such that $\alpha\wedge I = 0$ then $\alpha_1 = \alpha_2 = 0$. 
\end{proof}

\begin{defn}
A quadratic Lie superalgebra $\g$ is {\em indecomposable} if $\g =
\g_1 \oplusp \g_2$, with $\g_1$ and $\g_2$ ideals of $\g$, then $\g_1$
or $\g_2 = \{0\}$.
\end{defn}

 The following result shows that indecomposable and reduced notions
 are equivalent for singular quadratic Lie superalgebras.

\begin{prop}\label{3.2.11}
  Let $\g$ be a singular quadratic Lie superalgebra. Then $\g$ is reduced
  if and only if $\g$ is indecomposable.
\end{prop}
\begin{proof}
If $\g$ is indecomposable then it is obvious that $\g$ is reduced. If
$\g$ is reduced, assume that $\g = \g_1 \oplusp \g_2$, with $\g_1$ and
$\g_2$ ideals of $\g$, then $\Zs(\g_i)\subset[\g_i,\g_i]$ for $i =
1,2$. Therefore, $\g_i$ is reduced or $\g_i = \{0\}$. If $\g_1$ and
$\g_2$ are both reduced, by Lemma \ref{3.2.9}, then $\g$ is
ordinary. Hence $\g_1$ or $\g_2 = \{0\}$.
\end{proof}

\section{Elementary quadratic Lie superalgebras}

In this section, we consider the first non-trivial case of singular quadratic Lie superalgebras: elementary quadratic Lie superalgebras. We begin with the following definition.
\begin{defn}
Let $\g$ be a quadratic Lie superalgebra and $I$ be its associated invariant. We say that $\g$ is an {\em elementary} quadratic Lie superalgebra if $I$ is decomposable.
\end{defn}

Keep notations as in Section 2. If $I=I_0+I_1$ is decomposable, where $I_0\in
\Alt^3 (\g_\ze)$ and $I_1\in \Alt{}^1 (\g_\ze)\otimes \sym^2(\gI)$ then it is
obvious that $I_0$ or $I_1$ is zero. The case $I_1=0$, i.e. $I$ decomposable
in $\Alt^3 (\gO)$, corresponds to singular quadratic Lie superalgebras of
type $\Sb_3$ and then there is nothing to do. Now we assume $I$ is a nonzero
decomposable element in $\Alt^1 (\gO)\otimes \sym^2(\g_{\un})$ then $I$
can be written by:

\[I = \alpha\otimes pq
\]
where $\alpha\in \g_\ze^*$ and $p,q\in \g_\un^*$. It is clear that $\g$ is a singular quadratic Lie superalgebra of type $\Sb_1$.

\begin{lem}\label{3.3.2}
Let $\g$ be a reduced elementary quadratic Lie superalgebra having $I = \alpha\otimes pq$ where $\alpha\in \g_\ze^*$ and $p,q\in \g_\un^*$. Set $X_\ze = \phi^{-1}(\alpha)$ then one has:
\begin{enumerate}

\item $\dim(\g_\ze)=2$ and $\g_\ze\cap \Zs(\g) = \CC X_\ze$.
\item Let $X_\un =\phi^{-1} (p)$, $Y_\un =\phi^{-1} (q)$ and $U=\spa\{X_\un, Y_\un\}$ then
\begin{itemize}

\item[(i)] $\dim (\g_\un)=2$ if $\dim (U)=1$ or $U$ is non-degenerate.

\item[(ii)] $\dim (\g_\un)=4$ if $U$ is totally isotropic.
\end{itemize}
\end{enumerate}
\end{lem}

\begin{proof}\hfill
\begin{enumerate}
\item Let $\beta$ be an element in $\g_\ze^*$. It is easy to see that $\{\beta , \alpha\} = 0$ if and only if $\{\beta , I\} = 0$, equivalently $\phi^{-1}(\beta) \in \Zs(\g)$. Therefore, $(\phi^{-1}(\alpha))^\bot \cap \g_\ze \subset \Zs(\g)$. It means that $\dim(\g_\ze) \leq 2$ since $\g$ is reduced (see \cite{Bou59}). Moreover, $X_\ze = \phi^{-1}(\alpha)$ is isotropic then $\dim(\g_\ze) = 2$. If $\dim(\g_\ze\cap \Zs(\g)) = 2$ then $\g_\ze\subset \Zs(\g)$. Since $B$ is invariant we obtain $\g$ Abelian (a contradiction). Therefore, $\g_\ze\cap \Zs(\g) = \CC X_\ze$.
\item It is obvious that $\dim (\g_\un) \geq 2$. If $\dim (U)=1$ then $U$ is a totally isotropic subspace of $\g_\un$ so there exists a one-dimensional subspace $V$ of $\g_\un$ such that $B$ is non-degenerate on $U \oplus V$ (see \cite{Bou59}). Let $\g_\un = (U \oplus V)\oplusp W$ where $W = (U \oplus V)^\bot$ then for all $f \in \phi(W)$ one has:
\[ \{f, I\} = \{f,  \alpha\otimes pq \} = -\alpha\otimes (\{f, p\} q + p \{f, q\}) = 0.
\]
Therefore, $W \subset \Zs(\g)$. Since $B$ is non-degenerate on $W$ and $\g$ is reduced then $W = \{0\}$.

If $\dim (U)=2$ then $U$ is non-degenerate or totally isotropic. If $U$ is non-degenerate, let $\g_{\un} = U \oplusp W$ where $W = U ^\bot$. If $U$ is totally isotropic, let $\g_{\un} = (U \oplus V)\oplusp W$ where $W = (U \oplus V)^\bot$ in $\g_{\un}$ and $B$ is non-degenerate on $U \oplus V$. In the both cases, similarly as above, one has $W$ a non-degenerate central ideal so $W = \{0\}$. Therefore, $\dim (\g_\un) = \dim (U) = 2$ if $U$ is non-degenerate and  $\dim (\g_\un) = \dim (U\oplus V) = 4$ if $U$ is totally isotropic.
\end{enumerate}
\end{proof}

In the sequel, we obtain the classification result.

\begin{prop} \label{3.3.3}
Let $\g$ be a reduced elementary quadratic Lie superalgebra then $\g$
is i-isomorphic to one of the following Lie superalgebras:

\begin{enumerate}

\item $\g_i$ $(3\leq i\leq 6)$ the reduced singular quadratic Lie
  algebras of type $\Sb_3$ given in \cite{PU07}.

\item $\g_{4,1}^s = (\CC X_{\zero} \oplus \CC Y_{\zero})\oplus (\CC
  X_{\un} \oplus \CC Z_{\un})$ where $\g_\ze=\spa\{X_\ze,Y_\ze\}$,
  $\g_\un=\spa\{X_\un,Z_\un\}$, the bilinear form $B$ is defined
  by \[B(X_{\zero}, Y_{\zero})=B(X_{\un}, Z_{\un})=1,\] the other are
  zero and the Lie super-bracket is given
  by \[[Z_{\un},Z_{\un}]=-2X_{\zero},\ \ [Y_{\zero},Z_{\un}]=-2X_{\un},\]
  the other are trivial.

\item $\g_{4,2}^s = (\CC X_{\zero} \oplus \CC Y_{\zero})\oplus (\CC
  X_{\un} \oplus \CC Y_{\un})$ where $\g_\ze=\spa\{X_\ze,Y_\ze\}$,
  $\g_\un=\spa\{X_\un,Y_\un\}$, the bilinear form $B$ is defined
  by \[B(X_{\zero}, Y_{\zero})=B(X_{\un}, Y_{\un})=1,\] the other are
  zero and the Lie super-bracket is given
  by \[[X_{\un},Y_{\un}]=X_{\zero},\ \ [Y_{\zero},X_{\un}]=X_{\un},
  \ \ [Y_{\zero},Y_{\un}]=-Y_{\un},\] the other are trivial.  

\item $\g_6^s = (\CC X_{\zero} \oplus \CC Y_{\zero})\oplus (\CC
  X_{\un} \oplus \CC Y_{\un}\oplus \CC Z_{\un}\oplus \CC T_{\un})$
  where $\g_\ze=\spa\{X_\ze,Y_\ze\}$, $\g_\un=\spa\{X_\un,Y_\un,
  Z_\un,T_\un\}$, the bilinear form $B$ is defined by \[B(X_{\zero},
  Y_{\zero})=B(X_{\un}, Z_{\un})=B(Y_{\un}, T_{\un})=1,\] the other
  are zero and the Lie super-bracket is given
  by \[[Z_{\un},T_{\un}]=-X_{\zero},
  \ \ [Y_{\zero},Z_{\un}]=-Y_{\un},\ \ [Y_{\zero},T_{\un}]=-X_{\un},\]
  the other are trivial.

\end{enumerate}
\end{prop}
\begin{proof}\hfill
\begin{enumerate}

\item This statement corresponds to the case where $I$ is a decomposable 3-form 
in $\Alt^3 (\gO)$. Therefore, the result is obvious.

Assume that $I = \alpha\otimes pq \in \Alt^1 (\gO)\otimes
\sym^2(\gI)$. By the previous lemma, we can write $\g_{\zero} = \CC X_{\zero}
\oplus \CC Y_{\zero}$ where $X_{\zero} = \phi^{-1}(\alpha)$, $B(X_{\zero},
X_{\zero}) = B(Y_{\zero}, Y_{\zero}) = 0$, $B(X_{\zero}, Y_{\zero}) = 1$. Let
$X_\un =\phi^{-1} (p),\ Y_\un =\phi^{-1} (q)$ and $U=\spa\{X_\un, Y_\un\}$.

\item If $\dim (U) = 1$ then $Y_\un = kX_\un$ with some nonzero $k \in
  \CC$. Therefore, $q = kp$ and $I = k\alpha\otimes p^2$. Replacing
  $X_{\zero}$ by $kX_{\zero}$ and $Y_\ze$ by $\frac{1}{k}Y_\ze$, we
  can assume that $k=1$. Let $Z_\un$ be an element in $\g_{\un}$ such
  that $B(X_{\un}, Z_{\un}) = 1$.

Now, let $X\in\g_\ze,Y,Z\in\g_\un$. By using (1.7) and (1.8) of \cite{BP89}, one has:
\[ B(X,[Y,Z]) = -2\alpha (X)p(Y)p(Z) = -2B(X_{\zero} , X) B(X_{\un}, Y) B(X_{\un}, Z).
\]
Since $B|_{\g_\ze\times\g_\ze}$ is non-degenerate then:
\[ [Y, Z] = -2B(X_{\un}, Y)B(X_{\un}, Z)X_{\zero} ,\ \forall\ Y, Z \in \g_{\un}.
\]
Similarly and by the invariance of $B$, we also obtain:
\[ [X, Y] = -2B(X_{\zero} , X) B(X_{\un}, Y) X_{\un},\ \forall\ X \in \g_{\zero} ,\ Y\in \g_\un
\]
and (2) follows.

\item If $\dim (U) = 2$ and $U$ is non-degenerate then $B(X_{\un}, Y_{\un}) = a \neq 0$. Replacing $X_{\un}$ by $\frac{1}{a} X_{\un}$, $\ X_{\zero}$ by $aX_{\zero}$ and $Y_{\zero}$ by $\frac{1}{a}Y_{\zero}$, we can assume that $a=1$. Then one has $\g_\ze=\spa\{X_{\zero},Y_{\zero}\}$, $\g_\un=\spa\{X_{\un}, Y_{\un}\}$, $B(X_{\zero}, Y_{\zero})=B(X_{\un}, Y_{\un})=1$ and $I = \alpha\otimes pq$. 
Let $X\in\g_\ze,\ Y,Z\in\g_\un$, we have:
\[ B(X,[Y,Z]) = I(X,Y,Z) = -\alpha (X)(p(Y)q(Z) + p(Z)q(Y)).
\]
Therefore, the Lie super-bracket is defined:
\[ [Y, Z] = -(B(X_{\un}, Y)B(Y_{\un}, Z) + B(X_{\un}, Z)B(Y_{\un}, Y))X_{\zero} ,\ \forall\ Y, Z \in \g_{\un},
\]
\[ [X, Y] = -B(X_{\zero}, X)(B(X_{\un}, Y)Y_{\un} + B(Y_{\un}, Y)X_{\un}) ,\ \forall\  X \in \g_{\zero} ,\ Y\in \g_\un
\]
and (3) follows.

\item If $\dim (U) = 2$ and $U$ is totally isotropic: let $V =
  \spa\{Z_{\un}, T_{\un}\}$ be a 2-dimensional totally isotropic
  subspace of $\g_{\un}$ such that $\g_{\un}=U\oplusp V$ and
  $B(X_{\un}, Z_{\un})=B(Y_{\un}, T_{\un})=1$.  If
  $X\in\g_\ze,Y,Z\in\g_\un$ then:
\[ B(X,[Y,Z]) = I(X,Y,Z) = -\alpha (X)(p(Y)q(Z) + p(Z)q(Y)).
\]
We obtain the Lie super-bracket as follows:
\[ [Y, Z] = -(B(X_{\un}, Y)B(Y_{\un}, Z) + B(X_{\un}, Z)B(Y_{\un}, Y))X_{\zero} ,\ \forall\ Y, Z \in \g_{\un},
\]
\[ [X, Y] = -B(X_{\zero}, X)(B(X_{\un}, Y)Y_{\un} + B(Y_{\un}, Y)X_{\un}) ,\ \forall\  X \in \g_{\zero} ,\ Y\in \g_\un.
\]
Thus, $[Z_{\un},T_{\un}]=-X_{\zero}$, $[Y_{\zero},Z_{\un}]=-Y_{\un}$, $[Y_{\zero},T_{\un}]=-X_{\un}$.
\end{enumerate}
\end{proof}

\section{Quadratic Lie superalgebras with 2-dimensional even part}

This section is devoted to study another particular case of singular quadratic Lie
superalgebras: quadratic Lie superalgebras with 2-dimensional even part. As we
shall see, they are can be seen as a symplectic version of solvable singular
quadratic Lie algebras. The first result classifies these algebras with respect to the $\dup$-number.

\begin{prop}\label{3.4.1}
Let $\g$ be a non-Abelian quadratic Lie superalgebra with $\dim(\g_\ze) =2$. Then $\g$ is a singular quadratic Lie superalgebra of type $\Sb_1$.
\end{prop}

\begin{proof}
Let $I$ be the associated invariant of $\g$. By a remark in \cite{PU07}, every non-Abelian quadratic Lie algebra must have the dimension more than 2 so $\g_{\zero}$ is Abelian and as a consequence, $I \in \Alt^1 (\gO) \otimes \sym^2(\gI)$. We choose a basis $\{X_{\zero}, Y_{\zero}\}$ of $\g_{\zero}$ such that $B(X_{\zero}, X_{\zero}) = B(Y_{\zero}, Y_{\zero}) = 0$ and $B(X_{\zero}, Y_{\zero}) = 1$. Let $\alpha = \phi(X_{\zero})$, $\beta = \phi(Y_{\zero})$ and we can assume that
\[ I = \alpha \otimes \Omega_1 + \beta\otimes\Omega_2
\]
where $\Omega_{1}, \Omega_{2} \in \sym^2(\gI)$. Then one has:
\[ \{I, I\} = 2(\Omega_{1} \Omega_{2} + \alpha \wedge \beta \otimes \{\Omega_{1}, \Omega_{2}\}).
\]
Therefore, $\{I, I\} = 0$ implies that $\Omega_{1} \Omega_{2} = 0$. So $\Omega_{1} = 0$ or $\Omega_{2} = 0$. It means that $\g$ is a singular quadratic Lie superalgebra of type $\Sb_1$.
\end{proof}

\begin{prop}\label{3.4.2}
Let $\g$ be a singular quadratic Lie superalgebra with Abelian even part. If $\g$ is reduced then $\dim(\g_{\zero})=2$.
\end{prop}
\begin{proof}
Let $I$ be the associated invariant of $\g$. Since $\g$ has the Abelian even part one has $I  \in \Alt^1 (\gO )\otimes \sym^2(\gI)$. Moreover $\g$ is singular then
\[ I = \alpha \otimes \Omega
\]
where $\alpha \in \g_\ze^*,\ \Omega \in \sym^2(\g_\un)$. The proof
follows exactly Lemma \ref{3.3.2}. Let $\beta \in \g_\ze^*$ then
$\{\beta , \alpha\} = 0$ if and only if $\{\beta , I\} = 0$,
equivalently $\phi^{-1}(\beta) \in \Zs(\g)$. Therefore,
$(\phi^{-1}(\alpha))^\bot \cap \g_\ze \subset \Zs(\g)$. It means that
$\dim(\g_\ze) = 2$ since $\g$ is reduced and $\phi^{-1}(\alpha)$ is
isotropic in $\Zs(\g)$.
\end{proof}

Now, let $\g$ be a non-Abelian quadratic Lie superalgebra with 2-dimensional
even part. By Proposition \ref{3.4.1}, $\g$ is singular of type $\Sb_1$. Fix
$\alpha \in\Vs_I$ and choose $\Omega \in \sym^2(\g)$ such that  \[I = \alpha
\otimes \Omega. \] We define $C : \g_\un \to \g_\un$ by $B(C(X), Y) =
\Omega(X,Y)$, for all $X,Y\in\g_\un$ and let $X_{\zero}=\phi^{-1}(\alpha)$.

\begin{lem} \label{3.4.3} The following assertions are equivalent:

\begin{enumerate}

\item $\{I,I\} = 0$,

\item $\{\alpha, \alpha\} = 0$,

\item $B(X_{\zero} , X_{\zero}) =0$.

\end{enumerate}

In this case, one has $X_{\zero} \in \Zs(\g)$.
\end{lem}

\begin{proof}
  It is easy to see that: 
  \[ \{I, I\} = 0 \Leftrightarrow \{\alpha,
  \alpha \} \otimes \Omega^2 = 0.
  \] 
  Therefore the assertions are equivalent. Moreover, since $\{\alpha, I\}=0$ one has $X_{\zero} \in \Zs(\g)$.
\end{proof}

We keep the notations as in the previous sections. Then there exists an isotropic element $Y_{\zero}\in \g_{\zero}$ such that $B(X_{\zero}, Y_{\zero})=1$ and one has the following proposition:
\begin{prop} \label{3.4.4}
\hfill
\begin{enumerate}

\item The map $C$ is skew-symmetric (with respect to $B$), that is \[B(C(X),Y) = -B(X,C(Y))\] for all $X$, $Y \in \g_\un$.

\smallskip

\item $[X,Y] = B(C(X),Y) X_\ze$, for all $X, Y \in \g_\un$ and $C = \ad(Y_\ze)|_{\g_\un}$.

\smallskip

\item $\Zs(\g) = \ker(C) \oplus \CC X_\ze$ and $[\g,\g] =
 \im(C)\oplus \CC X_\ze$. Therefore, $\g$ is reduced if and only if $\ker(C)\subset \im(C)$.

\smallskip

\item The Lie superalgebra $\g$ is solvable. Moreover, $\g$ is nilpotent
  if and only if $C$ is nilpotent.

\smallskip

\end{enumerate}

\end{prop}

\begin{proof} \hfill

\begin{enumerate}

\item For all $X$, $Y \in \g_\un$, one has
\[B(C(X),Y) = \Omega(X,Y)=\Omega(Y,X)=B(C(Y),X)=-B(X,C(Y)).\]
\item Let $X\in \g_\ze,Y,Z\in\g_\un$ then
\[B(X,[Y,Z]) = (\alpha
  \otimes \Omega) (X,Y,Z) = \alpha(X)\Omega(Y,Z).\]
  Since  $\alpha(X) = B(X_{\zero},X)$ and $\Omega(Y,Z) = B(C(Y),Z)$ so one has
	\[B(X,[Y,Z]) = B(X_{\zero},X)B(C(Y),Z).\] The non-degeneracy of $B$ implies $[Y,Z] = B(C(Y),Z)X_\ze$. Set $X = Y_\ze$ then $B(Y_\ze,[Y,Z]) = B(C(Y),Z)$. By the invariance of $B$,  we obtain $[Y_\ze,Y] = C(Y)$.
	
  \item It follows from the assertion (2).
  \item $\g$ is solvable since $\g_\ze$ is solvable, or since $[[\g,\g],[\g,\g]]\subset\CC X_\ze$. If $\g$ is nilpotent then $C=\ad(Y_\ze)$ is nilpotent obviously. Conversely, if $C$ is nilpotent then it is easy to see that $\g$ is nilpotent since $(\ad(X))^k(\g) \subset \CC X_\ze\oplus \im(C^k)$ for all $X\in\g$. 
\end{enumerate}
\end{proof}

\begin{rem}
  The choice of $C$ is unique up to a nonzero scalar. Indeed, assume
  that $I = \alpha'\otimes \Omega'$ and $C'$ is the map associated to
  $\Omega'$. Since $\Zs(\g)\cap\g_\ze = \CC X_\ze$ and
  $\phi^{-1}(\alpha')\in \Zs(\g)$ one has $\alpha' = \lambda\alpha$
  for some nonzero $\lambda\in \CC$. Therefore,
  $\alpha\otimes(\Omega-\lambda\Omega')=0$. It means that $\Omega =
  \lambda\Omega'$ and then we get $C = \lambda C'$.
\end{rem}
\subsection{Double extension of a symplectic vector space}\hfill

Double extensions are a very useful method initiated by V. Kac to construct
quadratic Lie algebras (see \cite{Kac85} and \cite{MR85}). They are
generalized to many algebras endowed with a non-degenerate invariant
bilinear form, for example quadratic Lie superalgebras (see \cite{BB99} and
\cite{BBB}). In \cite{DPU10}, we consider a particular case that is the double
extension of a quadratic vector space by a skew-symmetric map. From this we
obtain the class of solvable singular quadratic Lie algebras. Here, we use
this notion in yet another context, replacing the quadratic vector space by
a symplectic vector space.

\begin{defn} \label{3.4.6}\hfill

\begin{enumerate}

\item Let $(\qk, B_\qk)$ be a symplectic vector space equipped with a symplectic bilinear form $B_\qk$ and $\cb : \qk
  \to \qk$ be a skew-symmetric map, that is,
  \[B_\qk(\cb(X),Y) = -B_\qk(X,\cb(Y)),\ \forall\ X, Y \in \qk.\]
  Let $(\tk = \spa \{ X_\ze, Y_\ze \}, B_\tk)$ be a 2-dimensional quadratic vector space with the symmetric bilinear form $B_\tk$
  defined by \[ B_\tk(X_\ze, X_\ze) = B_\tk(Y_\ze, Y_\ze) = 0, \ B_\tk(X_\ze,
  Y_\ze) = 1.\] 
  Consider the vector space \[\g = \tk \oplusp \qk\] equipped with the
  bilinear form $B = B_\tk + B_\qk $ and define a bracket on $\g$
  by \[ [\lambda X_\ze + \mu Y_\ze + X, \lambda' X_\ze + \mu' Y_\ze + Y] =
  \mu \cb (Y) - \mu' \cb (X) + B( \cb (X), Y) X_\ze, \] for all $X, Y
  \in \qk, \lambda, \mu, \lambda', \mu' \in \CC$. Then $(\g, B)$ is a
  quadratic solvable Lie superalgebra with $\g_\ze = \tk$ and $\g_\un =\qk$. We say that $\g$ is the {\em double extension of $\qk$ by $\cb$}.

\item Let $\g_i$ be double extensions of symplectic vector spaces
  $(\qk_i, B_i)$ by skew-symmetric maps $\cb_i \in \Lc(\qk_i)$, for $1
  \leq i \leq k$. The {\em amalgamated product} \[ \g = \g_1 \ta \g_2
  \ta \dots \ta \g_k\] is defined as follows:

\begin{itemize}

\item consider $(\qk,B)$ be the symplectic vector space with $\qk =
  \qk_1 \oplus \qk_2 \oplus \dots \oplus \qk_k$ and the bilinear form
  $B$ such that $B(\sum_{i = I}^k X_i, \sum_{i = I}^k Y_i) = \sum_{i =
    I}^k B_i(X_i, Y_i)$, for $X_i, Y_i \in \qk_i$, $1 \leq i \leq k$.

\item the skew-symmetric map $\cb \in \Lc(\qk)$ is defined by $\cb(
  \sum_{i = I}^k X_i) = \sum_{i = I}^k \cb_i(X_i)$, for
  $X_i \in \qk_i$, $1 \leq i \leq k$.

\end{itemize}
 
\medskip

Then $\g$ is the double extension of $\qk$ by $\cb$.

\end{enumerate}

\end{defn}

\begin{lem} \label{3.4.7} We keep the notation above. 

\begin{enumerate}

\item Let $\g$ be the double extension of $\qk$ by
  $\cb$. Then \[ [X, Y] = B(X_{\ze}, X) C (Y) - B(X_{\ze}, Y) C(X) +
  B(C(X), Y) X_\ze, \ \forall \ X, Y \in \g, \] where $C =
  \ad(Y_{\ze})$. Moreover, $X_\ze \in \Zs(\g)$ and $C|_\qk = \cb$.

\item Let $\g'$ be the double extension of $\qk$ by $\cpb =
  \lambda \cb$, $\lambda \in \CC$, $\lambda \neq 0$. Then
  $\g$ and $\g'$ are i-isomorphic.

\end{enumerate}

\end{lem}

\begin{proof}\hfill

\begin{enumerate}

\item This is a straightforward computation by Definition \ref{3.4.6}.

\item Write $\g = \tk \oplusp \qk = \g'$. Denote by $[\cdot, \cdot]'$
  the Lie super-bracket on $\g'$. Define $A : \g \to \g'$ by $A(X_\ze) =
  \lambda X_\ze$, $A(Y_\ze) = \dfrac{1}{\lambda} Y_\ze$ and $A|_\qk =
  \Id_\qk$. Then $A([Y_\ze,X]) = C(X) = [A(Y_\ze), A(X)]'$ and $A([X,Y]) =
  [A(X), A(Y)]'$, for all $X, Y \in \qk$. So $A$ is an i-isomorphism.

\end{enumerate}

\end{proof}

\begin{prop} \label{3.4.8} \hfill

\begin{enumerate}

\item Let $\g$ be a non-Abelian quadratic Lie superalgebra with 2-dimensional even part. Keep the notations as in Proposition \ref{3.4.4}. Then $\g$ is the double extension of $\qk =(\CC X_\ze \oplus \CC Y_\ze)^\perp = \g_\un$ by $\cb = \ad(Y_\ze)|_\qk$.

\smallskip
\item Let $\g$ be the double extension of a symplectic vector space
  $\qk$ by a map $\cb \neq 0$. Then $\g$ is a singular
  solvable quadratic Lie superalgebra with 2-dimensional even part. Moreover:

\begin{itemize}

\smallskip

\item[(i)] $\g$ is reduced if and only if $\ker(\cb)
  \subset \im(\cb)$.

\smallskip

\item[(ii)] $\g$ is nilpotent if and only if $\cb$ is nilpotent.
\end{itemize}

\item Let $(\g,B)$ be a quadratic Lie superalgebra. Let $\g'$ be the double
  extension of a symplectic vector space $(\qk',B')$ by a map
  $\cpb$. Let $A$ be an i-isomorphism of $\g'$ onto $\g$ and
  write $\qk = A(\qk')$. Then $\g$ is the double extension of $(\qk,
  B|_{\qk \times \qk})$ by the map $\cb = \overline{A} \
  \cpb \ \overline{A}^{-1}$ where $\overline{A} = A|_{\qk'}$.

\smallskip

\end{enumerate}

\end{prop}

\begin{proof} The assertions (1) and (2) follow Proposition \ref{3.4.4} and Lemma \ref{3.4.7}. For (3), since $A$ is i-isomorphic then $\g$ has also 2-dimensional even part. Write $\g' = (\CC X_\ze' \oplus \CC Y_\ze') \oplusp \qk'$. Let $X_\ze
  = A(X_\ze')$ and $Y_\ze = A (Y_\ze')$. Then $\g = (\CC X_\ze \oplus \CC Y_\ze)
  \oplusp \qk$ and one has: \[ [Y_\ze, X] = (A \cpb A^{-1}) (X), \
  \forall \ X \in \qk, \text{ and } \] \[ [X, Y] = B((A \cpb
  A^{-1}) (X), Y) X_\ze, \ \forall \ X, Y \in \qk.\] This proves
  the result.
\end{proof}
\begin{ex}
From the point of view of double extensions, for reduced elementary quadratic Lie superalgebras with 2-dimensional even part in Section 3 one has
\begin{enumerate}
	\item $\g_{4,1}^s$ is  the double extension of the 2-dimensional symplectic vector space $\qk=\CC^2$ by the map having matrix:
	\[\cb =  \begin{pmatrix} 0 & 1 \\ 0 & 0 \end{pmatrix}\]
	in a Darboux basis $\{E_1,E_2\}$ of $\qk$ where $B(E_1,E_2)=1$.
\item $\g_{4,2}^s$ the double extension of the 2-dimensional symplectic vector space $\qk=\CC^2$ by the map having matrix:
\[\cb =  \begin{pmatrix} 1 & 0 \\ 0 & -1 \end{pmatrix}\]
	in a Darboux basis $\{E_1,E_2\}$ of $\qk$ where $B(E_1,E_2)=1$.
\item $\g_6^s$ is the double extension of the 4-dimensional symplectic vector space $\qk=\CC^4$ by the map having matrix:
\[\cb =  \begin{pmatrix} 0 & 1 & 0 & 0 \\ 0 & 0 & 0 & 0 \\ 0 & 0 & 0 & 0 \\ 0 & 0 & -1 & 0 \end{pmatrix}\]
in a Darboux basis $\{E_1,E_2,E_3,E_4\}$ of $\qk$ where $B(E_1,E_3)=B(E_2,E_4)=1$, the other are zero.
\end{enumerate}
\end{ex}

Let $(\qk,B)$ be a symplectic vector space. We denote by $\Sp(\qk)$
the isometry group of the symplectic form $B$ and by $\spk(\qk)$ its Lie
algebra, i.e. the Lie algebra of skew-symmetric maps with respects to
$B$. The {\em adjoint action} is the action of $\Sp(\qk)$ on
$\spk(\qk)$ by the conjugation (see Appendix). Also, we denote by
$\ps(\spk(2n))$ the projective space of $\spk(2n)$ with the action
induced by $\Sp(2n)$-adjoint action on $\spk(2n)$.

\begin{prop} \label{3.4.10} 
Let $(\qk, B)$ be a symplectic vector space. Let $\g = (\CC X_{\ze}
\oplus \CC Y_{\ze}) \oplusp \qk$ and $\g' = (\CC X_{\ze}' \oplus \CC
Y_{\ze}') \oplusp \qk$ be double extensions of $\qk$, by
skew-symmetric maps $\cb$ and $\cpb$ respectively. Then:

\begin{enumerate}

\item there exists a Lie superalgebra isomorphism between $\g$ and
  $\g'$ if and only if there exist an invertible map $P \in \Lc(\qk)$
  and a nonzero $\lambda \in \CC$ such that $\cpb = \lambda \ P \cb
  P^{-1}$ and $P^* P \cb = \cb$ where $P^*$ is the adjoint map of $P$
  with respect to $B$.

\smallskip

\item there exists an i-isomorphism between $\g$ and $\g'$ if and only
  if $\cpb$ is in the $\Sp(\qk)$-adjoint orbit through $\lambda \cb$
  for some nonzero $\lambda \in \CC$.

\end{enumerate}

\end{prop}

\begin{proof}  The assertions are obvious if $\cb=0$. We assume $\cb\neq 0$.
\begin{enumerate}

\item Let $A : \g \to \g'$ be a Lie superalgebra isomorphism then
  $A(\CC X_\ze \oplus \CC Y_\ze)= \CC X_\ze' \oplus \CC Y_\ze'$ and
  $A(\qk)=\qk$. It is obvious that $\cpb\neq 0$. It is easy to see
  that $\CC X_\ze = \Zs(\g)\cap \g_\ze$ and $\CC X'_\ze = \Zs(\g')\cap
  \g'_\ze$ then one has $A(\CC X_\ze) = \CC X_\ze'$. It means
  $A(X_\ze) = \mu X_\ze'$ for some nonzero $\mu \in \CC$. Let $A|_\qk
  = Q$ and assume $A(Y_\ze) = \beta Y_\ze'+ \gamma X_\ze'$. For all
  $X$, $Y \in \qk$, we have $A ([X,Y]) = \mu B(\cb(X), Y)
  X_{\ze}'$. Also,
\[
  A([X,Y]) = [Q(X) , Q(Y)]' = B ( \cpb Q(X), Q(Y)) X_{\ze}'.
\]
It results that $Q^* \cpb Q = \mu \cb$.

Moreover, $A([Y_{\ze}, X]) = Q(\cb(X)) = [\beta Y_{\ze}'+ \gamma X_{\ze}', Q(X)]' = \beta \cpb Q(X)$, for all $X \in
\qk$. We conclude that $Q \ \cb \ Q^{-1} = \beta \cpb$ and since
$Q^* \cpb Q = \mu \cb$, then $Q^* Q \cb = \beta \mu \cb$.

Set $P = \dfrac{1}{ (\mu \beta)^{\frac12}} Q$ and $\lambda = \dfrac{1}{\beta}$. It follows that $\cpb= \lambda P \cb
P^{-1}$ and $P^* P \cb = \cb$.

Conversely, assume that $\g = (\CC X_{\ze} \oplus \CC Y_{\ze}) \oplusp
\qk$ and $\g' = (\CC X_{\ze}' \oplus \CC Y_{\ze}') \oplusp \qk$ be
double extensions of $\qk$, by skew-symmetric maps $\cb$ and $\cpb$
respectively such that $\cpb= \lambda P \cb P^{-1}$ and $P^* P \cb =
\cb$ with an invertible map $P \in \Lc(\qk)$ and a nonzero $\lambda
\in \CC$. Define $A : \g \to \g'$ by $A(X_{\ze}) = \lambda X_{\ze}'$,
$A(Y_{\ze}) = \dfrac{1}{\lambda}Y_{\ze}'$ and $A(X) = P(X)$, for all
$X \in \qk$ then it is easy to check that $A$ is a Lie superalgebra isomorphism.
  
\smallskip

\item If $\g$ and $\g'$ are i-isomorphic, then the isomorphism $A$ in
  the proof of (1) is an isometry. Hence $P \in \Sp(\qk)$ and $\cpb= \lambda P \cb
P^{-1}$ gives the result.

  Conversely, define $A$ as above (the sufficiency of (1)). Then $A$ is an
  isometry and it is easy to check that $A$ is an i-isomorphism.
\end{enumerate}

\end{proof}

\begin{cor} \label{3.4.11} Let $(\g,B)$ and $(\g',B')$ be double
  extensions of $(\qk, \overline{B})$ and $(\qk', \overline{B'})$
  respectively where $\overline{B} = B|_{\qk \times \qk}$ and
  $\overline{B'} = B'|_{\qk' \times \qk'}$. Write $\g = (\CC X_{\ze}
  \oplus \CC Y_{\ze}) \oplusp \qk$ and $\g' = (\CC X_{\ze}' \oplus \CC Y_{\ze}')
  \oplusp \qk'$. Then:

\begin{enumerate}

\item there exists an i-isomorphism between $\g$ and $\g'$ if and only
  if there exists an isometry $\overline{A} : \qk \to \qk'$ such that
  $\cpb = \lambda \ \overline{A} \ \cb \ \overline{A}{}^{-1}$, for
  some nonzero $\lambda \in \CC$.

\smallskip

\item there exists a Lie superalgebra isomorphism between $\g$ and $\g'$
  if and only if there exist invertible maps $\overline{Q} : \qk \to
  \qk'$ and $\overline{P} \in \Lc(\qk)$ such that 

\begin{itemize}

\item[(i)] $\cpb = \lambda \ \overline{Q} \ \cb \ \overline{Q}^{-1}$
  for some nonzero $\lambda \in \CC$,

\item[(ii)] $\overline{P}^* \ \overline{P} \ \cb = \cb$ and

\item[(iii)] $\overline{Q} \ \overline{P}^{-1}$ is an isometry from
  $\qk$ onto $\qk'$.

\end{itemize}

\end{enumerate}

\end{cor}

\begin{proof} The proof is completely similar to Corollary 4.6 in
  \cite{DPU10}. It is sketched as follows. First we can assume
  $\dim(\g)=\dim(\g')$ and define then a map $F:\g'\rightarrow\g$ by
  $F(X_\ze')=X_\ze$, $F(Y_\ze')=Y_\ze$ and $\bar{F}=F|_{\qk'}$ is an
  isometry from $\qk'$ onto $\qk$. We define a new Lie bracket on $\g$
  by\[[X,Y]'' = F \left( [F^{-1}(X), F^{-1}(Y)]' \right), \ \forall X,
  Y \in \g.\] and denote by $(\g'', [\cdot, \cdot]'')$ this new Lie
  superalgebra. So $F$ is an i-isomorphism from $\g'$ onto $\g''$ and
  by Proposition \ref{3.4.8} (3) $\g'' = (\CC X_1 \oplus \CC Y_1)
  \oplusp \qk$ is the double extension of $\qk$ by $\overline{C''}$
  with $\overline{C''}= \overline{F}
  \ \overline{C'}\ \overline{F}^{-1}$. We have that $\g$ and $\g'$ are
  isomorphic (resp. i-isomorphic) if and only if $\g$ and $\g''$ are
  isomorphic (resp. i-isomorphic). Finally, by applying Proposition
  \ref{3.4.11} to the Lie superalgebras $\g$ and $\g''$ we obtain the
  corollary.
	\end{proof}

It results that quadratic Lie superalgebra
structures on the quadratic $\ZZ_2$-graded vector space $\CC^{2} \underset{\ZZ_2}{\oplus} \CC^{2n}$ can be classified up to i-isomorphism  in terms of $\Sp(2n)$-orbits in $\ps(\spk(2n))$. This work is like what have been done in \cite{DPU10}. We need the following lemma:

\begin{lem} \label{3.4.12} Let $V$ be a quadratic $\ZZ_2$-graded vector space such that
  its even part is 2-dimensional. We write $V = (\CC X_{\ze}' \oplus \CC Y_{\ze}') \oplusp \qk'$ with $X_{\ze}'$, $Y_{\ze}'$
  isotropic elements in $V_{\ze}$ and $B(X_{\ze}', Y_{\ze}') = 1$. Let $\g$ be a quadratic Lie superalgebra with $\dim(\g_{\ze} ) = \dim(V_{\ze})$ and $\dim(\g) = \dim(V)$. Then,
  there exists a skew-symmetric map $\cpb : \qk' \to \qk'$ such that
  $V$ is considered as the double extension of $\qk'$ by $\cpb$ that is
  i-isomorphic to $\g$.
\end{lem}

\begin{proof}
  By Proposition \ref{3.4.8}, $\g$ is a double
  extension. Let us write $\g = (\CC X_{\ze} \oplus \CC Y_{\ze}) \oplusp \qk$
  and $\cb = \ad(Y_{\ze})|_\qk$. Define $A : \g \to V$ by $A(X_{\ze}) = X_{\ze}'$,
  $A(Y_{\ze}) = Y_{\ze}'$ and $\overline{A} = A |_\qk$ any isometry from $\qk
  \to \qk'$. It is clear that $A$ is an isometry from $\g$ to
  $V$. Now, define the Lie super-bracket on $V$ by: \[ [X,Y] = A \left( [
    A^{-1} (X), A^{-1}(Y)] \right), \ \forall \ X, Y \in V.\] Then $V$
  is a quadratic Lie superalgebra, that is i-isomorphic to $\g$. Moreover, $V$ is obviously the double extension of $\qk'$
  by $\cpb = \overline{A} \ \cb \ \overline{A}{}^{-1}$.
\end{proof}

Proposition \ref{3.4.8}, Proposition \ref{3.4.10}, Corollary
\ref{3.4.11} and Lemma \ref{3.4.12} are enough for us to apply the
classification method in \cite{DPU10} for the set $\Ss(2+2n)$ of
quadratic Lie superalgebra structures on the quadratic $\ZZ_2$-graded
vector space $\CC^{2} \underset{\ZZ_2}{\oplus} \CC^{2n}$ by only
replacing the isometry group $\OO(m)$ by $\Sp(2n)$ and $\ok(m)$ by
$\spk(2n)$ to obtain completely similar results. One has the first
characterization of the set $\Ss(2+2n)$:
\begin{prop}\label{3.4.13}
Let $\g$ and $\g'$ be elements in $\Ss(2+2n)$. Then $\g$ and $\g'$ are i-isomorphic if and only if they are isomorphic.
\end{prop}

By using the notion of double extension, we call the Lie superalgebra $\g\in\Ss(2+2n)$ {\em diagonalizable} (resp. {\em invertible}) if it is a double extension by a diagonalizable (resp. invertible) map. Denote the subsets of nilpotent elements, diagonalizable elements and invertible elements in $\Ss(2+2n)$, respectively by $\Ns(2+2n)$, $\Ds(2+2n)$ and by $\Si(2+2n)$. Denote by $\Nsh(2+2n)$, $\Dsh(2+2n)$, $\Sih(2+2n)$ the sets of isomorphism classes in $\Ns(2+2n)$, $\Ds(2+2n)$, $\Si(2+2n)$, respectively and $\Dsh_\mathrm{red}(2+2n)$ the subset of $\Dsh(2+2n)$ including reduced ones. Applying Appendix, we have the classification result of these sets as follows:  

\begin{prop}\label{3.4.14}\hfill
  
\begin{enumerate}
\item There is a bijection between $\ \Nsh(2+2n)$ and the set of
  nilpotent $\ \Sp(2n)$-adjoint orbits of $\ \spk(2n)$ that induces a
  bijection between $\ \Nsh(2+2n)$ and the set of partitions $\
  \Pc_{-1}(2n)$.
\item There is a bijection between $\ \Dsh(2+2n)$ and the set of
  semisimple $\ \Sp(2n)$-orbits of $\ \ps(\spk(2n))$ that induces a
  bijection between $\ \Dsh(2+2n)$ and $\Lambda_n/H_n$ where $H_n$ is
  the group obtained from the group $G_n$ by adding maps $(\lambda_1,
  \dots, \lambda_n)\mapsto\lambda(\lambda_1, \dots, \lambda_n), \
  \forall \lambda\in\CC,\ \lambda\neq 0$. In the reduced case,
  $\Dsh_\mathrm{red}(2+2n)$ is bijective to $\Lambda_n^+/H_n$ with
  $\Lambda_n^+ = \{ (\lambda_1, \dots, \lambda_n) \mid\ \lambda_1,
  \dots, \lambda_n \in \CC,\ \lambda_i \neq 0, \ \forall i \}$.
\item There is a bijection between $\Sih(2+2n)$ and the set of
  invertible $\ \Sp(2n)$-orbits of $\ \ps(\spk(2n))$ that induces a
  bijection between $\ \Sih(2+2n)$ and $\ \Jc_n/\CC^*$.
\item There is a bijection between $\ \widehat{\Ss}(2+2n)$ and the set
  of $\ \Sp(2n)$-orbits of $\ \ps(\spk(2n))$ that induces a bijection
  between $\ \widehat{\Ss}(2+2n)$ and $\ \Dc(2n)/\CC^*$.
\end{enumerate}
\end{prop}

Next, we will describe the sets $\Ns(2+2n)$, $\Ds_\mathrm{red}(2+2n)$ the subset of $\Ds(2+2n)$ including reduced ones, and $\Si(2+2n)$ in term of amalgamated product in Definition \ref{3.4.6}. Remark that except for the nilpotent case, the amalgamated product may have a bad behavior with respect to isomorphisms.
  
\begin{defn} Keep the notation $J_p$ for the Jordan block of size $p$ as in Appendix and define two types of double extension as follows:

\begin{itemize}

\item for $p \geq 2$, we consider the symplectic vector space $\qk = \CC^{2p}$ equipped with its
  canonical bilinear form $\overline{B}$ and the map $\cb_{2p}^J$ having
  matrix \[ \begin{pmatrix} J_p & 0 \\ 0 & -{}^t J_p\end{pmatrix}\] in
  a Darboux basis. Then $\cb_{2p}^J \in \spk(2p)$ and we denote by
  $\jk_{2p}$ the double extension of $\qk $ by $\cb_{2p}^J$. So
  $\jk_{2p} \in \Ns(2+2p)$.

\item for $p \geq 1$, we consider the symplectic vector space $\qk = \CC^{2p}$ equipped with its
  canonical bilinear form $\overline{B}$ and the map $\cb_{p+p}^J$
  with matrix \[ \begin{pmatrix} J_{p} & M \\ 0 & -{}^t
    J_p\end{pmatrix}\] in a Darboux basis where $M=(m_{ij})$
  denotes the $p \times p$-matrix with $m_{p,p} = 1$ and
  $m_{ij} = 0$ otherwise. Then $\cb_{p+p}^J \in \spk(2p)$ and we
  denote by $\jk_{p+p}$ the double extension of $\qk $ by
  $\cb_{p+p}^J$. So $\jk_{p+p} \in \Ns(2+2p)$.

\end{itemize}

\medskip

The Lie superalgebras $\jk_{2p}$ or $\jk_{p+p}$ will be called {\em nilpotent
  Jordan-type} Lie superalgebras.

\end{defn}

Keep the notations as in Appendix. For $n \in \NN$, $n \neq 0$, each $[d]\in \Pc_{-1}(2n)$ can be written as 
\[
[d]=(p_1, p_1, p_2, p_2, \dots, p_k, p_k, 2
  q_1, \dots 2 q_\ell),\]
with all $p_i$ odd, $p_1 \geq p_2 \geq
  \dots \geq p_k$ and $q_1 \geq q_2 \geq \dots \geq q_\ell$.

We associate the partition $[d]$ with the map $\cb_{[d]}
\in \spk(2n)$ having matrix
\[ \diag_{k + \ell}(\cb^J_{2p_1}, \cb^J_{2p_2}, \dots, \cb^J_{2p_k},
\cb^J_{q_1+q_1}, \dots, \cb^J_{q_\ell+q_\ell} )\] in a Darboux basis
of $\CC^{2n}$ and denote by $\g_{[d]}$ the double extension of $\CC^{2n}$ by
$\cb_{[d]}$. Then $\g_{[d]} \in \Ns(2+2n)$ and $\g_{[d]}$ is an
amalgamated product of nilpotent Jordan-type Lie superalgebras. More
precisely, \[ \g_{[d]} = \jk_{2p_1} \ta \jk_{2p_2} \ta \dots \ta
\jk_{2p_k} \ta \jk_{q_1+q_1} \ta \dots \ta \jk_{q_\ell+q_\ell }.\]

\begin{prop}
Each $\g\in \Ns(2+2n)$ is i-isomorphic to a unique amalgamated product $\g_{[d]},\ [d]\in \Pc_{-1}(2n)$, of nilpotent Jordan-type Lie superalgebras.
\end{prop}

For the reduced diagonalizable case, let $\g_{4}^s(\lambda)$ be the double
extension of $\qk = \CC^2$ by $\cb = \begin{pmatrix} \lambda & 0 \\ 0
  & - \lambda \end{pmatrix}$, $\lambda\neq 0$. By  Lemma \ref{3.4.7}, $\g_4^s(\lambda)$ is
i-isomorphic to $\g_{4}^s(1) = \g_{4,2}^s$.

\begin{prop}\label{3.4.17}
  Let $\g\in\Ds_\mathrm{red}(2+2n)$ then $\g$ is an amalgamated product of quadratic Lie superalgebras all
  i-isomorphic to $\g_{4,2}^s$.
\end{prop}

Finally, for the invertible case, we recall the matrix $J_p(\lambda) =
\diag_p(\lambda, \dots, \lambda) + J_p$, $p \geq 1, \ \lambda \in \CC$ and set
 \[\cb_{2p}^J (\lambda)
= \begin{pmatrix} J_p(\lambda) & 0 \\ 0 & - {}^t
  J_p(\lambda) \end{pmatrix}\] in a Darboux basis of $\CC^{2p}$ then $\cb_{2p}^J(\lambda) \in
\spk(2p)$. Let $\jk_{2p}(\lambda)$ be the double
  extension of $\CC^{2p}$ by $\cb_{2p}^J(\lambda)$ then it is called a {\em Jordan-type} quadratic Lie superalgebra.

  When $\lambda = 0$ and $p \geq 2$, we recover the nilpotent
  Jordan-type Lie superalgebras $\jk_{2p}$. If $\lambda \neq 0$, $\jk_{2p}(\lambda)$ becomes an invertible singular
  quadratic Lie superalgebra and \[\jk_{2p}(- \lambda) \simeq
  \jk_{2p}(\lambda).\]

\begin{prop}\label{3.4.18}
  Let $\g\in \Si(2+2n)$ then $\g$ is an amalgamated
  product of quadratic Lie superalgebras all i-isomorphic to Jordan-type quadratic Lie superalgebras
  $\jk_{2p} (\lambda)$, with $\lambda \neq 0$.
\end{prop}

\subsection{Quadratic dimension of reduced quadratic Lie superalgebras
having the 2-dimensional even part}\hfill

Let $(\g,B)$ be a quadratic Lie superalgebra. To any bilinear form $B'$ on $\g$, there is an associated map $D :\g \to \g$ satisfying 
\[ B'(X,Y) = B(D(X),Y),\ \forall \ X,Y \in
\g.\]

\begin{lem}
If $B'$ is even then $D$ is even.
\end{lem}
\begin{proof}
Let $X$ be an element in $\g_\ze$ and assume that $D(X) = Y + Z$ with $Y\in \g_\ze$ and $Z\in \g_\un$. Since $B'$ is even then $B'(X,\g_\un)= 0$. It implies that $B(D(X),\g_\un)= B(Z,\g_\un) = 0$. By the non-degeneracy of $B$ on $\g_\un$, we obtain $Z = 0$ and then $D(\g_\ze)\subset \g_\ze$. Similarly to the case $X\in \g_\un$, it concludes that $D(\g_\un)\subset \g_\un$. Thus, $D$ is even.
\end{proof}

\begin{lem}\label{3.4.20}
  Let $(\g,B)$ be a quadratic Lie superalgebra, $B'$ be an even bilinear form on
  $\g$ and $D \in \Lc(\g)$ be its associated map. Then:

\begin{enumerate}

\item $B'$ is invariant if and only if $D$ satisfies 
\[
D([X,Y]) =  [D(X), Y] = [X, D(Y)]  , \ \forall \ X,Y \in \g.
\]

\smallskip

\item $B'$ is supersymmetric if and only if $D$ satisfies 

 \[ B(D(X), Y) = B(X, D(Y))
  , \ \forall \ X,Y \in \g.\]

In this case, $D$ is called {\em symmetric}.
\smallskip

\item $B'$ is non-degenerate if and only if $D$ is invertible.

\smallskip

\end{enumerate}

\end{lem}
\begin{proof}Let $X,Y$ and $Z$ be homogeneous elements in $\g$ of degrees $x$, $y$ and $z$, respectively.
\begin{enumerate}
	\item If $B'$ is invariant then
	\[B'([X,Y],Z) = B'(X,[Y,Z]).\]
	That means $B(D([X,Y]),Z = B(D(X),[Y,Z]) = B([D(X),Y],Z)$.
	Since $B$ is non-degenerate, one has $D([X,Y]) = [D(X),Y]$. As a consequence, $D([X,Y])=-(-1)^{xy}D([Y,X]) = -(-1)^{xy}[D(Y),X] = [X,D(Y)]$ by $D$ even.
	
	Conversely, if $D$ satisfies $D([X,Y]) =[D(X), Y] = [X, D(Y)]$, for all $X,Y \in \g$, it is easy to check that $B'$ is invariant.
	\item $B'$ is supersymmetric if and only if $B'(X,Y) = (-1)^{xy}B'(Y,X)$. Therefore, $B(D(X),Y) = (-1)^{xy}B(D(Y),X) = B(X,D(Y))$ by $B$ supersymmetric.
	\item It is obvious since $B$ is non-degenerate.
\end{enumerate}
\end{proof}

\begin{defn}
Let $\g$ be a quadratic Lie superalgebra. An even and symmetric map $D\in\Lc(\g)$ satisfying Lemma \ref{3.4.20} (1) is called  a {\em centromorphism} of $\g$.
\end{defn}

By \cite{Ben03}, given a quadratic Lie superalgebra $\g$, the space of centromorphisms of $\g$ and the space generated by invertible ones are the same, denote it by $\Cs(\g)$. As a consequence, the space of even invariant supersymmetric bilinear forms on $\g$ coincides with its subspace generated by non-degenerate ones. Moreover, all those spaces have the same dimension called the {\em quadratic dimension of $\g$} and denoted by $d_q(\g)$. The following proposition gives the formula of $d_q(\g)$ for reduced quadratic Lie
  superalgebras with 2-dimensional even part.

\begin{prop} \label{3.4.22} Let $\g$ be a reduced quadratic Lie
  superalgebra with 2-dimensional even part and $D \in \Lc(\g)$ be an even symmetric map. Then:

\begin{enumerate}

\item $D$ is a centromorphism if and only if there exist $\mu \in
  \CC$ and an even symmetric map $\Zb : \g \to \Zs(\g)$ such that
  $\Zb|_{[\g,\g]} = 0$ and $D = \mu \Id + \Zb$. Moreover $D$ is
  invertible if and only if $\mu \neq 0$.

\smallskip

\item \[ d_q(\g) = 2 + \dfrac{(\dim(\Zs(\g) - 1)) (\dim(\Zs(\g) - 2)}{2}. \]

\smallskip

\end{enumerate}

\end{prop}

\begin{proof} 
 The proof goes exactly as Proposition 7.2 given in \cite{DPU10}, the
 reader may refer to it.

\end{proof}

\section{Singular quadratic Lie superalgebras of type $\Sb_1$}

Let $\g$ be a singular quadratic Lie superalgebra of type $\Sb_1$ such
that $\g_\ze$ is non-Abelian. If $[\g_\un,\g_\un] = \{0\}$ then
$[\g_\un,\g]=\{0\}$ and therefore $\g$ is an orthogonal direct sum of
a singular quadratic Lie algebra of type $\Sb_1$ and a vector
space. There is nothing to do. We can assume that $[\g_\un,\g_\un]
\neq \{0\}$. Fix $\alpha \in \Vs_I$ and choose $\Omega_0 \in
\Alt^2(\gO)$, $\Omega_1 \in \sym^2(\gI)$ such that \[ I = \alpha
\wedge \Omega_0 + \alpha \otimes \Omega_1. \]

Let $X_\ze = \phi^{-1}(\alpha)$ then $X_\ze\in \Zs(\g)$ and $B(X_\ze,X_\ze)=0$. We define linear maps $C_0:\g_\ze\rightarrow\g_\ze$,  $C_1:\g_\un\rightarrow\g_\un$ by $\Omega_0(X,Y) = B(C_0(X),Y)$ if $ X,Y\in\g_\ze$ and $\Omega_1(X,Y) = B(C_1(X),Y)$ if $X,Y\in\g_\un$. Let $C:\g\rightarrow \g$ defined by $C(X+Y) = C_0(X) +C_1(Y)$, for all $ X\in\g_\ze,\ Y\in\g_\un$.
\begin{prop}\label{3.5.1}
For all $X,Y\in\g$, the Lie super-bracket of $\g$ is defined by: 
\[[X,Y] = B(X_\ze,X)C(Y) - B(X_\ze,Y)C(X) + B(C(X),Y)X_\ze.
\]
In particular, if $X,Y\in\g_\ze$, $Z,T\in\g_\un$ then
\begin{enumerate}
	\item $[X,Y] = B(X_\ze,X)C_0(Y) - B(X_\ze,Y)C_0(X) + B(C_0(X),Y)X_\ze$,
	\item $[X,Z] = B(X_\ze,X)C_1(Z)$,
	\item $[Z,T] = B(C_1(Z),T)X_\ze$
\end{enumerate}
\end{prop}
\begin{proof}
By Proposition \ref{3.2.7}, $\g_\ze$ is a singular quadratic Lie algebra so the assertion (1) follows \cite{DPU10}. Given $X\in\g_\ze$, $Y,Z\in\g_\un$, one has
\[B([X,Y],Z) = \alpha\otimes\Omega_1(X,Y,Z) = \alpha(X)\Omega_1(Y,Z) = B(X_\ze,X)B(C_1(Y),Z).
\]
Hence we obtain (2) and (3).
\end{proof}

Now, we show that $\g_\ze$ is solvable. Consider the quadratic Lie algebra $\g_\ze$ with 3-form $I_0 = \alpha\wedge\Omega_0$. Write $\Omega_0 = \sum_{i<j}{a_{ij}\alpha_i\wedge\alpha_j}$, with $a_{ij}\in\CC$. Set $X_i = \phi^{-1}(\alpha_i)$ then \[C_0 = \sum_{i<j}{a_{ij}(\alpha_i\otimes X_j-\alpha_j\otimes X_i)}.\] Define the space $W_{I_0}\subset\g_\ze^*$ by:
\[\Ws_{I_0} = \{\iota_{X\wedge Y}(I_0)\ |\ X,Y\in\g_\ze\}.\] Then $\Ws_{I_0} = \phi([\g_\ze,\g_\ze])$ and that implies $\im(C_0)\subset [\g_\ze,\g_\ze]$. In Section 2, it is known that $\{\alpha,I_0\} = 0$ and then $[X_\ze,\g_\ze] = 0$. As a sequence, $B(X_\ze,[\g_\ze,\g_\ze])=0$. That deduces $B(X_\ze,\im(C_0)) =0$. Therefore $[[\g_\ze,\g_\ze],[\g_\ze,\g_\ze]] = [\im(C_0),\im(C_0)]\subset \CC X_\ze\subset \Zs(\g)$ and we conclude that $\g_\ze$ is solvable.

By $B$ non-degenerate there is an element $Y_\ze\in\g_\ze$ isotropic such that $B(X_\ze,Y_\ze) = 1$. Moreover, combined with $\g_\ze$ a solvable singular quadratic Lie algebra, we can choose $Y_\ze$ satisfying $C_0(Y_\ze)=0$ and we obtain then a straightforward consequence as follows:
\begin{cor}\label{3.5.2}\hfill
\begin{enumerate}
	\item $C = \ad(Y_\ze)$, $\ker(C) = \Zs(\g)\oplus \CC Y_\ze$ and $[\g,\g] = \im(C)\oplus\CC X_\ze$.
	\item The Lie superalgebra $\g$ is solvable. Moreover, $\g$ is nilpotent if and only if C is nilpotent.
\end{enumerate}
\end{cor}
\subsection{Singular quadratic Lie superalgebras of type $\Sb_1$ and double extensions}\hfill

The description of the Lie super-bracket in Proposition \ref{3.5.1}
allows us to propose a definition of double extension of a quadratic
$\ZZ_2$-graded vector space as follows:

\begin{defn}\label{3.5.3}
Let $(\qk = \qk_\ze\oplus\qk_\un,B_\qk)$ be a quadratic $\ZZ_2$-graded
vector space and $\cb$ be an even endomorphism of $\qk$. Assume that
$\cb$ is skew-supersymmetric, that is, $B(\cb(X),Y) = -B(X,\cb(Y))$,
for all $ X,Y\in\qk$. Let $(\tk = \spa\{X_\ze,Y_\ze\},B_\tk)$ be a
2-dimensional quadratic vector space with the symmetric bilinear form
$B_\tk$ defined by:
\[B_\tk(X_\ze,X_\ze) = B_\tk(Y_\ze,Y_\ze) = 0 \ \text{and } B_\tk(X_\ze,Y_\ze) = 1.\]

Consider the vector space $\g = \tk\oplusp \qk$ equipped with the
bilinear form $B = B_\tk+B_\qk$ and define on $\g$ the following
bracket:
\[ [\lambda X_{\ze} + \mu Y_{\ze} + X, \lambda' X_{\ze} + \mu' Y_{\ze} + Y] =
  \mu \cb (Y) - \mu' \cb (X) + B( \cb (X), Y) X_\ze, \] for all $X, Y
  \in \qk, \lambda, \mu, \lambda', \mu' \in \CC$. Then $(\g, B)$ is a
  quadratic solvable Lie superalgebra with $\g_{\ze} =
  \tk\oplus\qk_\ze$ and $\g_{\un} =\qk_\un$. We say that $\g$ is the
              {\em double extension of $\qk$ by $\cb$}.
\end{defn}

Note that an even skew-supersymmetric endomorphism $\cb$ on $\qk$ can be written by $\cb= \cb_0 + \cb_1$ where $\cb_0\in\ok(\qk_\ze)$ and $\cb_1\in\spk(\qk_\un)$.
\begin{cor}\label{3.5.4}
Let $\g$ be the double extension of $\qk$ by $\cb$. Denote by $C = \ad(Y_\ze)$ then one has
\begin{enumerate}
	\item $[X,Y] = B(X_\ze,X)C(Y) - B(X_\ze,Y)C(X) + B(C(X),Y)X_\ze$, for all $X,Y\in\g.$
	\item $\g$ is a singular quadratic Lie superalgebra. If
       $\cb|_{\qk_\un}$ is nonzero then $\g$ is of type $\Sb_1$.
\end{enumerate}
\end{cor}
\begin{proof}
  The assertion (1) is direct from the above definition. Let $\alpha =
  \phi(X_\ze)$ and define the bilinear form $\Omega:\g\rightarrow\g$
  by $\Omega(X,Y) = B(C(X),Y)$ for all $X,Y\in\g$. By $B$ even and
  supersymmetric, $C$ even and skew-supersymmetric (with respect to
  $B$) then $\Omega = \Omega_0 + \Omega_1\in\Alt^2(\gO) \oplus
  \sym^2(\gI)$. The formula in (1) can be replaced by $I =
  \alpha\wedge\Omega_0 + \alpha\otimes\Omega_1 =
  \alpha\wedge\Omega$. Therefore, $\dup(\g)\geq 1$ and $\g$ is
  singular. If $\cb|_{\qk_\un}$ is nonzero then $\Omega_1\neq 0$. In
  this case, $[\g_\un,\g_\un]\neq \{0\}$ and thus $\dup(\g)=1$.
\end{proof}

As a consequence of Proposition \ref{3.5.1} and Definition \ref{3.5.3}, one has
\begin{lem}\label{3.5.5}
Let $(\g,B)$ be a singular quadratic Lie superalgebra of type $\Sb_1$. Keep the notations as in Proposition \ref{3.5.1} and Corollary \ref{3.5.2}. Then $(\g,B)$ is the double extension of $\qk=(\CC X_\ze\oplus\CC Y_\ze)^\bot$ by $\cb = C|_\qk$.
\end{lem}

\begin{rem}
The above definition is a generalization of the definition of double extension of a quadratic vector space by a skew-symmetric map in \cite{DPU10} and Definition \ref{3.4.6}. Moreover, if let $\g = (\CC X_\ze\oplus \CC Y_\ze)\oplusp (\qk_\ze\oplus\qk_\un)$ be the double extension of $\qk=\qk_\ze\oplus\qk_\un$ by $\cb= \cb_0 + \cb_1$ then $\g_\ze$ is the double extension of $\qk_\ze$ by $\cb_0$ and the subalgebra $(\CC X_\ze\oplus \CC Y_\ze)\oplusp \qk_\un$ is the double extension of $\qk_\un$ by $\cb_1$.
\end{rem}

The proof of the proposition below is completely analogous to the
proof of Proposition \ref{3.4.10}, so we omit it.

\begin{prop}\label{3.5.7}
  Let $\g = (\CC X_\ze\oplus \CC Y_\ze)\oplusp (\qk_\ze\oplus\qk_\un)$
  and $\g' = (\CC X_\ze'\oplus \CC Y_\ze')\oplusp
  (\qk_\ze\oplus\qk_\un)$ be two double extensions of $\qk =
  \qk_\ze\oplus\qk_\un$ by $\cb =\cb_0+\cb_1$ and $\overline{C'}
  =\overline{C_0'}+\overline{C_1'}$, respectively. Assume that $\cb_1$
  is nonzero. Then

\begin{enumerate}
\item there exists a Lie superalgebra isomorphism between $\g$ and
  $\g'$ if and only if there exist invertible maps $P\in\Lc(\qk_\ze)$,
  $Q\in\Lc(\qk_\un)$ and a nonzero $\lambda\in\CC$ such that
	
\begin{itemize}
	\item[(i)] $\overline{C_0'}=\lambda P\cb_0 P^{-1}$ and $P^*P\cb_0 = \cb_0$.
	\item[(ii)] $\overline{C_1'}=\lambda Q\cb_1 Q^{-1}$ and $Q^*Q\cb_1 = \cb_1$.
\end{itemize}
where $P^*$ and $Q^*$ are the adjoint maps of $P$ and $Q$ with respect
to $B|_{\qk_\ze\times\qk_\ze}$ and $B|_{\qk_\un\times\qk_\un}$.
\item there exists an i-isomorphism between $\g$ and $\g'$ if and only
  if there is a nonzero $\lambda\in\CC$ such that $\overline{C_0'}$ is
  in the $\OO(\qk_\ze)$-adjoint orbit through $\lambda \cb_0$ and
  $\overline{C_1'}$ is in the $\Sp(\qk_\un)$-adjoint orbit through
  $\lambda \cb_1$.
\end{enumerate}
\end{prop}

\begin{rem}\label{3.5.8}
If let $M = P + Q$ then $M^{-1} = P^{-1} + Q^{-1}$ and $M^* =
P^*+Q^*$. The formulas in Proposition \ref{3.5.7} (1) can be written:
\[\overline{C'} =\lambda M\cb  M^{-1}\ \text{and }M^*M\cb = \cb.
\]
Hence, the classification problem of singular quadratic Lie
superalgebras of type $\Sb_1$ (up to i-isomorphism) can be reduced to
the classification of $\OO(\qk_\ze)\times \Sp(\qk_\un)$- orbits of
$\ok(\qk_\ze)\oplus\spk(\qk_\un)$, where $\OO(\qk_\ze)\times
\Sp(\qk_\un)$ denotes the direct product of two groups $\OO(\qk_\ze)$
and $\Sp(\qk_\un)$.
\end{rem}

\begin{defn}
Let $\qk=\qk_\ze\oplus\qk_\un$ be a quadratic $\ZZ_2$-graded vector
space. An even isomorphism $F\in \Lc(\qk)$ is called an {\em isometry}
of $\qk$ if $F|_{\qk_\ze}$ and $F|_{\qk_\un}$ are isometries.
\end{defn}

To prove the following Corollary, it is enough to follow exactly the
same steps as in Corollary \ref{3.4.11}.
 
\begin{cor}\label{3.5.10}
 Let $(\g,B)$ and $(\g',B')$ be double extensions of $(\qk,
 \overline{B})$ and $(\qk', \overline{B'})$ by $\cb$ and
 $\overline{C'}$ respectively where $\overline{B} = B|_{\qk \times
   \qk}$ and $\overline{B'} = B'|_{\qk' \times \qk'}$. Write $\g =
 (\CC X_{\ze} \oplus \CC Y_{\ze}) \oplusp \qk$ and $\g' = (\CC
 X_{\ze}' \oplus \CC Y_{\ze}') \oplusp \qk'$. Then:

\begin{enumerate}

\item there exists an i-isomorphism between $\g$ and $\g'$ if and only
  if there exists an isometry $\overline{A} : \qk \to \qk'$ such that
  $\cpb = \lambda \ \overline{A} \ \cb \ \overline{A}{}^{-1}$, for
  some nonzero $\lambda \in \CC$.

\smallskip

\item there exists a Lie superalgebra isomorphism between $\g$ and
  $\g'$ if and only if there exist even invertible maps $\overline{Q}
  : \qk \to \qk'$ and $\overline{P} \in \Lc(\qk)$ such that

\begin{itemize}

\item[(i)] $\cpb = \lambda \ \overline{Q} \ \cb \ \overline{Q}^{-1}$
  for some nonzero $\lambda \in \CC$,

\item[(ii)] $\overline{P}^* \ \overline{P} \ \cb = \cb$ and

\item[(iii)] $\overline{Q} \ \overline{P}^{-1}$ is an isometry from
  $\qk$ onto $\qk'$.

\end{itemize}

\end{enumerate}

\end{cor}

\subsection{Fitting decomposition of a skew-supersymmetric map}\hfill

We recall the following useful result (see for instance \cite{DPU10}):

\begin{lem} \label{3.5.11}
 Let $\cb $ and $\overline{C'}$ be nilpotent elements in
 $\ok(n)$. Then $\cb$ is conjugate to $\lambda \overline{C'}$ modulo
 $\OO(n)$ for some nonzero $\lambda \in \CC$ if and only if $\cb$ is
 conjugate to $\overline{C'}$.
\end{lem}

Remark that the lemma remains valid if we replace $\ok(n)$ by $\spk(2n)$ and $\OO(n)$ by $\Sp(2n)$.
\begin{prop}\label{3.5.12}
Let $\g$ and $\g'$ be two nilpotent singular quadratic Lie superalgebras. Then $\g$ and $\g'$ are isomorphic if and only if they are i-isomorphic.
\end{prop}
\begin{proof}
Singular quadratic Lie superalgebras $\g$ and $\g'$ are regarded as
double extensions $\g=(\CC X_\ze \oplus\CC Y_\ze)\oplusp\qk$ and $(\CC
X'_\ze \oplus\CC Y'_\ze)\oplusp\qk'$ by $\cb$ and $\overline{C'}$
where $\qk=\qk_\ze\oplus\qk_\un$ and $\qk'=\qk'_\ze\oplus\qk'_\un$. By
Corollary \ref{3.5.2}, $\cb$ and $\overline{C'}$ are
nilpotent. Rewrite $\cb=\cb_0+\cb_1$ and
$\overline{C'}=\overline{C_0'}+\overline{C_1'}$, where
$\cb_0\in\ok(\qk_\ze)$, $\overline{C_0'}\in\ok(\qk'_\ze)$,
$\cb_1\in\spk(\qk)$ and $\overline{C_1'}\in\spk(\qk')$.

If $\g$ and $\g'$ are isomorphic then $\dim(\qk_\ze)=\dim(\qk'_\ze)$
and $\dim(\qk_\un)=\dim(\qk'_\un)$. Thus, there exist isometries
$\overline{F}_0:\qk'_\ze\rightarrow \qk_\ze$ and
$\overline{F}_1:\qk'_\un\rightarrow \qk_\un$ and then we define an
isometry $\overline{F}:\qk'\rightarrow \qk$ by
$\overline{F}(X'+Y')=\overline{F}_0(X')+\overline{F}'_0(Y')$ for all
$X'\in\qk'_\ze$ and $Y'\in\qk'_\un$. We now set $F:\g'\rightarrow\g$
by $F(X'_\ze)=X_\ze$, $F(Y'_\ze)=Y_\ze$, $F|_{\qk'}=\overline{F}$ and
a new Lie super-bracket on $\g$ by:
\[[X,Y]'' = F \left( [F^{-1}(X), F^{-1}(Y)]'
  \right), \ \forall X, Y \in \g.\]

  Denote by $\g''$ this new quadratic Lie superalgebras. It is easy to
  see that $\g''=(\CC X_\ze \oplus\CC Y_\ze)\oplusp\qk$ is the double
  extension of $\qk$ by $\overline{C''}=\overline{F}\ \overline{C'}\
  \overline{F}^{-1}$ and $\g''$ is i-isomorphic to $\g'$. It need to
  prove that $\g''$ is i-isomorphic to $\g$. Write
  $\overline{C''}=\overline{C''_0}+\overline{C''_1}\in\ok(\qk_\ze)\oplus\spk(\qk_\un)$. Since
  $\g$ and $\g''$ are isomorphic then there exist invertible maps
  $P:\qk_\ze\rightarrow\qk_\ze$ and $Q:\qk_\un\rightarrow\qk_\un$ such
  that $\overline{C''_0}=\lambda\overline{P}\ \cb_0\
  \overline{P}^{-1}$ and $\overline{C''_1}=\lambda\overline{Q}\ \cb_1\
  \overline{Q}^{-1}$ for some nonzero $\lambda\in\CC$. By Lemma
  \ref{3.5.11}, $\cb_0$ and $\overline{C''_0}$ are conjugate under
  $\OO(\qk_\ze)$, $\cb_1$ and $\overline{C''_1}$ are conjugate under
  $\Sp(\qk_\un)$ and we can assume that $\lambda=1$. Therefore $\g$
  and $\g''$ are i-isomorphic. The proposition is proved.
\end{proof}

Let now $\g$ be a singular quadratic Lie superalgebra of type $\Sb_1$. Write $\g$ as a double extension of $(\qk=\qk_\ze\oplus\qk_\un,\overline{B})$ by $\cb= \cb_0 + \cb_1$ where $\cb=\ad(Y_\ze)|_\qk$, $\cb_0=\cb|_{\qk_\ze}$ and $\cb_1=\cb|_{\qk_\un}$. We consider the Fitting decomposition of $\cb_0$ on $\qk_\ze$ and $\cb_1$ on ${\qk_\un}$ by:
\[\qk_\ze=\qk_\ze^N\oplus\qk_\ze^I \ \text{ and } \ \qk_\un=\qk_\un^N\oplus \qk_\un^I\]
where $\qk_\ze^N$ and $\qk_\ze^I$ (resp. $\qk_\un^N$ and $\qk_\un^I$)
are $\cb_0$-stable (resp. $\cb_1$-stable),
$\cb_0^N=\cb_0|_{\qk_\ze^N}$ and $\cb_1^N=\cb_1|_{\qk_\un^N}$ are
nilpotent, $\cb_0^I=\cb_0|_{\qk_\ze^I}$ and
$\cb_1^I=\cb_1|_{\qk_\un^I}$ are invertible. Recall that $\cb$ is
skew-supersymmetric then $\qk_\ze^I=(\qk_\ze^N)^\bot$ in $\g_\ze$ and
$\qk_\un^I=(\qk_\un^N)^\bot$ in $\g_\un$.

Next, we set 
\[\qk_N=\qk_\ze^N\oplus\qk_\un^N \ \text{ and }\ \qk_I=\qk_\ze^I\oplus\qk_\un^I\]
As a consequence, $\cb_N=\cb|_{\qk_N}$ is nilpotent, $\cb_I=\cb|_{\qk_I}$ is invertible, $[\qk_N,\qk_I]=\{0\}$, the restrictions $\overline{B}_N=\overline{B}|_{\qk_N\times\qk_N}$ and $\overline{B}_I=\overline{B}|_{\qk_I\times\qk_I}$ are non-degenerate and supersymmetric. It is easy to check that $\cb_N = \cb_0^N+\cb_1^N$, $\cb_I = \cb_0^I+\cb_1^I$. Moreover, $\cb_N$, $\cb_I$ are skew-supersymmetric and they are Fitting components of $\cb$ in $\qk$. Let $\g_N = (\CC X_\ze \oplus \CC Y_\ze) \oplusp
\qk_N$ and $\g_I = (\CC X_\ze \oplus \CC Y_\ze) \oplusp \qk_I$. Then
$\g_N$ and $\g_I$ are subalgebras of $\g$, $\g_N$ is the double
extension of $\qk_N$ by $\cb_N$, $\g_I$ is the double extension of
$\qk_I$ by $\cb_I$ and $\g_N$ is a nilpotent singular quadratic Lie
superalgebra.
\begin{defn}
  The subalgebras $\g_N$ and $\g_I$ as above are respectively the {\em nilpotent} and {\em invertible Fitting components} of $\g$.
\end{defn}
\begin{defn}
  A double extension is called an {\em invertible quadratic Lie
    superalgebra} if the corresponding skew-supersymmetric map is invertible.
\end{defn}

It is easy to check that the dimension of an invertible quadratic Lie
superalgebra must be even. Moreover, following Corollary \ref{3.5.10},
two invertible quadratic Lie superalgebras are isomorphic if and only
if they are i-isomorphic. This property is still right for singular
quadratic Lie superalgebras of type $\Sb_1$.

\begin{prop}\label{3.5.15}
Let $\g$ and $\g'$ be singular quadratic Lie superalgebras of type $\Sb_1$ and
  $\g_N$, $\g_I$, $\g'_N$, $\g'_I$ be their Fitting components, respectively. Then

\begin{enumerate}

\item $\g \iiso \g'$ if and only if $\g_N \iiso \g'_N$ and $\g_I
  \iiso \g_I'$. The result remains valid if we replace $\iiso$ by
  $\simeq$.

\smallskip

\item $\g $ and $\g'$ are isomorphic if and only if they are i-isomorphic.

\end{enumerate}

\end{prop}

\begin{proof}
  The proposition is proved as Proposition 6.4 in \cite{DPU10}. It is
  sketched as follows. We assume that $\g\simeq\g'$. They are regarded
  as double extensions $\g=(\CC X_\ze \oplus\CC Y_\ze)\oplusp\qk$ and
  $(\CC X'_\ze \oplus\CC Y'_\ze)\oplusp\qk'$ by $\cb$ and
  $\overline{C'}$. By Corollary \ref{3.5.10}, there are an even
  invertible map $\overline{P} : \qk \to \qk'$ and a nonzero $\lambda
  \in \CC$ such that $\overline{C'} = \lambda \ \overline{P} \ \cb \
  \overline{P}^{-1}$, so $\qk_N' = \overline{P} (\qk_N)$ and $\qk_I' =
  \overline{P} (\qk_I)$, then $\dim(\qk'_N) = \dim(\qk_N)$ and
  $\dim(\qk'_I) = \dim(\qk_I)$. Thus, there exist isometries $F_N :
  \qk'_N \to \qk_N$ and $F_I : \qk'_I \to \qk_I$ and we can define an
  isometry $\overline{F} : \qk' \to \qk$ by $\overline{F}(X_N' + X_I')
  = F_N(X_N') + F_I(X_I')$, $\forall X_N' \in \qk_N'$ and $ X_I' \in
  \qk_I'$. We now define $F : \g' \to \g$ by $F(X_1') = X_1$, $F(Y_1')
  = Y_1$, $F|_{\qk'} = \overline{F}$ and a new Lie super-bracket on
  $\g$: \[[X,Y]'' = F\left( [F^{-1}(X), F^{-1}(Y)]' \right), \ \forall
  X, Y \in \g.\]
	
  Denote by $\g''$ this new quadratic Lie superalgebra. It is obvious
  that $\g'\iiso\g''$.  It remains to prove the assertions for two
  quadratic Lie superalgebras $\g$ and $\g''$. Those follow Corollary
  \ref{3.5.10}, Lemma \ref{3.5.11} and Proposition \ref{3.5.12}.

\end{proof}
\begin{prop}\label{3.5.16}  The $\dup$-number is invariant under Lie superalgebra isomorphisms, i.e. if $(\g,B)$ and
  $(\g',B')$ are quadratic Lie superalgebras with $\g \simeq \g'$, then
  $\dup(\g) = \dup(\g')$.
\end{prop}

\begin{proof}
  By Lemma \ref{3.2.4} we can assume that $\g$ is reduced. By Proposition \ref{3.2.3}, $\g'$ is also reduced. Since $\g\simeq \g'$ then we can identify $\g = \g'$ as a Lie superalgebra equipped with the bilinear forms $B$, $B'$ and we have two $\dup$-numbers: $\dup_B(\g)$ and $\dup_{B'}(\g)$.
  
  We start with the case $\dup_B(\g) =3$. Since $\g$ is reduced then $\g_\un = \{0\}$ and $\g$ is a reduced singular quadratic Lie algebra of type $\Sb_3$. By \cite{PU07}, $\dim([\g,\g])=3$ and then $\dup_{B'} (\g) = 3$.

  If $\dup_B(\g) = 1$, then $\g$ is of type $\Sb_1$ with respect to
  $B$. There are two cases: $[\g_\un,\g_\un]= \{0\}$ and $[\g_\un,\g_\un]\neq \{0\}$. If $[\g_\un,\g_\un]= 0$ then $\g_\un = \{0\}$ by $\g$ reduced. In this case, $\g$ is a reduced singular quadratic Lie algebra of type $\Sb_1$. By \cite{DPU10}, $\g$ is also a reduced singular quadratic Lie algebra of type $\Sb_1$ with the bilinear form $B'$, i.e. $\dup_{B'} (\g) = 1$. 
  
  Assume that $[\g_\un,\g_\un]\neq \{0\}$, we need the following lemma:

\begin{lem}\label{3.5.17} Let $\g$ be a reduced quadratic Lie superalgebras of type $\Sb_1$ such that $[\g_\un,\g_\un]\neq 0$ and $D \in \Lc(\g)$ be an even symmetric map. Then $D$ is a centromorphism if and only if there exist $\mu \in  \CC$ and an even symmetric map $\Zb : \g \to \Zs(\g)$ such that
  $\Zb|_{[\g,\g]} = 0$ and $D = \mu \Id + \Zb$. Moreover $D$ is
  invertible if and only if $\mu \neq 0$.
\end{lem} 
\begin{proof}
First, $\g$ can be realized as the double  extension $\g = (\CC X_\ze\oplus \CC Y_{\ze}) \oplusp \qk$ by $C = \ad(Y_\ze)$ and let $\cb = C|_\qk$.

  Assume that $D$ is an invertible centromorphism. The condition (1) of Lemma \ref{3.4.20} implies that $D \circ \ad(X) =
  \ad(X) \circ D$, for all $X \in \g$ and then $D C = C
  D$. Using formula (1) of Corollary \ref{3.5.4} and $CD = DC$, from $[D(X),
  Y_{\ze}] = [X, D(Y_\ze)]$ we find 
  \[D(C(X)) =\mu C(X),\ \forall\ X\in\g,\ \text{where }  \mu = B(D(X_\ze), Y_\ze).\]
 Since $D$ is invertible, one has  $\mu \neq 0$ and $C (D - \mu \Id) = 0$. Recall that $\ker(C) = \CC X_\ze
  \oplus \ker(\cb) \oplus \CC Y_\ze = \Zs(\g) \oplus \CC Y_\ze$, there
  exist a map $\Zb : \g \to \Zs(\g)$ and $\varphi \in \g^*$ such that
  \[
  D - \mu \Id = \Zb + \varphi \otimes Y_\ze.\]
  
  It needs to show that $\varphi = 0$. Indeed, $D$ maps $[\g,\g]$
  into itself and $Y_\ze\notin [\g,\g]$, so $\varphi|_{[\g,\g]} = 0$. One has $[\g,\g] = \CC X_\ze
  \oplus \im(\cb)$. If $X \in \im(\cb)$, let $X = C(Y)$. Then $D(X) =
  D(C(Y)) = \mu C(Y)$, so $D(X) = \mu X$. For $Y_\ze$, $D([Y_\ze,X]) =
  DC(X) = \mu C(X)$ for all $X \in \g$. But also, $D([Y_\ze,X]) =
  [D(Y_\ze),X] = \mu C(X) + \varphi(Y_\ze) C(X)$, hence $\varphi(Y_\ze) =
  0$. As a consequence, $D(Y_\ze) = \mu Y_\ze+ \Zb(Y_\ze)$.

  Now, we prove that $D(X_\ze) = \mu X_\ze$. Indeed, since $D$ is even and $[\g_\un,\g_\un] = \CC X_\ze$ then one has
  \[D(X_\ze)\subset D([\g_\un,\g_\un])=[D(\g_\un),\g_\un]\subset [\g_\un,\g_\un] = \CC X_\ze.
  \]
  It implies that, $D(X_\ze) = aX_\ze$. Combined with $B(D(Y_\ze),X_\ze) =  B(Y_\ze,D(X_\ze))$, we obtain $\mu = a$. 
  
  Let $X \in \qk$,
  $B(D(X_\ze), X) = \mu B(X_\ze,X) = 0$. Moreover, $B(D(X_\ze), X) = B(X_\ze,
  D(X))$, so $\varphi(X) = 0$. 
  
  Since $\Cs(\g)$ is generated by invertible centromorphisms then the necessary condition of Lemma is finished. The sufficiency is obvious.
 \smallskip
  \end{proof}
  
  Let us return now to the proposition. By the previous lemma, the bilinear
  form $B'$ defines an associated invertible centromorphism $D =
  \mu\Id + \Zb$ for some nonzero $\mu\in\CC$ and $\Zb : \g \to
  \Zs(\g)$ satisfying $\Zb|_{[\g,\g]} = 0$. For all $X,Y,Z\in\g$, one
  has:
  \[I'(X,Y,Z) = B'([X,Y],Z) = B(D([X,Y]),Z) = B([D(X),Y],Z) = \mu B([X,Y],Z).\]
  
  That means $I' = \mu I$ and then $\dup_{B'} (\g) = \dup_{B} (\g) = 1$.
    
  Finally, if $\dup_B(\g) = 0$ then $\g$
  cannot be of type $\Sb_3$ or $\Sb_1$ with respect to $B'$, so
  $\dup_{B'}(\g) = 0$.
\end{proof}

Let $\g$ be a reduced singular quadratic Lie superalgebra of type $\Sb_1$ such that $[\g_\un,\g_\un]\neq 0$. Keep the notation as in Lemma \ref{3.5.17}. We set $\Zs(\g)_\ze=\Zs(\g)\cap\g_\ze$, $\Zs(\g)_\un=\Zs(\g)\cap\g_\un$, $[\g,\g]_\ze=[\g,\g]\cap\g_\ze$ and $[\g,\g]_\un=[\g,\g]\cap\g_\un$. It is obvious that $X_\ze\in\Zs(\g)_\ze\subset [\g,\g]_\ze$ and $\Zs(\g)_\un\subset [\g,\g]_\un$. In other words, $\Zs(\g)_\ze$ and $\Zs(\g)_\un$ are totally isotropic subspaces of $\g_\ze$ and $\g_\un$, respectively. Rewrite $\Zs(\g)_\ze=\CC X_\ze\oplus \lk_\ze$. Then there exist totally isotropic subspaces $\uk_\ze\oplus \CC Y_\ze$ of $\g_\ze$ and $\uk_\un$ of $\g_\un$ such that $\g_\ze=[\g,\g]_\ze\oplus(\uk_\ze\oplus \CC Y_\ze)$, $\g_\un=[\g,\g]_\un\oplus\uk_\un$, the subspaces $\Zs(\g)_\ze\oplus(\uk_\ze\oplus \CC Y_\ze)$ and $\Zs(\g)_\un\oplus\uk_\un$ are non-degenerate. Let us define
 \[\Zb:\uk_\ze\oplus\CC Y_\ze\oplus\uk_\un \rightarrow \lk_\ze\oplus\CC X_\ze\oplus\Zs(\g)_\un\]
by: set bases $\{X_1=X_\ze,X_2,...,X_r\}$ of $\lk_\ze\oplus\CC X_\ze$, $\{Y_1,...,Y_t\}$ of $\Zs(\g)_\un$, $\{X_1'=Y_\ze,X_2',...,X_r'\}$ of $\uk_\ze\oplus\CC Y_\ze$ and $\{Y'_1,...,Y'_t\}$ of $\uk_\un$ such that $B(X_i,X'_j)=\delta_{ij}$, $B(Y_k,Y'_l)=\delta_{kl}$. Then the map $\Zb$ is completely defined by
\[\Zb \left(
    \sum_{j=1}^r x_j X_j' \right) = \sum_{i=1}^r \left( \sum_{j=1}^r
    \mu_{ij} x_j \right) X_i,\] 
\[\Zb \left(
    \sum_{j=1}^t y_j Y_j' \right) = \sum_{i=1}^t \left( \sum_{j=1}^t
    \nu_{ij} y_j \right) Y_i
\]
		with $\mu_{ij} = \mu_{ji} = B(X_i',
  \Zb(X_j'))$ and $\nu_{ij} = -\nu_{ji} = B(Y_i',
  \Zb(Y_j'))$.
	
	It results that the quadratic dimension of $\g$ can be calculated as follows:
	\[
	d_q(\g) = 1 + \dfrac{\dim(\Zs(\g)_\ze) (1+\dim(\Zs(\g)_\ze))}{2} + \dfrac{\dim(\Zs(\g)_\un) (\dim(\Zs(\g)_\un)-1)}{2}.
	\]

\section{Quasi-singular quadratic Lie algebras}
By Definition \ref{3.5.3}, it is natural to question: let $(\qk = \qk_\ze\oplus\qk_\un,B_\qk)$ be a quadratic $\ZZ_2$-graded vector space and $\cb$ be an endomorphism of $\qk$. Let $(\tk = \spa\{X_\un,Y_\un\},B_\tk)$ be a 2-dimensional {\em symplectic} vector space with $B_\tk(X_\un,Y_\un) = 1$. Is there an extension $\g = \qk\oplusp\tk$ such that $\g$ equipped with the bilinear form $B = B_\qk+B_\tk$ becomes a quadratic Lie superalgebra such that $\g_\ze = \qk_\ze$, $\g_\un = \qk_\un\oplus\tk$ and the Lie super-bracket is represented by $\cb$? In this section, we will give an affirmative answer to this question.

The dup-number and the form of the associated invariant $I$ in the previous
sections suggest that it would be also interesting to study a quadratic Lie
superalgebra $\g$ whose associated invariant $I$ has the form \[I = J\wedge p
\] where $p\in \g_{\un}^*$ is nonzero, $J\in \Alt^1(\gO)\otimes
\sym^1(\gI)$ is indecomposable. We obtain the first result as follows:

\begin{prop}$\{J,J\} = \{p,J\}=0.$
\end{prop}
\begin{proof}
Apply Proposition \ref{3.1.5} to obtain
\[\{I,I\} = \{J\wedge p, J\wedge p\} = \{J\wedge p, J\}\wedge p + J\wedge\{J\wedge p, p\}
\]
\[ = -\{J, J\}\wedge p\wedge p + 2 J\wedge\{p, J\}\wedge p - J\wedge J\wedge \{p, p\}.
\]
Since the super-exterior product is commutative then one has $J\wedge J=0$. Moreover, $\{I,I\}=0$ implies that:
\[\{J, J\}\wedge p\wedge p= 2 J\wedge\{p, J\}\wedge p.
\]
That means $\{J, J\}\wedge p = 2 J\wedge\{p, J\}$.

If $\{J, J\}\neq 0$ then $\{J, J\}\wedge p\neq 0$, so $J\wedge\{p, J\}\neq 0$. Note that $\{p, J\}\in \Alt^1(\gO)$ so $J$ must contain the factor $p$, i.e. $J = \alpha\otimes p$ where $\alpha \in\g_\ze^*$. But $\{p, J\}=\{p, \alpha\otimes p\} = -\alpha\otimes\{p, p\} = 0$ since $\{p, p\}=0$. This is a contradiction and therefore $\{J, J\}= 0$.

As a consequence, $J\wedge\{p, J\}=0$. Set $\alpha=\{p, J\}\in
\Alt^1(\gO)$ then we have $J\wedge\alpha=0$. If $\alpha\neq 0$ then $J$
must have the form $J = \alpha \otimes q$ where $q\in \sym^1(\gI)$. That is a
contradiction since $J$ is indecomposable. \end{proof}

\begin{defn}\label{3.6.2}We continue to keep the condition $I =   J\wedge p$
with $p\in \g_{\un}^*$ nonzero and $J\in   \Alt^1(\gO)\otimes
\sym^1(\gI)$ indecomposable. We can   assume that \[ J =
\sum_{i=1}^n\alpha_i\otimes p_i \] where $\alpha_i\in \Alt^1(\gO), \ i =
1, \dots,n$ are linearly independent and $p_i\in\sym^1(\gI)$. A quadratic Lie
superalgebra having such associated invariant $I$ is called a {\em quasi-
singular quadratic Lie superalgebra}.  
\end{defn}

Let $U = \spa\{\alpha_1,\dots, \alpha_n\}$ and $V = \spa\{p_1,\dots,p_n\}$, one has $\dim(U)$ and $\dim(V)$ more than $1$ by if there is a contrary then $J$ is decomposable. Using Definition \ref{3.1.4}, we have:
\[
\{J, J\} = \left\{\sum_{i=1}^n\alpha_i\otimes p_i, \sum_{i=1}^n\alpha_i\otimes p_i\right\} = -\sum_{i,j=1}^n\left( \{\alpha_i, \alpha_j\}\otimes p_ip_j + (\alpha_i\wedge\alpha_j)\otimes\{p_i, p_j\}\right).
\]

Since $\{J, J\}=0$ and $\alpha_i, \ i = 1, \dots,n$ are linearly independent then $\{p_i, p_j\}=0$, for all $i,j$. It implies that $\{p_i, J\} = 0$, for all $i$.

Moreover, since $\{p, J\}=0$ we obtain $\{p, p_i\}=0$, consequently
$\{p_i, I\}=0$, for all $i$ and $\{p, I\}=0$. By Corollary
\ref{3.1.13} (2) and Lemma \ref{3.1.20} we conclude that
$\phi^{-1}(V+\CC p)$ is a subspace of $\Zs(\g)$ and totally isotropic.

Now, let $\{q_1,\dots,q_m\}$ be a basis of $V$ then $J$ can be rewritten by
\[J = \sum_{j=1}^m\beta_j\otimes q_j\] 
where $\beta_j\in U$, for all $j$. One has:
\[
\{J, J\} = \left\{\sum_{j=1}^m\beta_j\otimes q_j, \sum_{j=1}^m\beta_j\otimes q_j\right\} = -\sum_{i,j=1}^m\left( \{\beta_i, \beta_j\}\otimes q_iq_j + (\beta_i\wedge\beta_j)\otimes\{q_i, q_j\}\right).
\]

By the linear independence of the system $\{q_iq_j\}$, we obtain $\{\beta_i, \beta_j\}=0$, for all $i,j$. It implies that $\{\beta_j, I\}=0$, equivalently $\phi^{-1}(\beta_j)\in \Zs(\g)$, for all $j$. Therefore, we always can begin with $J = \sum_{i=1}^n\limits\alpha_i\otimes p_i$ satisfying the following conditions:
\begin{itemize}
	\item[(i)] $\alpha_i, \ i = 1, \dots,n$ are linearly independent,
	\item[(ii)] $\phi^{-1}(U)$ and $\phi^{-1}(V+\CC p)$ are totally isotropic subspaces of $\Zs(\g)$ where $U = \spa\{\alpha_1,\dots, \alpha_n\}$ and $V = \spa\{p_1,\dots,p_n\}$.
\end{itemize}
\bigskip

Let $X_\ze^i = \phi^{-1}(\alpha_i)$, $X_\un^i = \phi^{-1}(p_i)$, for all $i$ and $C:\g \rightarrow \g$ defined by \[J(X,Y) = B(C(X),Y),\ \forall\ X,Y\in\g.\]
\begin{lem}\label{3.6.3}
The mapping $C$ is a skew-supersymmetric homogeneous endomorphism of odd degree and $\im(C)\subset \Zs(\g)$. Recall that if $C$ is a homogeneous endomorphism of degree $c$ of $\g$ satisfying 
\[B(C(X),Y) = -(-1)^{cx}B(X,C(Y)),\ \forall\ X\in\g_x,\ Y\in\g\] then we say $C$ 
{\em skew-supersymmetric} (with respect to $B$).
\end{lem}
\begin{proof}
Since $J(\g_\ze,\g_\ze) = J(\g_\un,\g_\un)=0$ and $B$ is even then $C(\g_\ze)\subset \g_\un$ and $C(\g_\un)\subset \g_\ze$. That means $C$ is of odd degree. For all $X\in\g_\ze, \ Y\in\g_\un$ one has:
\[B(C(X),Y) = J(X,Y) = \sum_{i=1}^n\alpha_i\otimes p_i(X,Y)
 = \sum_{i=1}^n\alpha_i(X)p_i(Y) =  \sum_{i=1}^nB(X_\ze^i,X)B(X_\un^i,Y).
\]
By the non-degeneracy of $B$ and $J(X,Y)=-J(Y,X)$, we obtain:
\[C(X)= \sum_{i=1}^nB(X_\ze^i,X)X_\un^i \ \ \text{ and } \ \ C(Y)= -\sum_{i=1}^nB(X_\un^i,Y)X_\ze^i.
\]
Combined with $B$ supersymmetric, one has:
\[-B(Y,C(X)) = B(C(X),Y)=-B(C(Y),X)=-B(X,C(Y)).\]
It shows that $C$ is skew-supersymmetric. Finally, $\im(C)\subset \Zs(\g)$ since $X_\ze^i,\ X_\un^i\in\Zs(\g)$, for all $i$.
\end{proof}

\begin{prop} \label{3.6.4}Let $X_\un=\phi^{-1}(p)$ then for all $X\in\g_\ze, \ Y,Z\in\g_\un$ one has:
\begin{enumerate}
	\item $[X,Y] = -B(C(X),Y)X_\un - B(X_\un,Y)C(X)$,
	\item $[Y,Z] = B(X_\un,Y)C(Z) + B(X_\un,Z)C(Y)$,
	\item $X_\un\in\Zs(\g)$ and $C(X_\un) = 0$.
\end{enumerate}
\end{prop}
\begin{proof}
Let $X\in\g_\ze, \ Y,Z\in\g_\un$ then
\begin{eqnarray*} B([X,Y],Z) = J\wedge p(X,Y,Z) = -J(X,Y)p(Z)-J(X,Z)p(Y) \\
  = -B(C(X),Y)B(X_\un,Z) - B(C(X),Z)B(X_\un,Y).
\end{eqnarray*}
By the non-degeneracy of $B$ on $\g_\un\times\g_\un$, it shows that:
\[ [X,Y] = -B(C(X),Y)X_\un - B(X_\un,Y)C(X).
\]
Combined with $B$ invariant and $C$ skew-supersymmetric, one has:
\[ [Y,Z] = B(X_\un,Y)C(Z) + B(X_\un,Z)C(Y).
\]
Since $\{p,I\}=0$ then $X_\un\in\Zs(\g)$. Moreover, $\{p,p_i\}=0$ imply $B(X_\un,X_\un^i) = 0$, for all $i$. It means $B(X_\un,\im(C)) = 0$. And since $B(C(X_\un),X) = B(X_\un,C(X)) = 0$, for all $X\in\g$ then $C(X_\un) = 0$.
\end{proof}

Let $W$ be a complementary subspace of $\spa\{X_\un^1,\dots,X_\un^n, X_\un\}$ in $\g_\un$ and $Y_\un$ be an element in $W$ such that $B(X_\un,Y_\un)=1$. Let $X_\ze = C(Y_\un)$, $\qk=(\CC X_1\oplus \CC Y_\un)^\bot$ and $B_\qk = B|_{\qk\times\qk}$ then we have the following corollary:
\begin{cor}\label{3.6.5}\hfill
\begin{enumerate}
	\item $[Y_\un,Y_\un] = 2X_\ze$, $[Y_\un,X] = C(X)-B(X,X_\ze)X_\un$ and $[X,Y] = -B(C(X),Y)X_\un$, for all $X,Y\in\qk\oplus\CC X_\un$.
	\item $[\g,\g]\subset \im(C)+\CC X_\un \subset \Zs(\g)$ so $\g$ is 2-step nilpotent. If $\g$ is reduced then $[\g,\g]= \im(C)+\CC X_\un = \Zs(\g)$.
	\item $C^2 = 0$.
\end{enumerate}
\end{cor}
\begin{proof}\hfill
\begin{enumerate}
	\item The assertion (1) is obvious by Proposition \ref{3.6.4}.
	\item Note that $X_\ze\in\im(C)$ so $[\g,\g]\subset \im(C)+\CC X_\un$. By Lemma \ref{3.6.3} and Proposition \ref{3.6.4}, $\im(C)+\CC X_\un\subset\Zs(\g)$. If $\g$ is reduced then $\Zs(\g)\subset[\g,\g]$ and therefore $[\g,\g]= \im(C)+\CC X_\un = \Zs(\g)$.
		\item Since $\g$ is 2-step nilpotent then
		\[0=[Y_\un,[Y_\un,Y_\un]] = [Y_\un,2X_\ze] = 2C(X_\ze)-2B(X_\ze,X_\ze)X_\un.\]
	Since $X_\ze = C(Y_\un)$ and $\im(C)$ is totally isotropic then $B(X_\ze,X_\ze)=0$ and therefore $C(X_\ze) = C^2(Y_\un) = 0$.
		
		If $X\in \qk\oplus\CC X_\un$ then $0 = [Y_\un,[Y_\un,X]] = [Y_\un,C(X)]$. By the choice of $Y_\un$, it is sure that $C(X)\in \qk\oplus\CC X_\un$. Therefore, one has:
		\[0 = [Y_\un,C(X)] = C^2(X) -B(C(X),X_\ze)X_\un = C^2(X) -B(C(X),C(Y_\un))X_\un.\]
		By $\im(C)$ totally isotropic, one has $C^2(X)=0$.
\end{enumerate}
\end{proof}

Now, we consider a special case: $X_\ze = 0$. As a consequence, $[Y_\un,Y_\un] = 0$, $[Y_\un,X] = C(X)$ and $[X,Y] = -B(C(X),Y)X_\un$, for all $X,Y\in\qk$. Let $X\in\qk$ and assume that $C(X) = C_1(X) + aX_1$ where $C_1(X)\in \qk$ then
\[0=B([Y_\un,Y_\un],X) = B(Y_\un,[Y_\un,X]) = B(Y_\un,C_1(X) + aX_1) = a.\]
It shows that $C(X)\in \qk$, for all $X\in\qk$ and therefore we have an affirmative answer of the above question as follows:

\begin{prop}\label{3.6.6}
Let $(\qk = \qk_\ze\oplus\qk_\un,B_\qk)$ be a quadratic $\ZZ_2$-graded vector space and $\cb$ be an odd endomorphism of $\qk$ such that $\cb$ is skew-supersymmetric and $\cb^2 =0$. Let $(\tk = \spa\{X_\un,Y_\un\},B_\tk)$ be a 2-dimensional symplectic vector space with $B_\tk(X_\un,Y_\un) = 1$. Consider the space $\g = \qk\oplusp\tk$ and define the product on $\g$ by:
\[[Y_\un,Y_\un] = [X_\un,\g] = 0, \ [Y_\un,X] = \cb(X) \ \text{ and } [X,Y] = -B_\qk(\cb(X),Y)X_\un\]
for all $X\in\qk$. Then $\g$ becomes a 2-nilpotent quadratic Lie superalgebra with the bilinear form $B = B_\qk+B_\tk$. Moreover, one has $\g_\ze = \qk_\ze$, $\g_\un = \qk_\un\oplus\tk$.
\end{prop}

\begin{rem}
The method above remains valid for the elementary quadratic Lie superalgebra $\g_6^s$ with $I$ decomposable (see Section 3) as follows: let $\qk= (\CC X_{\ze} \oplus \CC Y_{\ze})\oplus (\CC Z_{\un}\oplus \CC T_{\un})$ where $\qk_\ze=\spa\{X_\ze,Y_\ze\}$, $\qk_\un=\spa\{Z_\un,T_\un\}$ and the bilinear form $B_\qk$ is defined by $B(X_{\ze}, Y_{\ze})=B(Z_{\un}, T_{\un})=1$, the other are zero. Let $C: \qk\ \rightarrow\ \qk$  be a linear map defined by:
\[\cb =  \begin{pmatrix} 0 & 0 & 0 & -1 \\ 0 & 0 & 0 & 0 \\ 0 & 1 & 0 & 0 \\ 0 & 0 & 0 & 0 \end{pmatrix}.\]
Then $C$ is odd and $C^2=0$. Set the vector space $\g = \qk\oplusp\tk$, where $(\tk = \spa\{X_\un,Y_\un\},B_\tk)$ is a 2-dimensional symplectic vector space with $B_\tk(X_\un,Y_\un) = 1$. Then $\g=\g_6^s$ with the Lie super-bracket defined as in Proposition \ref{3.6.6}.
\end{rem}

It remains to consider $X_\ze\neq 0$. The fact is that $C$ may be not
stable on $\qk$, that is, $C(X)\in \qk\oplus\CC X_\un$ if $X\in\qk$
but that we need here is an action stable on $\qk$. Therefore, we
decompose $C$ by $C(X) = \cb(X)+\varphi(X)X_\un$, for all $X\in\qk$
where $\cb:\qk\rightarrow \qk$ and $\varphi:\qk\rightarrow\CC$. Since
$B(C(Y_\un),X) = B(Y_\un,C(X)$ then $\varphi(X) = -B(X_\ze,X) =
-B(X,X_\ze)$, for all $ X\in\qk$. Moreover, $C$ is odd degree on $\g$
and skew-supersymmetric (with respect to $B$) implies that $\cb$ is
also odd on $\qk$ and skew-supersymmetric (with respect to
$B_\qk$). It is easy to see that $\cb^2 = 0$, $\cb(X_\ze) =0$ and we
have the following result:
\begin{cor}\label{3.6.7}
Keep the notations as in Corollary \ref{3.6.5} and replace $2X_\ze$ by $X_\ze$ then for all $X,Y\in\qk$, one has:
\begin{itemize}
	\item $[Y_\un,Y_\un] = X_\ze,$
	\item $[Y_\un,X] = \cb(X)-B(X,X_\ze)X_\un,$
	\item $[X,Y] = -B(\cb(X),Y)X_\un$. 
\end{itemize}
\end{cor}

Hence, we have a more general result of Proposition \ref{3.6.6}:
\begin{prop}\label{3.6.8}
Let $(\qk = \qk_\ze\oplus\qk_\un,B_\qk)$ be a quadratic $\ZZ_2$-graded vector space and $\cb$ an odd endomorphism of $\qk$ such that $\cb$ is skew-supersymmetric and $\cb^2 =0$. Let $X_\ze$ be an isotropic element of $\qk_\ze$, $X_\ze\in\ker(\cb)$ and $(\tk = \spa\{X_\un,Y_\un\},B_\tk)$ be a 2-dimensional symplectic vector space with $B_\tk(X_\un,Y_\un) = 1$. Consider the space $\g = \qk\oplusp\tk$ and define the product on $\g$ by:
\[[Y_\un,Y_\un] = X_\ze, \ [Y_\un,X] = \cb(X)-B_\qk(X,X_\ze)X_\un \text{ and } [X,Y] = -B_\qk(\cb(X),Y)X_\un\]
for all $X\in\qk$. Then $\g$ becomes a 2-nilpotent quadratic Lie superalgebra with the bilinear form $B = B_\qk+B_\tk$. Moreover, one has $\g_\ze = \qk_\ze$, $\g_\un = \qk_\un\oplus\tk$.
\end{prop}

A quadratic Lie superalgebra obtained in the above proposition is a special case of the generalized double extensions given in \cite{BBB} where the authors consider the generalized double extension of a quadratic $\ZZ_2$-graded vector space (regarded as an Abelian superalgebra) by a one-dimensional Lie superalgebra.

\section{Appendix: Adjoint orbits of $\spk(2n)$ and $\ok(m)$}

This appendix recalls a fundamental and really interesting problem in Lie
theory that is necessary for the paper: the classification of adjoint orbits
of classical Lie algebras $\spk(2n)$ and $\ok(m)$ where $m$, $n\in \NN^*$. A
brief overview can be found in \cite{Hum95} with interesting discussions. Many
results with detailed proofs can be found in \cite{CM93}.

A different point here is to use the Fitting decomposition to review this
problem. In particular, we parametrize the {\em invertible} component in the
Fitting decomposition of a skew-symmetric map and from this, we give an
explicit classification for $\Sp(2n)$-adjoint orbits of $\spk(2n)$ and
$\OO(m)$-adjoint orbits of $\ok(m)$ in the general case. In other words, we
establish a one-to-one correspondence between the set of orbits and some set
of indices. This is an rather obvious and classical result but in our
knowledge there is not a reference for that mentioned before.

Let $V$ be a $m$-dimensional complex vector space endowed with a non-
degenerate bilinear form $B_\epsilon$ where $\epsilon = \pm 1$ such that
$B_\epsilon (X,Y) =  \epsilon B_\epsilon (Y,X)$, for all $X, Y \in V$. If
$\epsilon =  1$ then the form $B_1$ is symmetric and we say $V$ a {\em
quadratic} vector space. If $\epsilon =  -1$ then $m$ must be even and we say
$V$ a {\em symplectic} vector space with symplectic form $B_{-1}$. We denote
by $\Lc(V)$ the {\em algebra of linear operators} of $V$ and by $\GL(V)$ the {\em
group of invertible operators} in $\Lc(V)$. A map $C\in \Lc(V)$ is called {\em
skew-symmetric} (with respect to $B_\epsilon$) if it satisfies the following
condition: \[ B_\epsilon (C(X), Y) = -B_\epsilon (X, C(Y)),\ \forall\ X, Y \in
V. \] We define \[I_\epsilon (V)= \{A\in \GL(V)\mid\ B_\epsilon (A(X),A(Y)) =
B_\epsilon (X,Y),\ \forall\ X, Y \in V\} \] \[ \text{and }\ \g_\epsilon (V) =
\{C\in \Lc(V)\mid\ C \text{ is skew-symmetric} \}. \] Then $I_\epsilon(V)$ is
the {\em isometry group} of the bilinear form $B_\epsilon$ and
$\g_\epsilon(V)$ is its Lie algebra. Denote by $A^* \in \Lc(V)$ the {\em
adjoint map} of an element $A \in \Lc(V)$ with respect to $B_\epsilon$, then
$A\in I_\epsilon(V)$ if and only if $A^{-1} = A^*$ and $C\in \g_\epsilon(V)$
if and only if $C^* = -C$. If $\epsilon = 1$ then $I_\epsilon(V)$ is denoted
by $\OO (V)$ and $\g_\epsilon(V)$ is denoted by $\ok(V)$. If $\epsilon = -1$
then $\Sp(V)$ stands for $I_\epsilon(V)$ and  $\spk(V)$ stands for
$\g_\epsilon(V)$.

Recall that the {\em adjoint action} $\Ad$ of $I_\epsilon(V)$ on
$\g_\epsilon(V)$ is given by  \[\Ad_U(C)= U C U^{-1}, \ \forall\ U \in
I_\epsilon(V),\ C \in \g_\epsilon(V).\] We denote by $\Oc_C =
\Ad_{I_\epsilon(V)}(C)$, the {\em adjoint orbit} of an element $C \in
\g_\epsilon(V)$ by this action.

If $V=\CC^m$, we call $B_\epsilon$ a \textit{canonical bilinear form} of
$\CC^m$. And with respect to $B_\epsilon$, we define a {\em canonical basis}
$\Bc = \{ E_1, \dots, E_m \}$ of $\CC^m$ as follows. If $m$ even, $m = 2n$,
write $\Bc = \{E_1, \dots, E_n, F_1, \dots, F_n\}$, if $m$ is odd, $m = 2n+1$,
write $\Bc = \{E_1, \dots, E_n, G, F_1, \dots, F_n \}$ and one has:

\begin{itemize}

\item if $m = 2n$ then   \[ B_1(E_i, F_j) = B_1(F_j, E_i)=\delta_{ij},\
B_1(E_i, E_j) = B_1(F_i, F_j) = 0,   \]   \[ B_{-1}(E_i, F_j) = -B_{-1}(F_j,
E_i)=\delta_{ij},\  B_{-1}(E_i, E_j) = B_{-1}(F_i, F_j) = 0,\] where $1 \leq
i,j \leq n$.

In the case of $\epsilon = -1$, $\Bc$ is also called a {\em Darboux} basis of $\CC^{2n}$.
\item if $m = 2n+1$ then $\epsilon=1$ and
  \[ \begin{cases} B_1(E_i, F_j) = \delta_{ij},\ B_1(E_i, E_j) = B_1(F_i,
    F_j) = 0, \\ B_1(E_i, G) = B_1(F_j, G) =
    0, \\ B_1(G,G) = 1\end{cases} \]
where $1 \leq i,j \leq n$.
\end{itemize}

Also, in the case $V = \CC^m$, we denote by $\GL(m)$ instead of $\GL(V)$, $\OO(m)$ stands for $\OO(V)$ and $\ok(m)$ stands for $\ok(V)$. If $m=2n$ then $\Sp(2n)$ stands for $\Sp(V)$ and $\spk(2n)$ stands for $\spk(V)$. We will also write $I_\epsilon=I_\epsilon(\CC^m)$ and $\g_\epsilon=\g_\epsilon(\CC^m)$. Our goal is classifying all of $I_\epsilon$-adjoint orbits of $\g_\epsilon$.

Finally, let $V$ is an m-dimensional vector space. If $V$ is quadratic then $V$ is isometrically isomorphic to the quadratic space $(\CC^m,B_1)$ and if $V$ is symplectic then $V$ is isometrically isomorphic to the symplectic space $(\CC^m,B_{-1})$ \cite{Bou59}.

\subsection{Nilpotent orbits}\label{S1.2}\hfill

Let $n\in\NN^*$, a {\em partition} $[d]$ of $n$ is a tuple $[d_1, ..., d_k]$ of positive integers satisfying
\[d_1\geq ...\geq d_k \text{ and } d_1+ ...+ d_k = n.
\]

Occasionally, we use the notation $[t_1^{i_1}, \dots, t_r^{i_r}]$ to replace the partition $[d_1, ..., d_k]$ where
\[ 
d_j = \left\{\begin{matrix} t_1 & 1\leq j\leq i_1 \\ t_2 & i_1 + 1\leq j\leq i_1 + i_2 \\ t_3 & i_1 + i_2 + 1\leq j\leq i_1 + i_2 + i_3 \\ \dots \end{matrix}\right.
\]
Each $i_j$ is called the {\em multiplicity} of $t_j$. Denote by $\Pc(n)$ the set of partitions of $n$.

Let $p \in \NN^*$. We denote the {\em Jordan block
    of size} $p$ by $J_1 = (0)$ and for $p \geq 2$, \[J_p :
  = \begin{pmatrix} 0 & 1 & 0 & \dots & 0 \\ 0 & 0 & 1 & \dots & 0\\
    \vdots & \vdots & \dots & \ddots & \vdots \\ 0 & 0 & \dots & 0 & 1
    \\ 0 & 0 & 0 & \dots & 0 \end{pmatrix}.\]
Then $J_p$ is a nilpotent endomorphism of $\CC^p$. Given a partition $[d] = [d_1, ..., d_k] \in \Pc(n)$ there is a nilpotent endomorphism of $\CC^n$ defined by
\[X_{[d]}: = \diag_k(J_{d_1}, ..., J_{d_k}).
\]
Moreover, $X_{[d]}$ is also a nilpotent element of $\slk(n)$ since its trace is zero. Conversely, if $C$ is a nilpotent element in $\slk(n)$ then $C$ is $\GL(n)$-conjugate to its {\em Jordan normal form} $X_{[d]}$ for some partition $[d] \in \Pc(n)$.

Given two different partitions $[d]=[d_1, ..., d_k]$ and $[d']=[d'_1, ..., d'_l]$ of $n$ then the $\GL(n)$-adjoint orbits through $X_{[d]}$ and $X_{[d']}$ respectively are disjoint by the unicity of Jordan normal form. Therefore, one has the following proposition:
\begin{prop}
There is a one-to-one correspondence between the set of nilpotent $\GL(n)$-adjoint orbits of $\slk(n)$ and the set $\Pc(n)$.
\end{prop}

Define the set 
\[
\Pc_\epsilon(m) = \{[d_1, ..., d_m] \in \Pc(m)|\ \sharp\{j\mid\ d_j = i\} \text{ is even for all i such that } (-1)^i = \epsilon\}.
\]
In particular, $\Pc_1(m)$ is the set of partitions of $m$ in which even parts occur with even multiplicity and $\Pc_{-1}(m)$ is the set of partitions of $m$ in which odd parts occur with even multiplicity.
\begin{prop}[Gerstenhaber]\label{1.2.10}\hfill

Nilpotent $I_\epsilon$-adjoint orbits in $\g_\epsilon$ are in one-to-one correspondence with the set of partitions in $\Pc_\epsilon(m)$. 
\end{prop}

Here, we give a construction of a nilpotent element in $\g_\epsilon$ from a partition $[d]$ of $m$ that is useful for this paper. Define maps in $\g_\epsilon$ as follows:
\begin{itemize}
	\item For $p \geq 2$, we equip the vector space $\CC^{2p}$ with its canonical bilinear form $B_\epsilon$ and the map $C_{2p}^J$ having the matrix
\[C_{2p}^J 
= \begin{pmatrix} J_p & 0 \\ 0 & - {}^t
  J_p \end{pmatrix}\] in a canonical basis where  ${}^t
  J_p$ denotes the {\em transpose} matrix of  the Jordan block $J_p$. Then $C_{2p}^J \in
\g_\epsilon(\CC^{2p})$.
\item For $p \geq 1$ we equip the vector space $\CC^{2p + 1}$ with its canonical bilinear form $B_1$ and the map $C_{2p +1}^J$ having the matrix \[C_{2p+1}^J 
= \begin{pmatrix} J_{p+1} & M \\ 0 & - {}^t
  J_p \end{pmatrix}\] in a canonical basis where $M = (m_{ij})$ denotes the $(p+1)\times p$-matrix with $m_{p+1,p} = -1$ and $m_{ij}=0$ otherwise. Then $C_{2p + 1}^J \in
\ok(2p+1)$
  \item For $p \geq 1$, we consider the vector space $\CC^{2p}$ equipped with its
  canonical bilinear form $B_{-1}$ and the map $C_{p+p}^J$
  with matrix \[ \begin{pmatrix} J_{p} & M \\ 0 & -{}^t
    J_p\end{pmatrix}\] in a canonical basis where $M=(m_{ij})$
  denotes the $p \times p$-matrix with $m_{p,p} = 1$ and
  $m_{ij} = 0$ otherwise. Then $C_{p+p}^J \in \spk(2p)$.
\end{itemize}

For each partition $[d]\in \Pc_{-1}(2n)$, $[d]$ can be written as 
\[(p_1, p_1, p_2, p_2, \dots, p_k, p_k, 2q_1, \dots, 2 q_\ell)\]
with all $p_i$ odd, $p_1 \geq p_2 \geq \dots \geq p_k$ and $q_1 \geq q_2 \geq \dots \geq q_\ell$. We associate $[d]$ to the map $C_{[d]}$ with matrix:
\[ \diag_{k + \ell}(C^J_{2p_1}, C^J_{2p_2}, \dots, C^J_{2p_k},
C^J_{q_1+q_1}, \dots, C^J_{q_\ell +q_\ell} )\] in a canonical basis
of $\CC^{2n}$ then $C_{[d]}\in \spk(2n)$. 

Similarly, let $[d]\in \Pc_{1}(m)$, $[d]$ can be written as 
\[(p_1, p_1, p_2, p_2, \dots, p_k, p_k, 2q_1+1, \dots, 2 q_\ell+1)\] with all $p_i$ even, $p_1 \geq p_2 \geq \dots \geq p_k$ and $q_1 \geq q_2 \geq \dots \geq q_\ell$. We associate $[d]$ to the map $C_{[d]}$ with matrix:
\[ \diag_{k + \ell}(C^J_{2p_1}, C^J_{2p_2}, \dots, C^J_{2p_k},
C^J_{2q_1+1}, \dots, C^J_{2q_\ell +1} ).\] in a canonical basis
of $\CC^{m}$ then $C_{[d]}\in \ok(m)$. 

By Proposition \ref{1.2.10}, it is sure that our construction is a bijection between the set $\Pc_{\epsilon}(m)$ and the set of nilpotent $I_\epsilon$-adjoint orbits in $\g_\epsilon$.

\subsection{Semisimple orbits}\label{S1.3}\hfill

We recall a well-known result \cite{CM93}:
\begin{prop}
Let $\g$ be a semisimple Lie algebra, $\hk$ be a Cartan subalgebra of $\g$ and $W$ be the associated Weyl group. Then there is a bijection between the set of semisimple orbits of $\g$ and $\hk/W$.
\end{prop}

For each $\g_\epsilon$, we choose the Cartan subalgebra $\hk$ given by the vector space of diagonal matrices of type 
\[
\diag_{2n}(\lambda_1, \dots, \lambda_n, -\lambda_1, \dots,
-\lambda_n)\]
if $\g_\epsilon = \ok(2n)$ or $\g_\epsilon = \spk(2n)$ and of type 
\[
\diag_{2n+1}(\lambda_1, \dots, \lambda_n, 0, -\lambda_1, \dots,
-\lambda_n)
\]
if $\g_\epsilon = \ok(2n+1)$.

Any diagonalizable (equivalently semisimple) $C \in \g_\epsilon$ is
conjugate to an element of $\hk$. 

If $\g_\epsilon = \spk(2n)$ then any two eigenvectors $v,w\in \CC^{2n}$ of $C\in\g_\epsilon$ with eigenvalues $\lambda,\lambda'\in\CC$ such that $\lambda + \lambda'\neq 0$ are orthogonal. Moreover, each eigenvalue pair $\lambda,-\lambda$ is corresponding to an eigenvector pair $(v,w)$ satisfying $B_\epsilon(v,w)=1$ and we can easily arrange for vectors $v,v'$ lying in a distinct pair $(v,w),(v',w')$ to be orthogonal, regardless of the eigenvalues involved. That means the associated Weyl group is of all coordinate permutations and sign changes of $(\lambda_1, \dots, \lambda_n)$. We denote it by $G_n$.

If $\g_\epsilon = \ok(2n)$, the associated Weyl group, when considered
in the action of the group $\SO(2n)$, consists all coordinate
permutations and even sign changes of $(\lambda_1, \dots,
\lambda_n)$. However, we only focus on $\OO(2n)$-adjoint orbits of
$\ok(2n)$ obtained by the action of the full orthogonal group, then
similarly to preceding analysis any sign change effects. The
corresponding group is still $G_n$. If $\g_\epsilon = \ok(2n+1)$, the
Weyl group is $G_n$ and there is nothing to add.

Now, let $\Lambda_n = \{ (\lambda_1, \dots, \lambda_n) \mid\ \lambda_1,
\dots, \lambda_n \in \CC,\ \lambda_i \neq 0 \ \text{ for some } \ i \}$.
\begin{cor}\label{5.8}
  There is a bijection between nonzero semisimple $I_\epsilon$-adjoint
  orbits of $\g_\epsilon$ and $\Lambda_n / G_n$.
\end{cor}

\subsection{Invertible orbits}\label{S1.4}\hfill

\begin{defn}
We say that the $I_\epsilon$-orbit $\Oc_X$ is {\em invertible} if $X$ is an invertible element in $\g_\epsilon$.
\end{defn}

Keep the above notations. We say an element $X\in V$ {\em isotropic} if $B_\epsilon(X,X) = 0$ and a subset $W\subset V$ {\em totally isotropic} if $B_\epsilon(X,Y) = 0$ for all $X,Y\in W$. 

We recall the classification method given in \cite{DPU10} as follows. First, we need the lemma:

\begin{lem} \label{1.4.2} Let $V$ be an even-dimensional vector space with a non-degenerate bilinear form $B_\epsilon$. Assume that $V = V_+ \oplus V_-$ where $V_\pm$ are totally isotropic
  vector subspaces.

\begin{enumerate}

\item Let $N \in \Lc(V)$ such that $N(V_\pm) \subset V_\pm$. We define
  maps $N_\pm$ by $N_+|_{V_+} = N|_{V_+}$, $N_+|_{V_-} = 0$,
  $N_-|_{V_-} = N|_{V_-}$ and $N_-|_{V_+} = 0$. Then $N \in \g_\epsilon(V)$
  if and only if $N_- = -N_+^*$ and, in this case, $N = N_+ -
  N_+^*$.

\smallskip

\item Let $U_+ \in \Lc(V)$ such that $U_+$ is invertible, $U_+(V_+) =
  V_+$ and $U_+|_{V_-} = \Id_{V_-}$. We define $U \in \Lc(V)$ by
  $U|_{V_+} = U_+|_{V_+}$ and $U|_{V_-} = \left(U_+^{-1}\right)^*|_{V_-}$. Then $U
  \in I_\epsilon(V)$.

\smallskip

\item Let $N' \in \g_\epsilon(V)$ such that $N'$ satisfies the assumptions of
  (1). Define $N_\pm$ as in (1). Moreover, we assume that there exists
  $U_+ \in \Lc(V_+)$, $U_+$ invertible such that \[ N_+'|_{V_+} =
  \left( U_+ \ N_+ \ U_+^{-1} \right)|_{V_+}.\] We extend $U_+$ to $V$
  by $U_+|_{V_-} = \Id_{V_-}$ and define $U \in I_\epsilon(V)$ as in
  (2). Then \[ N' = U \ N \ U^{-1}.\]

\end{enumerate}

\end{lem}

\begin{proof}\hfill
  
\begin{enumerate}
	\item It is obvious that $N = N_+ + N_-$. Recall that $N \in \g_\epsilon(V)$ if	and only if $N^* = -N$ so $N_+^* + N_-^* = -N_+ - N_-$. Since $B_\epsilon(N_+^*(V_+),V) = B_\epsilon(V_+,N_+(V))=0$ then $N_+^*(V_+)=0$. Similarly, $N_-^*(V_-)=0$. Hence, $N_- = -N_+^*.$
	\item We shows that $B_\epsilon(U(X),U(Y))=B_\epsilon(X,Y)$, for all $X,Y\in V$. Indeed, let $X=X_++X_-,Y=Y_++Y_-\in V_+ \oplus V_-$, one has
	\[B_\epsilon(U(X_++X_-),U(Y_++Y_-)) = B_\epsilon(U_+(X_+) + \left(U_+^{-1}\right)^*(X_-),U_+(Y_+) + \left(U_+^{-1}\right)^*(Y_-))
	\]
	\[=B_\epsilon(U_+(X_+),\left(U_+^{-1}\right)^*(Y_-)) + B_\epsilon(\left(U_+^{-1}\right)^*(X_-),U_+(Y_+))
	\]
	\[=B_\epsilon(X_+,Y_-) + B_\epsilon(X_-,Y_+) = B_\epsilon(X,Y).\]
	\item Since $B_\epsilon(U^{-1}(V_+),V_+) =  B_\epsilon(V_+,U(V_+)) = 0$, one has $U^{-1}(V_+)= V_+$ and $U^{-1}(V_-)= V_-$. Consequently, $(U \ N \ U^{-1})(V_+) \subset V_+$ and  $(U \ N \ U^{-1})(V_-) \subset V_-$. Clearly, $U \ N \ U^{-1}\in \g_\epsilon(V)$. By (1), we only show that
	\[(U \ N \ U^{-1})|_{V_+} = N_+'
	\]
	This is obvious since $U^{-1}|_{V_+} = U_+^{-1}$.
\end{enumerate}
\end{proof}

Let us now consider $C \in \g_\epsilon$, $C$ invertible. Then, $m$ must be even (obviously, it happened if $\epsilon = -1$),
$m = 2n$. Indeed, we assume that $\epsilon = 1$ then the skew-symmetric form $\Delta_C$ on $\CC^m$ defined by $\Delta_C(v_1,v_2) = B_1(v_1,C(v_2))$ is non-degenerate. and the assertion follows. We decompose $C = S + N$ into semisimple
and nilpotent parts, $S$, $N \in \g_\epsilon$ by its Jordan decomposition. It is clear that $S$ is invertible. We have $\lambda \in
\Lambda$ if and only if $-\lambda \in \Lambda$ where
$\Lambda$ is the spectrum of $S$. Also, $m(\lambda) = m(-\lambda)$, for
all $\lambda \in \Lambda$ with the multiplicity $m(\lambda)$. Since $N$
and $S$ commute, we have $N(V_{\pm \lambda}) \subset V_{\pm \lambda}$
where $V_\lambda$ is the eigenspace of $S$ corresponding to $\lambda
\in \Lambda$. Denote by $W(\lambda)$ the direct sum \[W(\lambda) =
V_\lambda \oplus V_{-\lambda}.\]

Define the equivalence relation $\Rs$ on $\Lambda$ by: \[ \lambda \Rs
\mu \ \text{ if and only if } \ \lambda = \pm \mu.\] Then \[\CC^{2n}
= \bigoplus^\bot_{\lambda \in \Lambda / \Rs} W(\lambda),\] and each
$(W(\lambda), B_\lambda)$ is a vector space with the non-degenerate form $B_\lambda$ given by: \[B_\lambda
= B_\epsilon|_{W(\lambda) \times W(\lambda)}.\]

Fix $\lambda \in \Lambda$. We write $W(\lambda) = V_+ \oplus V_-$ with
$V_\pm = V_{\pm \lambda}$. Then, according to the notation in Lemma \ref{1.4.2},
define $N_{\pm \lambda} = N_\pm$. Since $N|_{V_-} = - N_\lambda^*$, it
is easy to verify that the matrices of $N|_{V_+}$ and $N|_{V_-}$ have
the same Jordan form. Let $(d_1(\lambda), \dots,
d_{r_\lambda}(\lambda))$ be the size of the Jordan blocks in the
Jordan decomposition of $N|_{V_+}$. This does not depend on a possible
choice between $N|_{V_+}$ or $N|_{V_-}$ since both maps have the same
Jordan type.

Next, we consider \[\Dc = \bigcup_{r \in \NN^*} \{ (d_1, \dots, d_r)
\in \NN^r \mid\ d_1 \geq d_2 \geq \dots \geq d_r \geq 1 \}. \] Define $d
: \Lambda \to \Dc$ by $d(\lambda) = (d_1(\lambda), \dots,
d_{r_\lambda}(\lambda))$. It is clear that $\Phi \circ d = m$ where
$\Phi : \Dc \to \NN$ is the map defined by $\Phi(d_1, \dots, d_r) =
\sum_{i=1}^r d_i$.

Finally, we can associate to $C \in \g_\epsilon$ a triple $(\Lambda, m, d)$
defined as above.

\begin{defn}\label{1.4.3}

  Let $\Jc_n$ be the set of all triples $(\Lambda, m, d)$ such that:

\begin{enumerate}

\item $\Lambda$ is a subset of $\CC \setminus \{0\}$ with
  $\sharp \Lambda \leq 2n$ and $\lambda \in \Lambda$ if and only if
  $-\lambda \in \Lambda$.

\item $m : \Lambda \to \NN^*$ satisfies $m(\lambda) = m(-\lambda)$,
  for all $\lambda \in \Lambda$ and $\underset{\lambda\in \Lambda}{\sum}
  m(\lambda) = 2n$.

\item $d : \Lambda \to \Dc$ satisfies $d(\lambda) = d(-\lambda)$, for
  all $\lambda \in \Lambda$ and $\Phi \circ d = m$.

\end{enumerate}

\end{defn}

Let $\Ic(2n)$ be the set of invertible elements in $\g_\epsilon$ and
$\Ich(2n)$ be the set of $I_\epsilon$-adjoint orbits of elements in
$\Ic(2n)$. By the preceding remarks, there is a map $i : \Ic(2n) \to
\Jc_n$. Then we have a parametrization of the set $\Ich(2n)$ as follows:

\begin{prop} \hfill \label{1.4.4} 

  The map $i : \Ic (2n) \to \Jc_n$ induces a bijection $\tilde{i} :
  \Ich(2n) \to \Jc_n$.
\end{prop}

\begin{proof}
  Let $C$ and $C' \in \Ic(2n)$ such that $C' = U \ C \ U^{-1}$ with $U
  \in I_\epsilon$. Let $S$, $S'$, $N$, $N'$ be respectively the
  semisimple and nilpotent parts of $C$ and $C'$. Write $i(C) =
  (\Lambda, m, d)$ and $i(C') = (\Lambda', m', d')$. One has
  \[S'+N' = U \ (S+N) \ U^{-1} = U \ S \ U^{-1} + U \ N \ U^{-1}.
  \]
  
  By the unicity of Jordan decomposition, $S' = U \ S \ U^{-1}$ and $N' = U \ N \ U^{-1}$. So $\Lambda' = \Lambda$ and $m' =
  m$. Also, since $U\ S = S'\ U$ one has $U\ S(V_\lambda)= S'\ U (V_\lambda)$. It implies that
  \[S'\ (U (V_\lambda)) = \lambda U (V_\lambda).
  \]That means $U(V_\lambda) = V'_\lambda$, for all $\lambda \in
  \Lambda$. Since $N' = U \ N \ U^{-1}$ then
  $N|_{V_\lambda}$ and $N'|_{V'_\lambda}$ have the same Jordan
  decomposition, so $d = d'$ and $\tilde{i}$ is well defined.

  To prove that $\tilde{i}$ is onto, we start with $\Lambda = \{
  \lambda_1, -\lambda_1, \dots, \lambda_k, -\lambda_k \}$, $m$ and $d$
  as in Definition \ref{1.4.3}. Define on the canonical basis: \[ S =
  \diag_{2n} (\overbrace{\lambda_1, \dots, \lambda_1}^{m(\lambda_1)},
  \dots, \overbrace{\lambda_k, \dots, \lambda_k}^{m(\lambda_k)},
  \overbrace{-\lambda_1, \dots, -\lambda_1}^{m(\lambda_1)}, \dots,
  \overbrace{-\lambda_k, \dots, -\lambda_k}^{m(\lambda_k)}). \] For
  all $1 \leq i \leq k$, let $d(\lambda_i) = (d_1(\lambda_i) \geq
  \dots \geq d_{r_{\lambda_i}}(\lambda_i) \geq 1)$ and define \[
  N_+(\lambda_i) = \diag_{d(\lambda_i)} \left( J_{d_1(\lambda_i)},
    J_{d_2(\lambda_i)}, \dots, J_{d_{r_{\lambda_i}}(\lambda_i)}
  \right) \] on the eigenspace $V_{\lambda_i}$ and $0$ on the
  eigenspace $V_{-\lambda_i}$ where $J_d$ is the Jordan block of size
  $d$.

  By Lemma \ref{1.4.2}, $N(\lambda_i) = N_+(\lambda_i) -
  N_+^*(\lambda_i)$ is skew-symmetric on $V_{\lambda_i} \oplus
  V_{-\lambda_i}$. Finally, \[ \CC^{2n}= \bigoplus^\bot_{1\leq i\leq k}\left(
    V_{\lambda_i} \oplus V_{-\lambda_i} \right).\] Define $N \in
  \g_\epsilon$ by $N\left(\sum_{i=1}^k v_i \right) = \sum_{i=1}^k
  N(\lambda_i) (v_i)$, $v_i \in V_{\lambda_i} \oplus V_{-\lambda_i}$
  and $C = S + N \in \g_\epsilon$. By construction, $i(C) = (\Lambda,
  m,d)$, so $\tilde{i}$ is onto.

  To prove that $\tilde{i}$ is one-to-one, assume that $C$, $C' \in
  \Ic(2n)$ and that $i(C) = i(C') = (\Lambda, m, d)$. Using the
  previous notation, since their respective semisimple parts $S$ and
  $S'$ have the same spectrum and same multiplicities, there exist $U
  \in I_\epsilon$ such that $S' = U S U^{-1}$. For $\lambda \in \Lambda$,
  we have $U(V_\lambda) = V'_\lambda$ for eigenspaces $V_\lambda$ and
  $V'_\lambda$ of $S$ and $S'$ respectively.

  Also, for $\lambda \in \Lambda$, if $N$ and $N'$ are the nilpotent
  parts of $C$ and $C'$, then $N''(V_\lambda) \subset V_\lambda$, with
  $N'' = U^{-1} N' U$. Since $i(C) = i(C')$, then $N|_{V_\lambda}$ and
  $N'|_{V'_\lambda}$ have the same Jordan type. Since $ N'' = U^{-1}
  N' U$, then $N''|_{V_\lambda}$ and $N'|_{V'_\lambda}$ have the same
  Jordan type. So $N|_{V_\lambda}$ and $N''|_{V_\lambda}$ have the
  same Jordan type. Therefore, there exists $D_+ \in \Lc(V_\lambda)$
  such that $N''|_{V_\lambda} = D_+ N|_{V_\lambda} D_+^{-1}$. By Lemma
  \ref{1.4.2}, there exists $D(\lambda) \in I_\epsilon(V_\lambda
  \oplus V_{-\lambda})$ such that \[ N''|_{V_\lambda \oplus
    V_{-\lambda}} = D(\lambda) N|_{V_\lambda \oplus V_{-\lambda}}
  D(\lambda)^{-1}.\] We define $D \in I_\epsilon$ by $D|_{V_\lambda
    \oplus V_{-\lambda}} = D(\lambda)$, for all $\lambda \in
  \Lambda$. Then $N'' = D N D^{-1}$ and $D$ commutes with $S$ since
  $S|_{V_{\pm\lambda}}$ is scalar. Then $S' = (UD) S (UD)^{-1}$ and
  $N' = (UD) N (UD)^{-1}$ and we conclude that $ C' = (UD) C
  (UD)^{-1}$.

\end{proof}

\subsection{Adjoint orbits in the general case}\label{S1.5}\hfill

Let us now classify $I_\epsilon$-adjoint orbits of $\g_\epsilon$ in the general case as follows. Let $C$ be an element in $\g_\epsilon$ and consider the Fitting decomposition of $C$
\[
\CC^m = V_N\oplus V_I,
\]
where $V_N$ and $V_I$ are stable by $C$, $C_N=C|_{V_N}$ is nilpotent and $C_I=C|_{V_I}$ is invertible. Since $C$ is skew-symmetric, $B_\epsilon(C^k(V_N),V_I) = (-1)^kB_\epsilon(V_N,C^k(V_I))$ for any $k$ then one has $V_I =
(V_N)^\perp$. Also, the restrictions $B_\epsilon^N =
B_\epsilon|_{V_N \times V_N}$ and $B_\epsilon^I = B_\epsilon|_{V_I \times V_I}$ are non-degenerate. Clearly, $C_N\in \g_\epsilon(V_N)$ and $C_I\in \g_\epsilon(V_I)$. By Subsection \ref{S1.2} and Subsection \ref{S1.4}, $C_N$ is attached with a partition $[d]\in \Pc_\epsilon(n)$ and $C_I$ corresponds to a triple $T\in\Jc_\ell$ where $n = \dim(V_N)$, $2\ell =\dim(V_I)$. Let $\Dc(m)$ be the set of all pairs $([d],T)$ such that $[d]\in \Pc_\epsilon(n)$ and $T\in\Jc_\ell$ satisfying $n+2\ell = m$. By the preceding remarks, there exists a map $p: \g_\epsilon \rightarrow \Dc(m)$ . Denote by $\Oc(\g_\epsilon)$ the set of $I_\epsilon$-adjoint orbits of $\g_\epsilon$ then we obtain the classification of $\Oc(\g_\epsilon)$ as follows:

\begin{prop}\label{1.5.1}
The map $p: \g_\epsilon \rightarrow \Dc(m)$ induces a bijection $\widetilde{p}: \Oc(\g_\epsilon) \rightarrow \Dc(m)$.
\end{prop}
\begin{proof}
Let $C$ and $C'$ be two elements in $\g_\epsilon$. Assume that $C$ and $C'$ lie in the same $I_\epsilon$-adjoint orbit. It means that there exists an isometry $P$ such that $C' = PCP^{-1}$. So $C'^k\ P = P\ C^k$ for any $k$ in $\NN$. As a consequence, $P(V_N)\subset V'_N$ and $P(V_I)\subset V'_I$. However, $P$ is an isometry then $V'_N = P(V_N)$ and  $V'_I = P(V_I)$. Therefore, one has
\[C'_N = P_N\ C_N P^{-1}_N \text{ and } C'_I = P_I\ C_I P_I^{-1},
\]
where $P_N = P:V_N\rightarrow V'_N$ and $P_I = P:V_I\rightarrow V'_I$ are isometries. It implies that $C_N$, $C'_N$ have the same partition and $C_I$, $C'_I$ have the same triple. Hence, the map $\widetilde{p}$ is well defined.

For a pair $([d],T)\in \Dc(m)$ with $[d]\in \Pc_\epsilon(n)$ and $T\in\Jc_\ell$, we set a nilpotent map $C_N\in\g_\epsilon(V_N)$ corresponding to $[d]$ as in Section \ref{S1.2} and an invertible map $C_I\in\g_\epsilon(V_I)$ as in Proposition \ref{1.4.4} where $\dim(V_N) = n$ and $\dim(V_I) = 2\ell$. Define $C\in\g_\epsilon$ by $C(X_N+X_I) = C_N(X_N) + C_I(X_I)$, for all $X_N\in V_N,\ X_I\in V_I$. By construction, $p(C) = ([d],T)$ and $\widetilde{p}$ is onto.

To prove $\widetilde{p}$ is one-to-one, let $C, C'\in \g_\epsilon$ such that $p(C) = p(C') = ([d],T)$. Keep the above notations, since $C_N$ and $C'_N$ have the same partition then there exists an isometry $P_N: V_N\rightarrow V'_N$ such that $C'_N = P_N\ C_N\ P_N^{-1}$. Similarly $C_I$ and $C'_I$ have the same triple and then there exists an isometry $P_I: V_I\rightarrow V'_I$ such that $C'_I = P_I\ C_I\ P_I^{-1}$. Define $P: V\rightarrow V$ by $P(X_N + X_I) = P_N(X_N) + P_I(X_I)$, for all $X_N\in V_N,X_I\in V_I$ then $P$ is an isometry and $C' = P\ C\ P^{-1}$. Therefore, $\widetilde{p}$ is one-to-one.

\end{proof}

\bibliographystyle{alpha}

\begin{bibdiv}

\begin{biblist}

\bib{BB97}{article}{
   author={Bajo, I.},
   author={Benayadi, S.},
   title={Lie algebras admitting a unique quadratic structure},
   journal={Communications in Algebra},
   volume={25},
   number={9},
   date={1997},
   pages={2795 -- 2805},

} 
 
\bib{BBB}{article}{
	 author={Bajo, I.},
   author={Benayadi, S.},
   author={ Bordemann, M.},
   title={Generalized double extension and descriptions of quadratic Lie superalgebras},
   journal={arXiv:0712.0228v1},

} 

\bib{BB99}{article}{
	 author={Benamor, H.},
   author={Benayadi, S.},
   title={Double extension of quadratic Lie superalgebras},
   journal={Comm. in Algebra},
   fjournal={Communications in Algebra},
   volume={27},
   number={1},
   date={1999},
   pages={67 -- 88},

} 

\bib{BP89}{article}{
	 author={Benamor, H.},
   author={Pinczon, G.},
   title={The graded Lie algebra structure of Lie superalgebra deformation theory},
   journal={Lett. Math. Phys.},
   fjournal={Letters in Mathematical Physics},
   volume={18},
   number={4},
   date={1989},
   pages={307 -- 313},

} 

\bib{Ben03}{article}{
   author={Benayadi, S.},
   title={Socle and some invariants of quadratic Lie superalgebras},
   journal={J. of Algebra},
	fjournal={Journal of Algebra},
   volume={261},
   number={2},
   date={2003},
   pages={245 -- 291},

} 

\bib{Bor97}{article}{
	 author={Bordemann, M.},
	 title={Nondegenerate invariant bilinear forms on nonassociative algebras},
   journal={Acta Math. Univ. Comenianae},
   fjournal={Acta Mathematica Universitatis Comenianae},
   volume={LXVI},
   number={2},
   date={1997},
   pages={151 -- 201},
}


\bib{Bou58}{book}{
   author={Bourbaki, N.},
   title={\'El\'ements de Math\'ematiques. Alg\`ebre, Alg\`ebre Multilin\'eaire},
   volume={Fasc. VII, Livre II},
   publisher={Hermann},
   place={Paris},
   date={1958},
   pages={},
}


\bib{Bou59}{book}{
   author={Bourbaki, N.},
   title={\'El\'ements de Math\'ematiques. Alg\`ebre, Formes sesquilin\'eaires et formes quadratiques},
   volume={Fasc. XXIV, Livre II},
   publisher={Hermann},
   place={Paris},
   date={1959},
   pages={},
}


\bib{Bou71}{book}{
   author={Bourbaki, N.},
   title={El\'ements de Math\'ematiques. Groupes et Alg\`ebres de Lie},
   volume={Chapitre I, Alg\`ebres de Lie},
   publisher={Hermann},
   place={Paris},
   date={1971},
   pages={},
}


 
\bib{CM93}{book}{
   author={Collingwood, D. H.},
   author={McGovern, W. M.},
   title={Nilpotent Orbits in Semisimple Lie Algebras},
   publisher={Van Nostrand Reihnhold Mathematics Series},
   place={New York},
   date={1993},
   pages={186},
}

\bib{DPU10}{article}{
   author={Duong, M.T.},
   author={Pinczon, G.},
   author={Ushirobira, R.},
   title={A new invariant of quadratic Lie algebras},
   journal={Alg. Rep. Theory, DOI: 10.1007/s10468-011-9284-4},
	fjournal={Journal of Algebras and Representation Theory},
   volume={},
	 pages={41 pages},
} 
\bib{FS87}{article}{
   author={Favre, G.},
   author={Santharoubane, L.J.},
   title={Symmetric, invariant, non-degenerate bilinear form on a Lie algebra},
   journal={J. Algebra},
   fjournal={Journal of Algebra},
   volume={105},
   date={1987},
   pages={451--464},

} 

\bib{Gie04}{book}{
   author={Gi\'e, P. A.},
   title={Nouvelles structures de Nambu et super-th\'eor\`eme d'Amitsur-Levizki},
   publisher={Th\`ese de l'Universit\'e de Bourgogne, tel-00008876},
   date={2004},
   pages={153},
	 volume={}
} 

\bib{Hum95}{book}{
   author={Humphreys, J.},
   title={Conjugacy classes in semisimple algebraic groups},
   publisher={American Mathematical Society},
   series={Mathematical Surveys and Monographs},
   volume={43},
   date={1995},
   pages={xviii + 196 pp},

}

\bib{Kac85}{book}{
   author={Kac, V.},
   title={Infinite-dimensional Lie algebras},
   publisher={Cambrigde University Press},
   place={New York},
   date={1985},
   pages={xvii + 280 pp},

}

\bib{Kos50}{article}{
  author={Koszul, J-L.}, 
  title={Homologie et cohomologie des alg\`ebres de Lie}, 
  journal={Bulletin de la S. M. F}, 
  fjournal={Bulletin de la Soci\'et\'e Math\'ematique de France}, 
  volume={78}, 
  date={1950},
  pages={65--127},
  }

\bib{MR85}{article}{ 
  author={Medina, A.}, 
  author={Revoy, P.},
  title={Alg\`ebres de Lie et produit scalaire invariant}, 
  journal={Ann. Sci. \'Ecole Norm. Sup.}, 
  fjournal={Annales Scientifiques de l'\'Ecole Normale Sup\'erieure}, 
  volume={4}, 
  date={1985},
  pages={553--561},

}

\bib{MPU09}{article}{
  author={Musson, I. A.}, 
  author={Pinczon, G.},
  author={Ushirobira, R.},
  title={Hochschild Cohomology and Deformations of Clifford-Weyl Algebras}, 
  journal={SIGMA},
  fjournal={Symmetry, Integrability and Geometry: Methods and Applications},
  volume={5}, 
  date={2009},
  pages={27 pp}
  }

\bib{NR66}{article}{ 
  author={Nijenhuis, A.}, 
  author={Richardson, R. W.},
  title={Cohomology and deformations in graded Lie algebras}, 
  journal={Bull. Amer. Math. Soc.}, 
  fjournal={Bulletin of the American Mathematical Society},
  volume={72}, 
  date={1966},
  pages={1--29},

} 

\bib{PU07}{article}{
   author={Pinczon, Georges},
   author={Ushirobira, Rosane},
   title={New Applications of Graded Lie Algebras to Lie Algebras, Generalized Lie Algebras, and Cohomology},
   journal={J. Lie Theory},
   fjournal={Journal of Lie Theory},
   volume={17},
   date={2007},
   number={3},
   pages={633 -- 668},

} 

\bib{Sam80}{book}{
   author={Samelson, H.},
   title={Notes on Lie algebras},
   series={Universitext},
   publisher={Springer-Verlag},
   place={},
   date={1980},
   pages={},

}

\bib{Sch79}{book}{
   author={Scheunert, M.},
   title={The Theory of Lie Superalgebras},
   series={Lecture Notes in Mathematics},
   volume={716},
   publisher={Springer-Verlag},
   place={Berlin},
   date={1979},
   pages={x + 271pp},

}

\end{biblist}
\end{bibdiv}


\end{document}